\documentclass[10pt]{article}
\usepackage{amsmath,amsthm,latexsym,amssymb,amsfonts,epsfig, psfrag,color}
\usepackage{graphicx}
\usepackage{caption}
\usepackage{subcaption}
\usepackage{sidecap}
\usepackage{wrapfig}
\usepackage{multirow}
\usepackage{url}
\usepackage[small,nohug,heads=vee]{diagrams}
\diagramstyle[labelstyle=\scriptstyle]

\makeatletter \@addtoreset{equation}{section} \makeatother

\newtheorem{lem}{Lemma}[section]

\newtheorem{Example}{Example}[section]
\newtheorem*{R}{Remark}

\def\be{\begin{equation}}
\def\ee{\end{equation}}
\def\ba{\begin{eqnarray}}
\def\ea{\end{eqnarray}}

\newcommand\nn{\nonumber}
\newcommand\q{\quad}
\def\Nl{{\mathchoice
{\setbox0=\hbox{$\displaystyle\rm N$}\hbox{\hbox to0pt
{\kern0.4\wd0\vrule height0.9\ht0\hss}\box0}}
{\setbox0=\hbox{$\textstyle\rm N$}\hbox{\hbox to0pt
{\kern0.4\wd0\vrule height0.9\ht0\hss}\box0}}
{\setbox0=\hbox{$\scriptstyle\rm N$}\hbox{\hbox to0pt
{\kern0.4\wd0\vrule height0.9\ht0\hss}\box0}}
{\setbox0=\hbox{$\scriptscriptstyle\rm N$}\hbox{\hbox to0pt
{\kern0.4\wd0\vrule height0.9\ht0\hss}\box0}}}}
\def\Zl{{\mathchoice
{\setbox0=\hbox{$\displaystyle\rm Z$}\hbox{\hbox to0pt
{\kern0.4\wd0\vrule height0.9\ht0\hss}\box0}}
{\setbox0=\hbox{$\textstyle\rm Z$}\hbox{\hbox to0pt
{\kern0.4\wd0\vrule height0.9\ht0\hss}\box0}}
{\setbox0=\hbox{$\scriptstyle\rm Z$}\hbox{\hbox to0pt
{\kern0.4\wd0\vrule height0.9\ht0\hss}\box0}}
{\setbox0=\hbox{$\scriptscriptstyle\rm Z$}\hbox{\hbox to0pt
{\kern0.4\wd0\vrule height0.9\ht0\hss}\box0}}}}
\def\Ql{{\mathchoice
{\setbox0=\hbox{$\displaystyle\rm Q$}\hbox{\hbox to0pt
{\kern0.4\wd0\vrule height0.9\ht0\hss}\box0}}
{\setbox0=\hbox{$\textstyle\rm Q$}\hbox{\hbox to0pt
{\kern0.4\wd0\vrule height0.9\ht0\hss}\box0}}
{\setbox0=\hbox{$\scriptstyle\rm Q$}\hbox{\hbox to0pt
{\kern0.4\wd0\vrule height0.9\ht0\hss}\box0}}
{\setbox0=\hbox{$\scriptscriptstyle\rm Q$}\hbox{\hbox to0pt
{\kern0.4\wd0\vrule height0.9\ht0\hss}\box0}}}}
\def\Rl{{\mathchoice
{\setbox0=\hbox{$\displaystyle\rm R$}\hbox{\hbox to0pt
{\kern0.4\wd0\vrule height0.9\ht0\hss}\box0}}
{\setbox0=\hbox{$\textstyle\rm R$}\hbox{\hbox to0pt
{\kern0.4\wd0\vrule height0.9\ht0\hss}\box0}}
{\setbox0=\hbox{$\scriptstyle\rm R$}\hbox{\hbox to0pt
{\kern0.4\wd0\vrule height0.9\ht0\hss}\box0}}
{\setbox0=\hbox{$\scriptscriptstyle\rm R$}\hbox{\hbox to0pt
{\kern0.4\wd0\vrule height0.9\ht0\hss}\box0}}}}
\def\Cl{{\mathchoice
{\setbox0=\hbox{$\displaystyle\rm C$}\hbox{\hbox to0pt
{\kern0.4\wd0\vrule height0.9\ht0\hss}\box0}}
{\setbox0=\hbox{$\textstyle\rm C$}\hbox{\hbox to0pt
{\kern0.4\wd0\vrule height0.9\ht0\hss}\box0}}
{\setbox0=\hbox{$\scriptstyle\rm C$}\hbox{\hbox to0pt
{\kern0.4\wd0\vrule height0.9\ht0\hss}\box0}}
{\setbox0=\hbox{$\scriptscriptstyle\rm C$}\hbox{\hbox to0pt
{\kern0.4\wd0\vrule height0.9\ht0\hss}\box0}}}}
\def\Hl{{\mathchoice
{\setbox0=\hbox{$\displaystyle\rm H$}\hbox{\hbox to0pt
{\kern0.4\wd0\vrule height0.9\ht0\hss}\box0}}
{\setbox0=\hbox{$\textstyle\rm H$}\hbox{\hbox to0pt
{\kern0.4\wd0\vrule height0.9\ht0\hss}\box0}}
{\setbox0=\hbox{$\scriptstyle\rm H$}\hbox{\hbox to0pt
{\kern0.4\wd0\vrule height0.9\ht0\hss}\box0}}
{\setbox0=\hbox{$\scriptscriptstyle\rm H$}\hbox{\hbox to0pt
{\kern0.4\wd0\vrule height0.9\ht0\hss}\box0}}}}
\def\Ol{{\mathchoice
{\setbox0=\hbox{$\displaystyle\rm O$}\hbox{\hbox to0pt
{\kern0.4\wd0\vrule height0.9\ht0\hss}\box0}}
{\setbox0=\hbox{$\textstyle\rm O$}\hbox{\hbox to0pt
{\kern0.4\wd0\vrule height0.9\ht0\hss}\box0}}
{\setbox0=\hbox{$\scriptstyle\rm O$}\hbox{\hbox to0pt
{\kern0.4\wd0\vrule height0.9\ht0\hss}\box0}}
{\setbox0=\hbox{$\scriptscriptstyle\rm O$}\hbox{\hbox to0pt
{\kern0.4\wd0\vrule height0.9\ht0\hss}\box0}}}}
\newcommand{\cc}{\mathcal C}

\newcommand{\cg}{\mathcal G}
\newcommand{\ch}{\mathcal H}

\newcommand{\cp}{\mathcal P}
\newcommand{\cq}{\mathcal Q}

\newcommand{\fn}{\mathfrak{n}}  
\newcommand{\fo}{\mathfrak{o}}

\def\nn{\nonumber}

\newcommand{\eqa}{\begin{eqnarray}}
\newcommand{\neqa}{\end{eqnarray}}

\newcommand{\p}{\partial}

\def\f{\frac}

\usepackage{bbm}

\def\vphi{\varphi}

\def\q{{\quad}}

\begin{document}

{\renewcommand{\thefootnote}{\fnsymbol{footnote}}

\title{Quantization of systems with temporally varying discretization I:\\
 Evolving Hilbert spaces}
\author{Philipp A H\"ohn\footnote{e-mail address: {\tt phoehn@perimeterinstitute.ca}}\\
 \small Perimeter Institute for Theoretical Physics,\\
 \small 31 Caroline Street North, Waterloo, Ontario, Canada N2L 2Y5
}

\date{}

}

\setcounter{footnote}{0}
\maketitle

{\abstract A temporally varying discretization often features in discrete gravitational systems and appears in lattice field theory models subject to a coarse graining or refining dynamics. To better understand such discretization changing dynamics in the quantum theory, an according formalism for constrained variational discrete systems is constructed. While the present manuscript focuses on global evolution moves and, for simplicity, restricts to flat configuration spaces $\mathbb{R}^N$, a companion article \cite{Hoehn:2014wwa} discusses local evolution moves. In order to link the covariant and canonical picture, the dynamics of the quantum states is generated by propagators which satisfy the canonical constraints and are constructed using the action and group averaging projectors. This projector formalism offers a systematic method for tracing and regularizing divergences in the resulting state sums. Non-trivial coarse graining evolution moves lead to non-unitary, and thus irreversible, projections of physical Hilbert spaces and Dirac observables such that these concepts become evolution move dependent on temporally varying discretizations. The formalism is illustrated in a toy model mimicking a `creation from nothing'. Subtleties arising when applying such a formalism to quantum gravity models are discussed. 

}

\section{Introduction}

In order to promote a classical theory to a quantum theory by means of a path integral, it is common practice to regularize the continuum dynamics by discretizing the action.  In mechanics and lattice field theory this permits to construct a state sum on a fixed and non-dynamical discretization. By contrast, for gravitational systems the dynamical nature of space-time also forces a discretization of the space-time geometry to be dynamical.
For instance, in Regge Calculus \cite{Regge:1961px,Williams:1996jb}---the most prominent simplicial discretization of General Relativity---the dynamical nature of the discretization is expressed, at the level of the action, by the fact that the equations of motion determine the lengths of the edges (in the bulk) of the space-time triangulation. 

In the canonical formulation of Regge Calculus \cite{Dittrich:2011ke,Hoehn:2011cm} the dynamical nature of the discretization additionally manifests itself in time evolution generically changing the spatial triangulation. Both its connectivity and the number of simplices contained in it vary in discrete time. This has severe consequences: it induces a temporally varying number of (kinematical and physical) geometric degrees of freedom \cite{Dittrich:2011ke,Dittrich:2013jaa}. Likewise, a varying number of matter degrees of freedom arises when modelling a field theory on a growing/shrinking lattice \cite{Dittrich:2013jaa,Hoehn:2014aoa,Foster:2004yc,Jacobson:1999zk}.

Such a discretization (or graph) changing canonical dynamics, in fact, appears in several quantum gravity approaches \cite{Thiemann:1996ay,Thiemann:1996aw,Alesci:2010gb,Dittrich:2011ke,Dittrich:2013xwa,Bonzom:2011hm} and remains a conundrum yet to be fully understood. Related to this, the relation between covariant state sum models and canonical quantum gravity approaches awaits a full clarification \cite{Noui:2004iy,Dittrich:2009fb,Dittrich:2011ke,Alesci:2011ia,Thiemann:2013lka,Alexandrov:2011ab,Bonzom:2011tf,Alesci:2010gb}.

%The Regge action is an important ingredient in several quantum gravity approaches and either appears in the construction of a gravitational path integral as in Quantum Regge Calculus or Causal Dynamical Triangulations \cite{qregge,cdt,lollreview}, or in the semiclassical limit of the state sum of spin foam models \cite{Conrady:2008mk,Barrett:2009gg,Perez:2012wv}. {\bf[maybe combine with statement that graph changing time evolution and relation canon=cov one of big issues in QG]}

%{\bf[combine this with previous paragraph!]} As a consequence of the dynamical space-time structure the discretization of the spatial geometry changes in time. In particular, for Regge Calculus time evolution generically leads to a change in both the connectivity of the spatial triangulation and the number of simplices contained in it \cite{Dittrich:2011ke}. This induces a temporally varying number of (kinematical and physical) degrees of freedom \cite{Dittrich:2013jaa}. However, such a temporally varying number of degrees of freedom also appears when modeling a quantum field theory on an expanding or shrinking universe by putting it on a growing or shrinking lattice \cite{foster,jacobson}. 

As a step in this direction, the classical dynamics of variational discrete systems with, in particular, temporally varying discretization have been systematically analyzed in \cite{Dittrich:2013jaa,Dittrich:2011ke,Hoehn:2014aoa}. Generalizing an earlier formalism \cite{marsdenwest,Gambini:2002wn,DiBartolo:2004cg,Gambini:2005vn,DiBartolo:2002fu,DiBartolo:2004dn,Campiglia:2006vy,Jaroszkiewicz:1996gr} (for constant discretization), this work provides a detailed constraint analysis and discussion of Dirac observables for such discrete systems.

What makes gravity special is that, as a generally covariant system, its (canonical) continuum dynamics is totally constrained and generated by the so-called Hamiltonian and diffeomorphism constraints \cite{martinbuch,Rovelli:2004tv,Thiemann:2007zz}. This is deeply intertwined with the diffeomorphism symmetry of the theory which likewise is generated by the constraints. Unfortunately, in the discrete the diffeomorphism symmetry is generically broken \cite{Bahr:2009ku,Dittrich:2008pw,Bahr:2011xs} such that Hamiltonian and diffeomorphism constraints do not in general arise \cite{Dittrich:2011ke}. Furthermore, the discrete dynamics is no longer generated by constraints or a proper Hamiltonian (which would have an infinitesimal action) but by so-called time evolution moves \cite{Dittrich:2009fb,Dittrich:2011ke,Dittrich:2013jaa,Hoehn:2011cm}. Only in special cases of so-called perfect discretizations \cite{Bahr:2009qc,Bahr:2011uj} does the discrete dynamics coincide with the (integrated) continuum dynamics.

Canonical constraints arising from the discrete action can thus assume quite different roles from their continuum counterparts. In particular, for variational discrete systems with temporally varying discretization such constraints will {\it always} arise even in the absence of any symmetries \cite{Dittrich:2011ke,Dittrich:2013jaa,Hoehn:2014aoa}. As argued in \cite{Dittrich:2013xwa} such constraints resulting from discretization changing time evolution moves can be viewed either as non-trivial coarse graining conditions or, in the opposite refinement case, as conditions that allow one to consistently represent a state carrying coarser information on a finer discretization. For a detailed classical discussion of the different roles of the constraints in the discrete, the notion of Dirac observables on dynamical discretizations and a classification of the many cases that can occur, we refer the reader to \cite{Dittrich:2013jaa,Hoehn:2014aoa}. 

%{\bf[following lines can be scrapped!]} This discussion also includes an analysis of the conditions under which constraints in the discrete are first or second class and under which conditions they generate gauge symmetries of the action. Furthermore, it includes a treatment of the notion of physical degrees of freedom as expressed in Dirac observables for variational discrete systems. For temporally varying discretizations the number of (independent) such Dirac observables too becomes temporally varying. 

The aim of the present work (incl.\ the companion article \cite{Hoehn:2014wwa}) is to better understand such discretization changing dynamics in variational discrete systems in the quantum theory. In this manuscript we shall focus on so-called {\it global} {time evolution moves} \cite{Dittrich:2013jaa} which, in a space-time context, correspond to evolving an entire spatial hypersurface at once. By contrast, the quantization of {\it local}{ time evolution moves} which correspond to local changes in the discretization (e.g.\ Pachner moves \cite{Dittrich:2011ke} in triangulations), will be discussed in the companion article \cite{Hoehn:2014wwa}. For simplicity, we shall restrict to systems with Euclidean configuration spaces $\mathbb{R}^N$, although the formalism can be suitably adapted to general configuration manifolds.

We shall discuss the different types of constraints and their roles in the quantum theory and examine under which conditions they are responsible for divergences in the path integral. The formalism developed here offers a systematic method for tracing and regularizing divergences arising in the construction of a path integral for variational discrete systems. In close analogy to the `general boundary formulation' of quantum theory \cite{Oeckl:2003vu,Oeckl:2005bv,Oeckl:2010ra,Oeckl:2011qd}, we shall deal with transition amplitudes between different discrete time steps which, in a space-time context, can correspond to very general (not necessarily spatial) discretized boundaries of space-time regions. 

Classically a temporally varying discretization leads to the notion of evolving phase spaces \cite{Dittrich:2013jaa,Dittrich:2011ke}. In the quantum theory it will be necessary to extend the analogous notion of finite dimensional evolving Hilbert spaces \cite{Doldan:1994yg} to the infinite dimensional case. We emphasize that the notion of an `evolving' Hilbert or phase space does not (necessarily) refer to a varying dimension but to the temporally varying number of degrees of freedom of the underlying discretization which are necessary in order to describe a classical or quantum state. As we shall see, this will in general lead to an evolution move dependence of the physical Hilbert spaces, i.e.\ the spaces of solutions to the quantum constraints, and of Dirac observables; a non-trivial coarse graining time evolution leads to a non-unitary projection of the physical Hilbert spaces and physical degrees of freedom.

%The novel formalism can be viewed as a discrete version of the `general boundary formulation' of quantum theory \cite{Oeckl:2003vu,Oeckl:2005bv,Oeckl:2010ra,Oeckl:2011qd} which permits the construction of transition amplitudes for very general (not necessarily spatial) boundaries of space-time regions. Here too we shall deal with transition amplitudes (or propagators) between different discrete time steps which, in a space-time context, can correspond to arbitrarily complicated triangulated hypersurfaces. 

Although the present formalism is motivated from the desire to better understand discretization changing dynamics in quantum gravity approaches, it does not directly apply to non-perturbative quantum gravity models. Rather, in the form given below it applies to systems such as a scalar field on a temporally varying discrete space-time structure in which only the field, but not the geometry itself is dynamical. Before applying this formalism to a quantum gravity model a few issues have to be taken into account. 

Firstly, Euclidean configuration spaces are not appropriate for non-perturbative quantum gravity. A generalization to arbitrary configuration manifolds should, however, leave the qualitative features of the formalism largely unchanged. Secondly, in the sequel we shall always assume that a (single) direction of the discrete evolution is given. In quantum gravity, on the other hand, one can argue that both `forward' and `backward' evolution are to be included in a path integral because these are indistinguishable from the perspective of the evolving hypersurface (see also the discussion in \cite{Dittrich:2013xwa,Hoehn:2014wwa,Teitelboim:1983fh}).\footnote{E.g., this explains why semiclassical spin foam amplitudes yield the Regge action in a cosine \cite{Conrady:2008mk,Barrett:2009gg,Perez:2012wv}.} Thirdly, below we shall determine the evolution of the physical quantum states, satisfying the quantum constraints, in discrete time. In discrete quantum gravity models, on the other hand, physical quantum states do {\it not} `evolve' in an external discrete time. In particular, if the continuum symmetries survive in the discrete, time evolution acts as a projector onto solutions to the quantum constraints \cite{Halliwell:1990qr,Rovelli:1998dx,Noui:2004iy,Thiemann:2013lka} such that physical states do not evolve. Instead, time evolution moves should rather be viewed as refining, coarse graining or entangling operations on physical states rather than generating an `evolution' of the latter \cite{Dittrich:2013xwa}. %;  physical states related by a refining time evolution can be identified as the same physical state but represented on different discretizations. However, if the continuum symmetries are broken, physical states do, indeed, change under these operations. 
We shall discuss these issues further in section \ref{sec_qg} and in \cite{Hoehn:2014wwa}.

The remainder of this article is organized as follows. In section \ref{sec_rev} we review the classical formalism for global evolution moves in variational discrete systems to make the article relatively self-contained. In section \ref{sec_reg1} we provide a summary of the situation in the quantum theory for regular global evolution moves where no constraints occur. A discussion of constrained global evolution moves in the quantum theory for systems with constant discretization is given in section \ref{sec_irreg}. Here group averaging projectors are used to construct physical states and propagators. In section \ref{sec_evol} the formalism is extended to temporally varying discretizations in which case the notion of cylindrical consistency appears. Section \ref{sec_fullcomp} discusses the generally non-trivial composition of constrained moves and the construction of a state sum. Section \ref{sec_dirac} elaborates on the notion of Dirac observables for temporally varying discretizations. To explicitly exhibit the concepts of this manuscript, we shall showcase a toy model for a `creation from nothing' in section \ref{sec_ex}. A conceptual discussion of the subtleties that arise when applying this formalism to quantum gravity models is provided in section \ref{sec_qg}. The paper finishes with a conclusion in section \ref{sec_conc}. 

Finally, local quantum evolution moves are discussed in the companion article \cite{Hoehn:2014wwa}.

\section{Review of the classical formalism}\label{sec_rev}

Before delving into the details of the quantization, it is necessary to review basic facts about the classical description of (global) time evolution in constrained variational discrete systems. For a detailed discussion and introduction to this formalism we refer the reader to \cite{Dittrich:2013jaa,Dittrich:2011ke} which builds up on and generalizes the earlier works \cite{marsdenwest,Jaroszkiewicz:1996gr,Gambini:2002wn,DiBartolo:2004cg,Gambini:2005vn,DiBartolo:2002fu,DiBartolo:2004dn,Campiglia:2006vy}.\footnote{For a related multisymplectic formulation see also the recent \cite{Arjang:2013dya}.} An explicit construction of the formalism and a detailed classification of constraints and degrees of freedom for quadratic discrete actions is given in \cite{Hoehn:2014aoa}.

In the systems under consideration, time evolution maps between discrete time steps such that a Hamiltonian as a generator of the dynamics (which would be infinitesimal) does not exist. Instead, the dynamics is generated by so-called time evolution moves. A {\it global evolution move}, which henceforth shall be denoted $n\rightarrow n+1$, carries out the discrete time evolution of the system from a given time step $n$ to the next time step $n+1$, where $n\in\mathbb{Z}$ labels the time steps. To each global evolution move there is associated an action $S_{n+1}(x_n,x_{n+1})$ which depends on the (continuous) configuration variables $x_n$ and $x_{n+1}$ from time steps $n,n+1$, respectively. That is, $x_n$ and $x_{n+1}$ coordinatize the (continuous) configuration manifolds $\cq_n$ and $\cq_{n+1}$ of the system at steps $n$ and $n+1$, respectively. For notational simplicity we shall often suppress a further index $x^i_n$, $i=1,\ldots,\dim\cq_n$, on the configuration variables. We emphasize that the formalism is applicable to arbitrary configuration manifolds and, in particular, to temporally varying discretizations which means that $\cq_n\ncong\cq_{n+1}$ is expressly allowed. A global evolution move is characterized by the feature that no neighbouring time steps $n,n+1$ share any subsets of coinciding variables. That is, different time steps $n,n+1$ do not overlap except in a possible boundary. In a space-time context, an evolution move can be viewed as a piece of space-time and the composition of two moves can be viewed as gluing two space-time regions together. This is illustrated in figure \ref{fig1}.

\begin{figure}[hbt!]
\begin{center}
\psfrag{0}{$0$}
\psfrag{1}{$1$}
\psfrag{2}{$2$}
\psfrag{s1}{$S_1$}
\psfrag{s2}{$S_2$}
\psfrag{p0}{\small${}^-p^0,{}^-C^0$}
\psfrag{p1}{\small ${}^+p^1,{}^+C^1$}
\psfrag{p12}{\small${}^-p^1,{}^-C^1$}
\psfrag{p2}{\small${}^+p^2,{}^+C^2$}
\includegraphics[scale=.5]{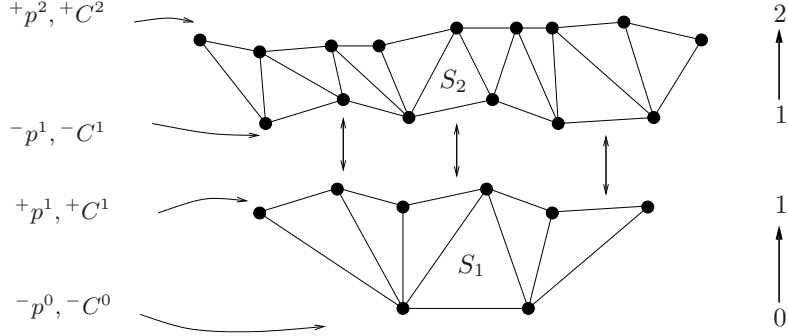}
\caption{\small Schematic illustration of two evolution moves $0\rightarrow1$ and $1\rightarrow2$. In a space-time context, the individual moves correspond to pieces of space-time. Each piece comes with its own set of pre-- and post--momenta and pre-- and post--constraints. The composition of two moves corresponds to a gluing of the two pieces and to solving the equations of motion or, equivalently, matching momenta at $n=1$, ${}^+p^1={}^-p^1$.}\label{fig1}
\end{center}
\end{figure}

An underlying requirement of the classical formalism is {\it additivity} of the action. More precisely, given a sequence of global evolution moves $n\rightarrow n+1\rightarrow\cdots\rightarrow n+X$, one can build an `effective' evolution move $n\rightarrow n+X$ by adding the action contributions of the individual moves
\ba
S_X=\sum_{m=1}^XS_{n+m}(x_{n+m-1},x_{n+m})\nn
\ea
and integrating out the variables associated to the intermediate time steps $n+1,\ldots,n+X-1$ by solving the intermediate equations of motion. Given that we are considering variational discrete systems and global moves, the (covariant) equations of motion for variables $x_n$ are obtained by varying the action $S_X$ with respect to these variables, 
\ba
\f{\p S_n(x_{n-1},x_n)}{\p x_n}+\f{\p S_{n+1}(x_n,x_{n+1})}{\p x_n}=0.\label{eom}
\ea
Since for global moves steps $n,n+1$ do not overlap, $x_n$ only occurs in the action contributions $S_n,S_{n+1}$ such that the equations of motion take the above form.

The classical formalism applies to `bare' evolutions moves as well as `effective' evolution moves such as $n\rightarrow n+X$ and we shall not further distinguish between these. The variables $x_n,x_{n+1}$ in the action contribution $S_{n+1}$ governing the global move $n\rightarrow n+1$ may thus contain `bulk' variables, i.e.\ variables $x_n^t$ whose equations of motion can be constructed from $S_{n+1}$ alone, $\f{\p S_{n+1}}{\p x^t_n}=0$.

The switch to the canonical picture can be made by noting that the discrete action $S_{n+1}(x_n,x_{n+1})$ coincides with Hamilton's principal function which is a generating function of the first kind (i.e.\ it depends on `old' and `new' configuration variables) for the canonical time evolution. `Old' and `new' momenta are obtained by differentiating the generating function,
\ba
{}^-p^n:=-\f{\p S_{n+1}(x_n,x_{n+1})}{\p x_n},\q\q\q{}^+p^{n+1}:=\f{\p S_{n+1}(x_n,x_{n+1})}{\p x_{n+1}}.\label{ham}
\ea
These equations also define the pre-- and post--Legendre transforms which map from $\cq_n\times\cq_{n+1}$---the discrete analogue of the tangent bundle $T\cq$ in the continuum---to the phase spaces $\cp_n:=T^*\cq_n$ at step $n$ and $\cp_{n+1}:=T^*\cq_{n+1}$ at step $n+1$, respectively \cite{Dittrich:2013jaa,Dittrich:2011ke}. Indeed, the so-called {\it pre--momenta} ${}^-p^n$ together with the $x_n$ define a canonical Darboux coordinate system on $\cp_n$, while, likewise, the {\it post--momenta} ${}^+p^{n+1}$ together with the $x_{n+1}$ define canonical coordinates on $\cp_{n+1}$. In particular, (\ref{ham}) defines an implicit global Hamiltonian time evolution map $\mathfrak{H}_{n}:\cp_n\rightarrow\cp_{n+1}$ because, given $(x_n,{}^-p^n)$, the left equation in (\ref{ham}) determines (possibly non-uniquely) $x_{n+1}$, while the right equation determines ${}^+p^{n+1}$.

If the Lagrangian two--form on $\cq_n\times\cq_{n+1}$ is degenerate---this is equivalent to
$\det\left(\f{\p^2S_{n+1}}{\p x^i_n\p x_{n+1}^j}\right)=0$---the Hamiltonian time evolution map $\mathfrak{H}_n$ (\ref{ham}) cannot be unique. Indeed, by the implicit function theorem, in this case both the left and right equations in (\ref{ham}) are not all independent and there exist relations 
\ba
{}^-C^n(x_n,{}^-p^n)=0,\q\q\q\q {}^+C^{n+1}(x_{n+1},{}^+p^{n+1})=0 \nn
\ea
on $\cp_n$ and $\cp_{n+1}$ which are called {\it pre--} and {\it post--constraints}, respectively \cite{Dittrich:2013jaa,Dittrich:2011ke}. The pre--constraints among themselves and the post--constraints among themselves each constitute first class Poisson algebras \cite{Dittrich:2013jaa}. The pre--constraint surface $\cc^-_n\subset\cp_n$ is the pre--image of $\mathfrak{H}_n$, while the post--constraint surface $\cc^+_{n+1}\subset\cp_{n+1}$ is the image of $\mathfrak{H}_n$, i.e.\ $\mathfrak{H}_n:\cc^-_n\rightarrow \cc_{n+1}^+$. In \cite{Dittrich:2011ke} it was shown that $\mathfrak{H}_n$ defines a pre--symplectic transformation; it preserves the symplectic structure which, however, at step $n$ is restricted to $\cc^-_n$ and at $n+1$ restricted to $\cc^+_{n+1}$. The pre-- and post--constraints account for all possible {\it primary} constraints arising at a given time step, not all of which are necessarily symmetry generators \cite{Dittrich:2013jaa,Hoehn:2014aoa}. Holonomic or boundary data constraints which are neither pre-- nor post--constraints arise only as {\it secondary} constraints upon implementing equations of motion \cite{Hoehn:2014aoa}.

Given that there are not sufficiently many independent equations in (\ref{ham}), for every post--constraint at $n+1$ there will exist a configuration datum $\lambda_{n+1}$ which cannot be predicted by the canonical data at $n$ and which we therefore call an {\it a priori free} parameter. Similarly, to every pre--constraint at $n$ there is associated a configuration datum $\mu_n$ at $n$ which cannot be postdicted, given the canonical data at $n+1$ and which we thus call an {\it a posteriori free} parameter. The set of pre--constraints generates the {\it pre--orbit} $\cg^-_n\subset\cp_n$ parametrized by the $\mu^i_n$, while the set of post--constraints generates the {\it post--orbit} $\cg^+_{n+1}\subset\cp_{n+1}$ parametrized by the $\lambda^i_1$ \cite{Dittrich:2013jaa}. %The $\lambda_{n+1}$ parametrize the {\it post--orbit} $\cg^+_{n+1}$ of the post--constraints, while the $\mu_n$ parametrize the {\it pre--orbit} $\cg^-_n$ of the pre--constraints at $n$.
 %This will be of importance later in the quantum theory.

Since (\ref{ham}) can be applied to every move, there exist both pre-- and post--momenta and both pre-- and post--constraints at each time step (see figure \ref{fig1}). However, the covariant equations of motion (\ref{eom}) are equivalent to a {\it momentum matching} \cite{marsdenwest,Dittrich:2011ke,Dittrich:2013jaa,Jaroszkiewicz:1996gr} at step $n$
\ba
p^n:={}^-p^n={}^+p^n\nn
\ea
such that on solutions each step is equipped with unique momenta. This momentum matching ultimately renders the canonical formalism equivalent to the covariant formulation.

The matching of the symplectic structures at step $n$ has very non-trivial repercussions for the dynamics \cite{Dittrich:2013jaa,Hoehn:2014aoa}, in particular, for a temporally varying discretization. We shall explain these repercussions at the relevant places along the way in the main body of this work. Let us, nevertheless, briefly summarize a few of them. The matching of symplectic structures at $n$ requires that the generally distinct sets of pre-- and post--constraints both have to be satisfied. This can lead to non-trivial restrictions of the dynamics and to a move dependence of the number of constraints at a given time step. Accordingly, the number of propagating degrees of freedom becomes move dependent in the general case. Furthermore, the combined set of pre-- and post--constraints at a step $n$ may generally be second class, while genuine gauge symmetry generators are first class. The formalism consistently describes how many canonical concepts become evolution move dependent for temporally varying numbers of degrees of freedom.

\section{Regular global discrete time evolution moves}\label{sec_reg1}

We begin by quantizing {\it regular}, i.e.\ non-constrained, {\it global} time evolution moves. The quantization of constrained global moves, including those leading to a temporally varying discretization, is discussed in sections \ref{sec_irreg}--\ref{sec_ex}. However, before we do so, it is useful to recall elementary properties of continuum propagators. For the remainder of this work, we decide to keep track of $\hbar$.

\subsection{Prelude: review of propagators in the continuum}\label{sec_cont}

In continuum quantum mechanics, the propagator is defined, in the position representation, as the transition amplitude between states at different times
\ba
K(q_1,t_1;q_0,t_0)=\left\langle q_1\left|e^{-i(t_1-t_0)\hat{H}/\hbar}\right|q_0\right\rangle\label{trans1},
\ea
where $\hat{H}$ is the Hamiltonian of the system and $q_{0,1}$ coordinatize the configuration manifold $\cq$. It satisfies the Schr\"odinger equation in {\it both} sets of variables
\ba
i\hbar\,\partial_{t_1}K(q_1,t_1;q_0,t_0)=\hat{H}K(q_1,t_1;q_0,t_0),\q\q\q i\hbar\,\p_{t_0}K^*(q_1,t_1;q_0,t_0)=\hat{H}K^*(q_1,t_1;q_0,t_0),\label{s-eq}
\ea
and, importantly, maps wave functions at $t_0$ one-to-one (at least for non-constrained moves) to wave functions at $t_1$:
\ba
\psi(q_1,t_1)=\int_{\cq}\,dq_0\,K(q_1,t_1;q_0,t_0)\,\psi(q_0,t_0).\label{contmap}
\ea
Four of the continuum propagator's basic properties are (e.g., see \cite{Klauderbook}):
\begin{itemize}
\item[(i)] \underline{Composition:} $K(q_2,t_2;q_0,t_0)=\int_{\cq} dq_1K(q_2,t_2;q_1,t_1)K(q_1,t_1;q_0,t_0)$, where $t_1<t_0$ and $t_2<t_1$ are permitted,
\item[(ii)] \underline{Time reversal:} $K(q_0,t_0;q_1,t_1)=\left(K(q_1,t_1;q_0,t_0)\right)^*$, 
\item[(iii)] \underline{Invertibility:} $\int_\cq dq_1\left(K(q_1,t_1;q_0,t_0)\right)^*K(q_1,t_1;q'_0,t_0)=\delta(q_0-q'_0)$, and
\item[(iv)] \underline{Infinitesimal transition:} $\lim_{t_1\rightarrow t_0}K(q_1,t_1;q_0,t_0)=\delta(q_1-q_0)$.
\end{itemize}

In the sequel, we shall study, among other aspects, how these well-known continuum properties of the propagator do or do not translate into a consistent quantum formalism for variational discrete systems.%---subject to the condition that the canonical coordinates take value in the full real numbers.

\subsection{Propagators for regular global evolution moves}\label{sec_reg}

%The goal of this article is the quantization of variational discrete systems in which time evolution proceeds by discrete evolution moves that may be labeled by $n\in\mathbb{Z}$. The corresponding classical formalism has been worked out in detail in \cite{Dittrich:2011ke,Dittrich:2013jaa}. 
 %To begin with, let us implement the quantization of {\it regular}, i.e.\ non-constrained, {\it global} time evolution moves. The quantization of constrained moves and moves involving a varying discretization---such as {\it local} moves---will be investigated in sections \ref{sec_irreg}--\ref{sec_ex} below.

%Let $\cq_n$ be the configuration manifold at time step $n$. 

For simplicity, we shall henceforth restrict to variational discrete systems with flat Euclidean configuration spaces $\cq_n\simeq\mathbb{R}^{N_n}$ (here Euclidean does {\it not} refer to a space-time signature).\footnote{The classical formalism in \cite{Dittrich:2011ke,Dittrich:2013jaa}, on the other hand, is applicable to general configuration manifolds $\cq_n$.} %Let $\cq_n$ be coordinatized by $x^i_n$, $i=1,\ldots,N_n$ (for notational clarity, we shall often omit the index $i$). 
We also specify to systems---as in Regge Calculus or scalar field theory on the lattice---in which no discretized time variable $t_n$ exists. The canonical Darboux coordinates on $T^*\mathbb{R}^{N_n}\simeq \mathbb{R}^{2N_n}$ take value on the full real line, $x^i_n,p^n_i\in(-\infty,\infty)$. Choosing the Hilbert space for time step $n$ (in standard position representation) to be $\ch_n=L^2(\cq_n,dx_n)$ with Lebesgue measure $dx_n$, the spectrum of the corresponding self-adjoint quantum operators will likewise be the real line. In this manner we avoid worrying about global or topological non-trivialities in the quantization \cite{isham2} which, for the purpose of these notes, would unnecessarily cloud the main results. The case of bounded and/or compact spectra, for instance relevant for Regge Calculus (and thereby spin foam models), will be studied elsewhere. In the regular case, the configuration manifolds at different time steps will be isomorphic to one another $\cq_n\simeq\cq_{n+1}\simeq\mathbb{R}^N$.

Consider a `bare' global evolution move $0\rightarrow1$, i.e.\ an evolution move which does not involve any bulk variables. The discrete propagator associated to this move {\it cannot} be given as a transition amplitude between states on one and the same Hilbert space in the form of the right hand side of (\ref{trans1}): the eigenstates of the `position operators' $\hat{x}^i$ at steps $n=0,1$ $\left|\vec{x}_0\right\rangle\in\ch_0$ and $\left|\vec{x}_{1}\right\rangle\in\ch_{1}$, respectively, are elements of two distinct Hilbert spaces $\ch_0,\ch_{1}$ and refer to different variables (e.g., in Regge Calculus the variables $x^i_0$ and $x^j_1$ refer to edges in two different hypersurfaces within the triangulation). %That is, an expression such as $\left\langle\vec{x}_m\right|\left.\vec{x}_n\right\rangle$ is not defined. 
Correspondingly, an expression such as $\left\langle\vec{x}_1\right|\left.\vec{x}_0\right\rangle$ is not defined. Instead, we shall construct the propagator directly as the quantum time evolution map between $\ch_0$ and $\ch_1$---in analogy to (\ref{contmap}). The inner product of a state in $\ch_1$ with the time evolution image (in $\ch_1$) of a state in $\ch_0$ can then be interpreted as the transition amplitude. In addition, given that time evolution is generated by the evolution moves, a Hamiltonian (which would generate a continuous time evolution) is absent in the systems under consideration \cite{Dittrich:2011ke,Dittrich:2013jaa}. Consequently, in the quantum theory a unitary of the type $\hat{U}(n,m)=e^{-i(n-m)\hat{H}/\hbar}$ cannot arise and be used to define the desired map between Hilbert spaces associated to arbitrary time steps $n,m$.

We must therefore proceed differently: just as in the classical formalism \cite{marsdenwest,Dittrich:2009fb,Dittrich:2011ke,Dittrich:2013jaa}, we shall employ the action (or rather Hamilton's principal function), instead of a Hamiltonian, to construct the propagator (and thereby transition amplitude) and to define the dynamics. This will directly link the path integral formalism with the canonical formulation. 

More precisely, in the spirit of the configuration space path integral expression for the continuum propagator, we shall make the following ansatz for the propagator of a global time evolution move $0\rightarrow1$. We associate a (possibly complex) path integration measure $M_{0\rightarrow1}(x_1,x_0)$ to this move and absorb it in the definition of the propagator\footnote{Recall that a time variable $t_{0,1}$ does not occur.}
\ba
K_{0\rightarrow1}(x_1,x_0):=M_{0\rightarrow1}(x_0,x_1)\,e^{iS_1(x_1,x_0)/\hbar}
,\label{prop}
\ea
where $M_{0\rightarrow1}$ remains to be determined and $S_1(x_0,x_1)$ is the classical action associated to $0\rightarrow1$. 

One may be surprised about this distinction between the measure and a phase factor involving the classical action---after all, even knowing $S_1$ does not seem particularly helpful as long as one is ignorant about $M_{0\rightarrow1}$. But this distinction is motivated by the desire to ensure the correct propagator expression in the semiclassical limit. Indeed, the classical canonical time evolution associated to the move $0\rightarrow1$ is generated by the classical action $S_1$ \cite{Dittrich:2011ke,Dittrich:2013jaa,marsdenwest}. The continuum analysis in \cite{miller1,miller2} concerning the semiclassical transformations corresponding to classical canonical transformations are therefore directly applicable to (regular) global evolution moves. Applying the results of \cite{miller1,miller2} to the present case implies that the condition of unitarity of the move $0\rightarrow1$ uniquely singles out that, as $\hbar\rightarrow0$,  
\ba
K_{0\rightarrow1}(x_0,x_1)\approx\sqrt{\left(\f{1}{-2\pi i\hbar}\right)^N\det\left(\f{\p^2 S_1(x_0,x_1)}{\p x^i_0\p x^j_1}\right)}\,e^{\f{i S_1(x_0,x_1)}{\hbar}}.\label{miller}
\ea
As is well known, this expression is exact (i.e.\ the full quantum expression) for actions quadratic in the configuration variables \cite{Klauderbook}. Consequently, in these cases, the propagator is indeed of the form (\ref{prop}). Accordingly, following the idea to absorb arbitrary quantum corrections to the semiclassical expression (\ref{miller}) in the measure $M_{0\rightarrow1}$, we shall henceforth use (\ref{prop}) as the general ansatz for the discrete propagator. Therefore, the real information about the quantum theory is contained in the measure. Otherwise, at this stage the measure is kept general; however, along the way, non-trivial restrictions on it will arise as a result of desired properties of the propagator.

As a side remark, we recall from \cite{Dittrich:2011ke,Dittrich:2013jaa} that $\f{\p^2 S_1(x_0,x_1)}{\p x^i_0\p x^j_1}$ is the coordinate expression for the Lagrange two-form $\Omega_{1}$, defined on $\cq_0\times\cq_1$. While this matrix is non-degenerate for regular global evolution moves, it is degenerate for irregular (constrained) global moves \cite{Dittrich:2011ke,Dittrich:2013jaa,Hoehn:2014aoa}. The above expression for the semiclassical limit of the propagator is therefore not valid for constrained moves. However, in the context of quadratic discrete actions this expression can be suitably regulated for constrained moves \cite{Hoehn:2014aoa}.

In contrast to the continuum and given the absence of a Hamiltonian, neither the discrete propagator nor the states need to satisfy a differential evolution equation such as (\ref{s-eq}). The discrete evolution is therefore subject to less conditions than its continuum counterpart. This has severe repercussions: for instance, the path integral measure may generally be non-unique.

Let us now begin with the construction. Classically, to every global move $0\rightarrow1$ there is associated a pair of phase spaces $\cp_0:=T^*\mathbb{R}^N$ and $\cp_1:=T^*\mathbb{R}^N$ at steps $n=0,1$, respectively \cite{Dittrich:2011ke,Dittrich:2013jaa}. In analogy, in the quantum theory we associate a pair of Hilbert spaces $\ch_0:=L^2(\mathbb{R}^N,dx_0)$ and $\ch_1:=L^2(\mathbb{R}^N,dx_1)$ with every global move $0\rightarrow1$. We refer to the elements ${}^-\psi_0(x_0)\in\ch_0$ and ${}^+\psi_1(x_1)\in\ch_1$ associated with the move $0\rightarrow1$ as {\it pre--} and {\it post--states}, respectively. We use the discrete propagator (\ref{prop})---in analogy to (\ref{contmap})---to define the time evolution map $U_{0\rightarrow1}:\ch_0\rightarrow\ch_1$ from {\it pre--states} at $n=0$ to {\it post--states} at $n=1$ via $U_{0\rightarrow1}:=\int dx_0 \,K_{0\rightarrow1}$ such that
\ba
{}^+\psi_1(x_1)=\int\,dx_0\,K_{0\rightarrow1}(x_0,x_1){}^-\psi_0(x_0).
\ea

Just as in the continuum (property (ii) of section \ref{sec_cont}), we require the reverse propagator to be the complex conjugate (see also \cite{Jaroszkiewicz:1996gr}):
\ba
K_{1\rightarrow0}(x_0,x_1)=M_{0\rightarrow1}^*(x_0,x_1)e^{-iS_1(x_0,x_1)/\hbar}=\left(K_{0\rightarrow1}(x_1,x_0)\right)^*.\label{inv-prop}
\ea
We shall see shortly that this implies unitarity of the move $0\rightarrow1$. Hence,
\ba
{}^+\psi_1(x_1)=\int\,dx_0\,K_{0\rightarrow1}(x_1,x_0){}^-\psi_0(x_0)=\int\,dx_0\,K_{0\rightarrow1}(x_1,x_0)\int\,dx'_1\left(K_{0\rightarrow1}(x'_1,x_0)\right)^*{}^+\psi_1(x'_1).\nn
\ea
This entails
\ba
\int\,dx_0\,K_{0\rightarrow1}(x_1,x_0)\left(K_{0\rightarrow1}(x'_1,x_0)\right)^*=\delta^{(N)}(x'_1-x_1)\label{cond1}
\ea
which, by (\ref{prop}), is a condition on the measure $M_{0\rightarrow1}$. (\ref{cond1}) will, in general, not uniquely determine $M_{0\rightarrow1}$.\footnote{This is unrelated to the fact that the left hand side can always be multiplied by some function $f(x_1,x'_1)$ defined on $\cq_1\times\cq_1$ without changing the result as long as $f(x_1,x_1)=1$. One can easily convince oneself that such a function must be of the form $f(x_1,x_1')=e^{i(r(x_1)-r(x_1'))}$, where $r(x_1)$ is an arbitrary real function on $\cq_1$, in order to be absorbed into the measure. But this means that the propagator (and thus states in $\ch_1$) are multiplied by a phase factor $e^{ir(x_1)}$, which amounts to a unitary change of representation without changing any dynamics. On the other hand, the possible non-uniqueness of $M_{0\rightarrow1}$ is related to the fact that the propagator does not need to satisfy a Schr\"odinger equation. In the continuum, the measure can often be determined via (\ref{s-eq}).} In complete analogy, by considering ${}^-\psi_0$ in terms of ${}^+\psi_1$, one finds a further condition
\ba
\int dx_1\left(K_{0\rightarrow1}(x_1,x_0)\right)^*K_{0\rightarrow1}(x_1,x'_0)=\delta^{(N)}(x'_0-x_0)\,,\label{cond2}
\ea
on the measure $M_{0\rightarrow1}(x_1,x_0)$. Both (\ref{cond1}, \ref{cond2}) are the discrete incarnation of the continuum property (iii) in section \ref{sec_cont} above.

Provided the two conditions (\ref{cond1}, \ref{cond2}) are fulfilled, the propagator defines a bijective quantum time evolution map between $\ch_0$ and $\ch_1$. Moreover, any of the three conditions (\ref{inv-prop}--\ref{cond2}) entails unitarity of the discrete evolution move $0\rightarrow1$. For instance, using (\ref{cond2}),
\ba
\int dx_1\left({}^+\phi_1(x_1)\right)^*{}^+\psi_1(x_1)&=&\int dx_1\,dx_0\,dx'_0\left(K_{0\rightarrow1}(x_1,x_0)\right)^*K_{0\rightarrow1}(x_1,x'_0)\left({}^-\phi_0(x_0)\right)^*{}^-\psi_0(x'_0)\nn\\
&\underset{{(\ref{cond2})}}{=}&\int dx_0\left({}^-\phi_0(x_0)\right)^*{}^-\psi_0(x_0)
\ea
and thus
\ba
\left\langle{}^+\phi_1\left|\right.{}^+\psi_1\right\rangle_{\tiny\ch_1}=\left\langle{}^-\phi_0\left|\right.{}^-\psi_0\right\rangle_{\tiny\ch_0}.
\ea
By using (\ref{cond1}), one can prove the reverse direction.

The transition amplitudes for the evolution $0\rightarrow1$ can now be written as
\ba
\langle{}^+\phi_1\big|\,U_{0\rightarrow1}\,{}^-\psi_0\rangle_{\ch_1}=\int dx_1\,dx_0\,({}^+\phi_1)^*\,K_{0\rightarrow1}\,{}^-\psi_0.\nn
\ea

Next, let us compose the moves $0\rightarrow1$ and $1\rightarrow2$ to an effective move $0\rightarrow2$ such that we have bulk variables at $n=1$. Classically, there is a single phase space $\cp_1=T^*\mathbb{R}^N$ at $n=1$ for both moves $0\rightarrow1$ and $1\rightarrow2$. This leads to a momentum matching at $n=1$ which is equivalent to imposing the equations of motion \cite{marsdenwest,Dittrich:2011ke,Dittrich:2013jaa,Hoehn:2014aoa,Jaroszkiewicz:1996gr}. An analogous state of affairs holds in the quantum theory: there is only one Hilbert space $\ch_1=L^2(\mathbb{R}^N,dx_1)$ at $n=1$ for both moves $0\rightarrow1$ and $1\rightarrow2$ such that for consistency of the evolution we must perform a {\it state matching} at $n=1$, i.e.\ the {\it pre--states} must coincide with the {\it post--states}, 
\ba
{}^-\psi_1={}^+\psi_1.
\ea
Given that the momentum operators $\hat{p}^1_i$ acting on pre-- and post--states are the same, one finds a {\it quantum momentum matching} in the form $\hat{p}^1_i{}^-\psi_1=\hat{p}^1_i{}^+\psi_1$---which amounts to a matching of the momentum eigenvalues, provided ${}^+\psi_1$ is an eigenstate of $\hat{p}^1_i$. This enables us us to write down the quantum version of the classical pre-- and post--momenta \cite{Dittrich:2011ke,Dittrich:2013jaa,Hoehn:2014aoa}, ${}^-p^1=-\f{\p S_2}{\p x_1}$ and ${}^+p^1=\f{\p S_1}{\p x_1}$, respectively. For instance, the quantum version of the post--momenta reads
\ba
\hat{p}^1{}^+\psi_1=-i\hbar\,\p_{x_1}{}^+\psi_1=\int dx_0\left(\f{\p S_1}{\p x_1}-i\hbar\f{\p \ln(M_{0\rightarrow1})}{\p x_1}\right)K_{0\rightarrow1}{}^-\psi_0.\nn
\ea

In analogy to the continuum (property (i) in section \ref{sec_cont}), the propagator of the composition of the moves is to be the convolution
\ba
K_{0\rightarrow2}(x_2,x_0)=\int\,d{x}_1\,K_{1\rightarrow2}({x}_2,{x}_1)\,K_{0\rightarrow1}({x}_1,{x}_0).\label{comp}
\ea
This allows us to consistently write
\ba
{}^+\psi_2(x_2)&=&\int\,dx_1\,K_{1\rightarrow2}(x_2,x_1){}^-\psi_1(x_1)\nn\\
&=&\int\,dx_1\,K_{1\rightarrow2}(x_2,x_1)\,{}^+\psi_1(x_1)\nn\\
&=&\int\,dx_1\,K_{1\rightarrow2}(x_2,x_1)\int\,dx_0\,K_{0\rightarrow1}(x_1,x_0){}^-\psi_0(x_0)\nn\\
&=&\int\,dx_0\,K_{0\rightarrow2}(x_2,x_0)\,{}^-\psi_0(x_0)
\ea
and likewise for any other time step difference. Obviously, the propagator $K_{0\rightarrow2}$ must satisfy the analogous conditions to (\ref{inv-prop}--\ref{cond2}) for the `effective' move $0\rightarrow2$ to be unitary.

(\ref{inv-prop}--\ref{cond2}, \ref{comp}) are the four basic requirements on the propagators in the discrete. Notice that continuum property (iv), $\lim_{t_1\rightarrow t_2}K(q_1,t_1;q_0,t_0)=\delta(q_1-q_0)$, which is the initial value condition for the continuum propagator satisfying (\ref{s-eq}) \cite{Klauderbook}, is meaningless in the systems under consideration because of the absence of: (a) a time variable which could be made arbitrarily small, and (b) an evolution equation for the propagator such as (\ref{s-eq}).\footnote{The only situation in which one would regain an analogous property is a `time step relabeling' move $0\rightarrow1$ which keeps all variables fixed but replaces the time step label $`0`$ with the new label $`1`$. The propagator corresponding to such a relabeling move would indeed be $\delta^{(N)}=(x_1-x_0)$.}

The composition of the sequence of moves $0\rightarrow1\rightarrow2\rightarrow\cdots\rightarrow n$ to the `effective' move $0\rightarrow n$ yields the `path integral' (PI), or rather state sum, 
\ba\label{pi}
K_{0\rightarrow n}(x_n,x_0)&=&\int_{\cq^{n-1}}\prod^{n-1}_{j=0}K_{j\rightarrow j+1}(x_{j+1},x_j)\prod^{n-1}_{l=1}dx_l\nn\\
&\underset{(\ref{prop})}{=}&\int_{\cq^{n-1}}e^{i/\hbar\sum^n_{k=1}S_k(x_k,x_{k-1})}\prod^{n-1}_{j=0}M_{j\rightarrow j+1}(x_{j+1},x_j)\prod^{n-1}_{l=1}dx_l.
\ea

Given that we are employing the propagator defined via the action to construct the time evolution of quantum states, the equivalence of the canonical and path integral quantization of these variational discrete systems is essentially a tautology. For the same reason, this equivalence holds for the constrained evolution moves to be discussed below.

The construction presented above is formulated in what one could call the `Schr\"odinger picture' for time evolution on the set of Hilbert spaces $\{\ch_j\}_{j=0}^n$ which are unitarily related by the maps $U_{j\rightarrow j+1}$. In this picture the discrete time evolution of quantum states is determined by the propagators, while on each Hilbert space $\ch_n$ `observables' are described by (densely defined) self-adjoint operators $\hat{O}_n(\hat{x}_n,\hat{p}^n)$. This `Schr\"odinger picture' appears intuitive from a geometrical point of view, in particular, with an application of such a formalism to a scalar field on a varying spatial triangulation in mind where time evolution is geometrically described by (spatial hypersurface) discretization changing moves \cite{Dittrich:2013jaa,Hoehn:2014aoa,Dittrich:2013xwa}. In such a scenario, it is intuitive to regard the quantum state as an evolving entity and the `observables' as being associated to a fixed time step $n$ and acting on $\ch_n$. 

However, it may be useful to have an equivalent `Heisenberg picture' at hand in which operators, acting on one Hilbert space $\ch_n$, are evolved to operators acting on another $\ch_{m}$. This can easily be done within the present framework. For instance,
\ba
\hat{O}_0\,\cdot:=\int\,dx_n\,(K_{0\rightarrow n}(x'_0,x_n))^*\hat{O}_n(\hat{x}_n,\hat{p}^n)\int\,dx_0\,K_{0\rightarrow n}(x_0,x_n)\,\cdot\nn
\ea
is an `observable' acting on pre--states ${}^-\psi_0$ in $\ch_0$. In this way we can equivalently describe the dynamics on the fixed Hilbert space $\ch_0$ by evolving the `observables' forward and backward in discrete time. We shall not go into further details here. For related work on a `Heisenberg picture' in a discrete context with, however, evolution on a fixed Hilbert space $\ch$, $\forall\,n$, see \cite{Jaroszkiewicz:1996gr,Gambini:2002wn,DiBartolo:2002fu,Bahr:2011xs,Campiglia:2006vy}. In section \ref{sec_dirac} and in \cite{Hoehn:2014wwa} we shall elaborate on the evolution of Dirac observables in the context of temporally varying discretizations.

\section{Constrained global evolution moves in the quantum theory}\label{sec_irreg}

Next, let us elaborate on constrained global moves with the restriction that $\dim\cq_n=\dim\cq_{n+1}$. The generalization to temporally varying discretizations with $\dim\cq_n\neq\dim\cq_{n+1}$ will be studied in the next section \ref{sec_evol}. The composition of constrained moves and the path integral will be considered in section \ref{sec_fullcomp}. Again, we assume $\cq_n\simeq\mathbb{R}^N$.

Let $0\rightarrow1$ be a constrained global evolution move which, recalling section \ref{sec_rev}, is constrained by the pre--constraints ${}^-C_I^0$ at $n=0$ and the post--constraints ${}^+C_I^1$ at $n=1$, $I=1,\ldots,k$. %In the classical discrete formalism \cite{Dittrich:2011ke,Dittrich:2013jaa}, to each time step $n$, there are associated {\it two} generally distinct constraint surfaces in the phase space $T^*\cq_n$. Namely, for the move $0\rightarrow1$, the pre-- and post--constraint surfaces arise as the image of the pre-- and post--Legendre transforms which map from $\cq_0\times\cq_1$---the discrete analogue of the tangent bundle $T\cq$ in the continuum---to the phase spaces $T^*\cq_0$ and $T^*\cq_1$, respectively. These pre-- and post--Legendre transforms exist for {\it any} global time evolution move. $0\rightarrow1$ is thus constrained by the pre--constraints ${}^-C_I^0$ at $n=0$ and the post--constraints ${}^+C_I^1$ at $n=1$, $I=1,\ldots,k$, where the number of pre--constraints at $n=0$ coincides with the number of post--constraints at $n=1$. The pre-- and post--constraints account for all possible constraints arising at a given time step, not all of which are necessarily symmetry generators. However, the set of pre--constraints forms a first class Poisson algebra as does the set of post--constraints \cite{Dittrich:2013jaa}.
%Classically, every pre--constraint ${}^-C^0$ is associated with an {\it a posteriori} free parameter $\mu_0$, while every post--constraint ${}^+C^1$ is associated with an {\it a priori} free parameter $\lambda_1$. The set of pre--constraints and the set of post--constraints generate an {\it a posteriori free orbit} $\cg^-_0$ and an {\it a priori free orbit} $\cg^+_1$ at $n=0,1$, respectively \cite{Dittrich:2013jaa}. $\cg^+_1$ and $\cg^-_0$ are parametrized by the $\lambda^i_1$ and $\mu^i_0$, respectively. For details on the classical formalism, we refer the reader to \cite{Dittrich:2013jaa}. 
These pre-- and post--constraints ought to be satisfied in the quantum theory. To this end, we shall follow the Dirac algorithm \cite{Dirac,Henneaux:1992ig}, i.e.\ begin with a kinematic quantization followed by an imposition of the constraints on the quantum states and a completion of the resulting solution space to a physical Hilbert space. As the kinematical Hilbert spaces we shall choose $\ch^{\rm kin}_0:=L^2(\mathbb{R}^N,dx_0)$ and $\ch_1^{\rm kin}:=L^2(\mathbb{R}^N,dx_1)$ with standard Lebesgue measures $dx_0,dx_1$. %\footnote{Note that any other representation with non-trivial (inner product) measure $d\xi(x)=\xi(x)dx$ is equivalent to the standard one. In a (position) representation with measure $\xi(x)$, the self--adjoint momentum operator reads
%\ba
%\hat{p}_\xi=-i\hbar\left(\partial_x+\f{\p_x\xi(x)}{2\xi(x)}\right)\nn.
%\ea
%One easily convinces oneself that this representation is related to the standard representation with Lebesgue measure by the simple rescalings $\psi=\sqrt{\xi(x)}\,\psi_\xi$ and $\hat{p}=\sqrt{\xi(x)}\,\hat{p}_\xi\,(\sqrt{\xi(x)})^{-1}$, where $\psi,\psi_\xi$ are the states in the representations with Lebesgue and with non--trivial measure $\xi(x)$, respectively. $\hat{p}$ is the self--adjoint operator in the representation with Lebesgue measure. On account of this equivalence, we simply employ the standard representation.} 
The quantum pre-- and post--constraints are to be self-adjoint with respect to the $L^2$ kinematical inner product (KIP) on $\ch^{\rm kin}_{0,1}$.

At this stage we shall make some assumptions both for simplifying matters and in order not to keep the discussion entirely formal. We shall assume the following: 
\begin{itemize}
\item[(1)] the spectrum of the quantum pre-- and post--constraints is absolutely continuous; 
\item[(2)] the quantization of the constraints is consistent and anomaly free, i.e.\ the quantum pre-- and post--constraints each form a proper first class constraint algebra as they do in the classical formalism \cite{Dittrich:2013jaa}; 
\item[(3)] The pre-- and post--orbits $\cg^-_0,\cg^+_1$ (see section \ref{sec_rev}) are non--compact; %{\bf [does this actually follow from absolute continuity?]}
\item[(4)] there are no global obstructions for fixing the flow of any of the constraints on the corresponding $\cg^+_1,\cg^-_0$---i.e.\ no Gribov problem arises.
\end{itemize}
This has significant repercussions in the quantum theory: in general, one integrates over the non--compact $\cg^+_1,\cg^-_0$ such that we have to expect a number of divergencies to appear in the construction. These infinities must be regularized, e.g.\ by factoring out the orbit volume by means of the Faddeev--Popov trick. Fortunately, the presently devised formalism will naturally keep track of these divergencies. In section \ref{sec_lin} and in a concrete example in section \ref{sec_ex} below we shall see this beyond the formal level.

\subsection{Abbreviations}

In the sequel, we shall often make use of abbreviations. For reference and clarity, we list them here:
\begin{table}[htdp]
%\caption{default}
\begin{center}
\begin{tabular}{|c|c|}
\hline
PI & path integral\\\hline
KIP & kinematical inner product\\\hline
PIP & physical inner product\\\hline
PIP+ & post--physical inner product\\\hline
PIP-- & pre--physical inner product\\
\hline
\end{tabular}
\end{center}
\label{default}
\end{table}%

\subsection{Group averaging}

Given that we are dealing with mechanical systems on configuration spaces $\cq_n\simeq\mathbb{R}^N$, we shall henceforth employ group averaging techniques \cite{Marolf:1995cn,Marolf:2000iq,Rovelli:2004tv,Thiemann:2007zz} in order to construct what we shall call the {\it pre--} and {\it post--physical Hilbert spaces},  ${}^-\ch^{\rm phys}_0,{}^+\ch^{\rm phys}_1$, respectively.\footnote{This notion of a {\it pre--physical Hilbert space} should not be confused with the mathematical notion of a pre--Hilbert space. The latter is a complex vector space with hermitian inner product prior to Cauchy completion to a proper Hilbert space.} More precisely, the {\it pre--physical Hilbert space} at step $n$ is to consist of the set of all {\it pre--physical states} ${}^-\psi^{\rm phys}_n$ where the latter are the states annihilated by the quantum pre--constraints at $n$, ${}^-\hat{C}^n_I\,{}^-\psi^{\rm phys}_n=0$. Likewise, the {\it post--physical Hilbert space} at $n$ is comprised of the set of {\it post--physical states} ${}^+\psi^{\rm phys}_n$ which solve the quantum post--constraints at $n$ in the form ${}^+\hat{C}^n_J\,{}^+\psi^{\rm phys}_n=0$. 

In particular, under the above assumptions, we can define the (improper) {\it pre--} and {\it post--projectors} 
\ba\label{proj000}
{}^-\mathbb{P}_0:\ch^{\rm kin}_0\rightarrow{}^-\ch^{\rm phys}_0,\q\q\q{}^-\mathbb{P}_0:=\prod_{I=1}^k\delta({}^-\hat{C}^0_I),\nn\\
{}^+\mathbb{P}_1:\ch^{\rm kin}_1\rightarrow{}^+\ch^{\rm phys}_1,\q\q\q{}^+\mathbb{P}_1:=\prod_{I=1}^k\delta({}^+\hat{C}^1_I),
\ea
respectively, where $\delta(\hat{C})=1/(2\pi\hbar)\int_\mathbb{R} ds\,e^{is\hat{C}/\hbar}$. For instance, at $n=1$ we can then construct the {\it post--physical states} by an (improper) projection onto solutions of the $k$ post--constraints
\ba
{}^+\psi^{\rm phys}_1(x_1):={}^+\mathbb{P}_1\,\psi^{\rm kin}_1(x_1):=\f{1}{(2\pi\hbar)^{k}}\int_{\mathbb{R}^{k}}\prod^{k}_{J=1}\left(ds^J_1e^{is^J_1{}^+\hat{C}^1_J/\hbar}\right)\psi^{\rm kin}_1(x_1),\label{proj}
\ea
where $\psi^{\rm kin}_1\in\ch^{\rm kin}_1$ is a state\footnote{We abstain from indexing kinematical states with a `+' or `--' since $\ch^{\rm kin}_1$ is the kinematical Hilbert space at $n=1$ for both $0\rightarrow1$ and $1\rightarrow2$.} in the kinematical Hilbert space at $n=1$ and the $s^J_1$ are classical parameters parametrizing the flows of ${}^+\hat{C}^1_J$. %the {\it a priori} free parameters of the move $0\rightarrow1$. %A physical pre--state at $n=0$ is obtained in analogous fashion with the pre--constraints simply replacing the post--constraints. 
This expression formally solves all the quantum post--constraints because $\cg^+_1$ is averaged out. By relabeling the integration parameter in (\ref{proj}), one can easily convince oneself that ${}^+\mathbb{P}_1=({}^+\mathbb{P}_1)^\dag$ and similarly ${}^-\mathbb{P}_0=({}^-\mathbb{P}_0)^\dag$. (\ref{proj}) is at this stage a formal expression because in general ambiguities in the ordering of the $\delta({}^+\hat{C}^1_I)$ will arise. We shall not worry about this last issue here (for a general discussion of group averaging, see \cite{Marolf:1995cn,Marolf:2000iq,Thiemann:2007zz}).

\begin{R}
\emph{Before we proceed, it is important to note that there exists an inequivalent alternative construction. Instead of defining the} {post--projector} \emph{via (\ref{proj000}, \ref{proj}), one could have constructed it as}
\ba
{}^+\mathbb{P}'_1:=\f{1}{(2\pi\hbar)^{k}}\int_{\mathbb{R}^{k}}\prod^{k}_{J=1}ds^J_1e^{i/\hbar\sum_{J=1}^ks^J_1{}^+\hat{C}^1_J}\label{altproj}
\ea
\emph{(and analogously for the} pre--projector\emph{). This definition is inequivalent to (\ref{proj000}, \ref{proj}) because the constraints will, in general, not commute. (There may even be cases in which at least one of the two methods will not lead to a valid explicit construction of a projector onto solutions of} all \emph{quantum constraints simultaneously.) The advantage of (\ref{altproj}) is that ${}^+\mathbb{P}'_1$ thus defined is invariant under a linear transformation of the constraints (and a simultaneous inverse linear transformation of the integration parameters)---in contrast to ${}^+\mathbb{P}_1$ in (\ref{proj000}, \ref{proj}). This is closer in philosophy to the classical theory where the particular choice of the constraints is unimportant as long as they are related by linear transformations.}

\emph{However, in the sequel we shall rarely make explicit use of the particular formal construction of the projectors. Below it will only matter that the projectors impose the associated constraints and not in which of the two ways they are implemented. Hence, the subsequent formal calculations would be unaffected if one used ${}^+\mathbb{P}'_1$ defined in (\ref{altproj}) instead of ${}^+\mathbb{P}_1$ (and similarly in the case of the} pre--projectors\emph{). The exception is section \ref{sec_fullcomp} where we need to distinguish different types of pre-- and post--constraints when composing pairs of evolution moves and new constraints arise from the `future' or `past' evolution move. At this stage the construction of (\ref{proj000}, \ref{proj}) becomes notationally advantageous for decomposing ${}^+\mathbb{P}_1,{}^-\mathbb{P}_1$ into sub-projectors. This is the reason why we shall henceforth work with the definition of the} pre-- \emph{and} post--projectors ${}^+\mathbb{P}_1,{}^-\mathbb{P}_1$ \emph{as given in (\ref{proj000}, \ref{proj}). Nevertheless, one could similarly work with the construction as in (\ref{altproj}). In section \ref{sec_fullcomp} we shall briefly comment on how one would have to proceed in this case. }

\emph{We thus emphasize that the entire formalism is (up to the differences in section \ref{sec_fullcomp}) independent of which specific formal construction for the projectors is employed. In particular, the conclusions drawn in this article are valid for both choices of the construction. Furthermore, the explicit construction of section \ref{sec_lin} and the toy model of section \ref{sec_ex} involve Abelian constraints in which case the two choices are equivalent. Without loss of generality of the formal calculations and the qualitative conclusions we shall thus henceforth work with the construction as given in (\ref{proj000}, \ref{proj}).}
\end{R}

${}^+\mathbb{P}_1$ is an improper projector because $\delta({}^+\hat{C}^1_I){}^+\psi^{\rm phys}_1=1/(2\pi\hbar)^{k} \int ds^I_1{}^+\psi^{\rm phys}_1=``\infty" {}^+\psi^{\rm phys}_1$ as a consequence of the non-compact orbits. Hence, $({}^+\mathbb{P}_1)^2$, acting on a kinematical state, leads to a divergence. This is the origin of those divergences in the path integral which are related to gauge symmetry and will become important  further below.

The post--physical states are not normalizable with respect to the KIP because of the integration over the non--compact $\cg^+_1$. The constraint quantization for $0\rightarrow1$, i.e.\ the construction of the pre-- and post--physical Hilbert spaces ${}^-\ch^{\rm phys}_0,{}^+\ch^{\rm phys}_1$, is of course not complete unless we endow the solution spaces to the constraints with a physical inner product normalizing the physical states. (One would also have to complete in norm as well as to divide out spurious solutions and zero physical norm states which, however, we shall ignore in this work.) Group averaging permits us to define the {\it post--physical inner product} (PIP+) between two physical post--states at $n=1$ as follows \cite{Marolf:1995cn,Marolf:2000iq,Rovelli:2004tv,Thiemann:2007zz}
\ba
\left\langle{}^+\psi^{\rm phys}_1\Big|{}^+\phi^{\rm phys}_1\right\rangle_{\rm phys+}&=&\left\langle{}^+\mathbb{P}_1\,\psi^{\rm kin}_1\Big|{}^+\mathbb{P}_1\,\phi^{\rm kin}_1\right\rangle_{\rm phys+}\nn\\
&:=&\left\langle\psi_1^{\rm kin}\Big|{}^+\mathbb{P}_1\,\phi^{\rm kin}_1\right\rangle_{\rm kin},\label{pip}
\ea
where $\langle.|.\rangle_{\rm kin}$ is the KIP of $\ch_1^{\rm kin}:=L^2(\mathbb{R}^N,dx_1)$. As one can check, this PIP+ is indeed defined on ${}^+\ch_1^{\rm phys}$, namely on the equivalence classes of kinematical states, where $\psi^{\rm kin}_1,\tilde{\psi}^{\rm kin}_1\in\ch_1^{\rm kin}$ are equivalent if they yield the same physical post--state, i.e.\ ${}^+\mathbb{P}_1\,\psi^{\rm kin}_1={}^+\mathbb{P}_1\,\tilde{\psi}^{\rm kin}_1$.

It is a typical feature of the group averaging method that it is rather difficult to make any general statements beyond the formal level. However, we may expect that, for a large class of systems, we can formally rewrite the PIP+  in the position representation as
\ba
\int dx_1(\psi^{\rm kin}_1(x_1))^*{}^+\phi^{\rm phys}_1(x_1)=\int d\xi^+_1(x_1)({}^+\psi^{\rm phys}_1(x_1))^*{}^+\phi^{\rm phys}_1(x_1),\label{regpip}
\ea
where $d\xi_1^+$ is some, in general, non-trivially regularized measure which breaks the {\it a priori} free parameter flow of the quantum post--constraints. We shall call $d\xi^+_1$ the {\it post--measure} of the PIP+ and often consider systems for which the PIP+ takes this regularized form. In section \ref{sec_lin}, we shall impose some (rather strong) conditions which will allow us to explicitly determine $d\xi_1^+$ via the Faddeev-Popov trick. These conditions are, in particular, fulfilled by systems governed by quadratic discrete actions \cite{Hoehn:2014aoa}.

The {\it pre--physical states} at $n=0$
\ba
{}^-\psi^{\rm phys}_0(x_0):={}^-\mathbb{P}_0\,\psi^{\rm kin}_0(x_0)\label{prestates}
\ea
and the {\it pre--physical inner product} (PIP--) at $n=0$
\ba
\left\langle{}^-\psi^{\rm phys}_0\Big|{}^-\phi^{\rm phys}_0\right\rangle_{\rm phys-}:=\left\langle\psi_0^{\rm kin}\Big|{}^-\mathbb{P}_0\,\phi^{\rm kin}_0\right\rangle_{\rm kin},\nn
\ea
are constructed in complete analogy. If one considered a further move $1\rightarrow2$, the {\it pre--measure} $d\xi^-_1$ at $n=1$ will generally {\it not} coincide with the post--measure $d\xi_1^+$ such that PIP+ $\neq$ PIP-- at $n=1$. That is, in general, ${}^+\ch^{\rm phys}_1\neq{}^-\ch^{\rm phys}_1$. We shall worry about this issue in section \ref{sec_comp}.%In section \ref{sec_ex} below, we shall see such a construction explicitly in a toy model.

\subsection{Propagators and unitarity}\label{sec_unitary}

In section \ref{sec_reg} we employed the propagator associated to a regular global move $0\rightarrow1$ to construct a bijection between $\ch_0$ and $\ch_1$ (provided, of course, certain conditions on the propagator are satisfied). Clearly, for a constrained move there cannot be a dynamical bijection between $\ch^{\rm kin}_0$ and $\ch^{\rm kin}_1$ on account of the {\it a priori} and {\it a posteriori} free orbits. Instead, we shall employ the propagator in order to define an improper projector from $\ch^{\rm kin}_0$ to ${}^+\ch^{\rm phys}_1$ and, likewise from $\ch^{\rm kin}_1$ to ${}^-\ch^{\rm phys}_0$---much in analogy to (\ref{proj}).

The idea, coming from quantum gravity, is to employ the path integral as a `projector' onto solutions of the quantum constraints \cite{Halliwell:1990qr,Rovelli:1998dx,Noui:2004iy,Thiemann:2007zz,Dittrich:2013xwa,Thiemann:2013lka}. The action $S_1$ of the move $0\rightarrow1$ contains the information about the pre--constraints and their {\it a posteriori free} parameters $\mu_0$ at $n=0$ as well as the post--constraints and their {\it a priori free} parameters $\lambda_1$ at $n=1$. Therefore, heuristically an integration over $e^{iS_1/\hbar}$ contains a `group averaging' over the constraint flows and should thereby implement the quantum constraints. 

Let us make this more precise. For a constrained move $0\rightarrow1$ we make the ansatz
\ba
{}^+\psi^{\rm phys}_1(x_1)=\int dx_0\,K_{0\rightarrow1}(x_1,x_0)\,\psi^{\rm kin}_0(x_0),\label{proj01}
\ea
where $K_{0\rightarrow1}$, containing the factor $e^{iS_1/\hbar}$, is given by (\ref{prop}), as before. For the purpose of unitarity of $0\rightarrow1$, we again require (\ref{inv-prop}). Indeed, the measure $M_{0\rightarrow1}$ allows for enough freedom such that (\ref{proj01}) makes sense. There are now natural consistency conditions on this ansatz. In particular, imposing the post--constraints at $n=1$ in the quantum theory must yield
\ba
{}^+\hat{C}_I^1{}^+\psi^{\rm phys}_1(x_1)=\int dx_0{}^+\hat{C}_I^1K_{0\rightarrow1}(x_1,x_0)\,\psi^{\rm kin}_0(x_0)\overset{!}{=}0.\label{cond3}
\ea
The fact that ${}^+\hat{C}^1_I$ only acts on $K_{0\rightarrow1}$ and (\ref{cond3}) should hold for any (permissible) `initial' kinematical state $\psi^{\rm kin}_0$ entails that the propagator itself should satisfy the quantum post--constraints
\ba
{}^+\hat{C}_I^1K_{0\rightarrow1}(x_1,x_0)\overset{!}{=}0.\label{cond4}
\ea
By analogous arguments, the inverse propagator must be annihilated by the quantum pre--constraints
\ba
{}^-\hat{C}_I^0K_{1\rightarrow0}(x_0,x_1)={}^-\hat{C}_I^0\left(K_{0\rightarrow1}(x_1,x_0)\right)^*\overset{!}{=}0.\label{cond5}
\ea

This suggests to introduce the notion of a {\it kinematical propagator}, $\kappa_{0\rightarrow1}(x_0,x_1)=(\kappa_{1\rightarrow0}(x_0,x_1))^*$, as a function on $\cq_0\times\cq_1$ which does not satisfy any constraints but which is square integrable in the KIPs at both time steps $n=0$ and $n=1$---much in analogy to kinematical quantum states. Using the kinematical propagator, we can thus choose to write the {\it physical propagator} in the form ( ${}^*$ denotes complex---not hermitian---conjugation)
\ba
K_{0\rightarrow1}(x_0,x_1)={}^+\mathbb{P}_1\,({}^-\mathbb{P}_0)^*\,\kappa_{0\rightarrow1}(x_0,x_1).\label{conprop}
\ea
Just as for kinematical states there will exist a pre--orbit at $n=0$ and a post--orbit at $n=1$ for kinematical propagators such that there is no unique choice for the latter. However, the choice out of two kinematical propagators $\kappa_{0\rightarrow1},\kappa'_{0\rightarrow1}$ does not matter, even if
\ba
{}^+\mathbb{P}_1\,\kappa_{0\rightarrow1}(x_0,x_1)\neq{}^+\mathbb{P}_1\,\kappa'_{0\rightarrow1}(x_0,x_1),\label{propnomatch}
\ea 
as long as they are in the same orbit and therefore project to the same physical propagator
\ba
{}^+\mathbb{P}_1\,({}^-\mathbb{P}_0)^*\,\kappa_{0\rightarrow1}(x_0,x_1)={}^+\mathbb{P}_1\,({}^-\mathbb{P}_0)^*\,\kappa'_{0\rightarrow1}(x_0,x_1)=K_{0\rightarrow1}(x_0,x_1).
\ea
On the other hand, there will exist $\kappa_{0\rightarrow1}$ and $\kappa''_{0\rightarrow1}$ in different orbits such that 
\ba
K_{0\rightarrow1}={}^+\mathbb{P}_1\,({}^-\mathbb{P}_0)^*\,\kappa_{0\rightarrow1}(x_0,x_1)\neq{}^+\mathbb{P}_1\,({}^-\mathbb{P}_0)^*\,\kappa''_{0\rightarrow1}(x_0,x_1)=K'_{0\rightarrow1},
\ea
yields two inequivalent physical propagators, both of which satisfy all quantum pre-- and post--constraints. The choice of the orbit for the kinematical propagator (i.e.\ the choice between $\kappa_{0\rightarrow1}$ and $\kappa''_{0\rightarrow1}$) is dynamically restricted by the action $S_1$ of the move $0\rightarrow1$ since the physical propagator should be given in the form (\ref{prop}). Although it should be noted that even this restriction may, in general, not single out a unique physical propagator.

Given a choice of kinematical propagator permits us to write (\ref{proj01}) as an (improper) projection of the PIP-- at $n=0$ with ${}^+\mathbb{P}_1$,
\ba
{}^+\psi^{\rm phys}_1(x_1)={}^+\mathbb{P}_1\left\langle{}^-\mathbb{P}_0\,\kappa_{1\rightarrow0}\Big|\psi^{\rm kin}_0\right\rangle_{\rm kin}.\label{pstatepip}
\ea
%That is, one can qualitatively understand the propagator as $K_{0\rightarrow1}\sim |{}^+\psi^{\rm phys}_1\rangle\langle{}^-\phi^{\rm phys}_0|$.

Using that equation (\ref{proj01}) is essentially a PIP-- and that the constraints (in the projectors) are self-adjoint with respect to the KIP, entails that we may write
\ba
{}^+\psi^{\rm phys}_1(x_1)&=&\int dx_0\,K_{0\rightarrow1}(x_1,x_0)\,\psi^{\rm kin}_0(x_0)=\int dx_0\left({}^+\mathbb{P}_1\,({}^-\mathbb{P}_0)^*\,\kappa_{0\rightarrow1}(x_0,x_1)\right)\psi^{\rm kin}_0(x_0)\nn\\
&=&\int dx_0\,({}^+\mathbb{P}_1\,\left({}^-\mathbb{P}_0\,\kappa_{1\rightarrow0}(x_0,x_1)\right)^*)\,\psi^{\rm kin}_0(x_0)\nn\\
&=&\int dx_0\left({}^+\mathbb{P}_1\,\kappa_{0\rightarrow1}(x_0,x_1)\right){}^-\mathbb{P}_0\,\psi^{\rm kin}_0(x_0)\nn\\
&\underset{(\ref{prestates})}{=}&\int dx_0\left({}^+\mathbb{P}_1\,\kappa_{0\rightarrow1}(x_0,x_1)\right){}^-\psi^{\rm phys}_0(x_0)\nn\\
&=&\int dx_0\, K_{0\rightarrow1}^{f_+}(x_0,x_1)\,{}^-\psi^{\rm phys}_0(x_0),\label{conprop2}
\ea
where we have defined the {\it pre--fixed propagator}
\ba
K_{0\rightarrow1}^{f_+}(x_0,x_1):={}^+\mathbb{P}_1\,\kappa_{0\rightarrow1}(x_0,x_1).\label{pfprop}
\ea
We call this object {\it pre--fixed} because: (a) it solves the post--constraints at $n=1$, however, does not solve the pre--constraints at $n=0$, and (b) equation (\ref{conprop2}) yields a (formally) non-divergent map from ${}^-\ch^{\rm phys}_0$ to ${}^+\ch^{\rm phys}_1$---in contrast to
\ba\label{diverge}
\int dx_0\,K_{0\rightarrow1}{}^-\psi^{\rm phys}_0=\int dx_0{}^+\mathbb{P}_1\kappa_{0\rightarrow1}({}^-\mathbb{P}_0)^2\,\psi^{\rm kin}_0
\ea
which diverges because of the term $({}^-\mathbb{P}_0)^2$ or, equivalently, because of the integration over the non-compact $\cg^-_0$. $K_{0\rightarrow1}^{f_+}$ therefore regularizes or fixes the {\it a posteriori} free flow at $n=0$. By analogy, the {\it post--fixed propagator} is defined as $K_{0\rightarrow1}^{f_-}=({}^-\mathbb{P}_0)^*\,\kappa_{0\rightarrow1}(x_0,x_1)$. Notice the following relations between the pre-- and post--fixed propagator and reverse propagator
\ba
K^{f_+}_{1\rightarrow0}=\left(K^{f_-}_{0\rightarrow1}\right)^*,\q\q\q\q K^{f_-}_{1\rightarrow0}=\left(K^{f_+}_{0\rightarrow1}\right)^*.
\ea

Since (\ref{pstatepip}, \ref{conprop2}) is essentially a PIP$-$, we can expect for a large class of systems to be able to rewrite it in the same vein as (\ref{regpip}) as a regularized KIP
\ba\label{regprop}
{}^+\psi^{\rm phys}_1(x_1)&=&\int d\xi^-_0(x_0)\, K_{0\rightarrow1}(x_0,x_1){}^-\psi^{\rm phys}_0(x_0),
\ea
with regularized {\it pre--measure} $d\xi_0^-$. Again, below in section \ref{sec_lin}, we shall make this explicit under certain conditions.

In analogy to the regular global moves in section \ref{sec_reg}, we ought to require the fixed propagators to define a bijection between ${}^-\ch^{\rm phys}_0$ and ${}^+\ch^{\rm phys}_1$. Namely, in analogy to (\ref{cond1}, \ref{cond2}), we demand invertibility of the \underline{fixed} propagators,\footnote{Recall that the fixed propagators either satisfy pre-- or post--constraints but not both. Hence, the `fixed' expression $K^{f_+}_{0\rightarrow1}(K^{f_-}_{0\rightarrow1})^*$ in (\ref{coninvert}) will have a dependence on both {\it a priori} and {\it a posteriori free} parameters. In this way, one can produce the delta functions of also the free parameters on the right hand sides. For example, in the toy model of section \ref{sec_ex}, we shall see explicitly how the delta functions of free parameters arise via gauge fixing conditions in the fixed propagators.} 
\ba\label{coninvert}
\int\,dx_0\,K^{f_+}_{0\rightarrow1}(x_1,x_0)\left(K^{f_-}_{0\rightarrow1}(x'_1,x_0)\right)^*&=&\delta^{(N)}(x'_1-x_1),\nn\\
\int dx_1\left(K^{f_-}_{0\rightarrow1}(x_1,x_0)\right)^*K^{f_+}_{0\rightarrow1}(x_1,x'_0)&=&\delta^{(N)}(x'_0-x_0).
\ea 
(The analogous condition for $K_{0\rightarrow1}$ would necessarily diverge like (\ref{diverge}) due to the appearance of $({}^-\mathbb{P}_0)^2$ and $({}^+\mathbb{P}_1)^2$.) Notice that the invertibility conditions (\ref{coninvert}) can also be written in terms of (projected) PIPs (recall that $\kappa_{0\rightarrow1}$ is square integrable in $x_0$ and $x_1$):
\ba
{}^+\mathbb{P}_1\,\left\langle\kappa_{1\rightarrow0}(x'_1,x_0)\Big|{}^-\mathbb{P}_0\,\kappa_{1\rightarrow0}(x_0,x_1)\right\rangle_{\rm kin}&=&\delta^{(N)}(x'_1-x_1),\nn\\
{}^-\mathbb{P}_0\,\left\langle\kappa_{0\rightarrow1}(x_1,x_0)\Big|{}^+\mathbb{P}_1\,\kappa_{0\rightarrow1}(x'_0,x_1)\right\rangle_{\rm kin}&=&\delta^{(N)}(x'_0-x_0).
\ea

The situation in the quantum theory therefore bears some analogy to the classical situation: in the classical formalism, the Hamiltonian time evolution map from $T^*\cq_0$ to $T^*\cq_1$ (i) is defined only on and between the pre-- and post--constraint surfaces and (ii) is invertible only on the pre-- and post--constraint surfaces modulo the {\it a priori} and {\it a posteriori} free orbits, respectively. That is, it is invertible only on the `physical' space of propagating degrees of freedom for the move $0\rightarrow1$ \cite{Dittrich:2011ke,Dittrich:2013jaa}. In the quantum theory, the propagator assumes the role of the Hamiltonian time evolution map. As just seen, it is only invertible on and between ${}^-\ch^{\rm phys}_0$ and ${}^+\ch^{\rm phys}_1$.

In conclusion, in the quantum theory we find a whole set of maps and projectors for a constrained evolution move $0\rightarrow1$. For easy reference and completeness, we shall list them here:
\begin{itemize}
\item ${}^+\mathbb{P}_1:=\prod_I\delta({}^+\hat{C}^1_I)$ is an improper projector from $\ch^{\rm kin}_1$ to ${}^+\ch^{\rm phys}_1$,
\item ${}^-\mathbb{P}_0:=\prod_I\delta({}^-\hat{C}^0_I)$ is an improper projector from $\ch^{\rm kin}_0$ to ${}^-\ch^{\rm phys}_0$,
\item $P_{0\rightarrow1}:=\int dx_0\,K_{0\rightarrow1}$ is an improper projector from $\ch^{\rm kin}_0$ to ${}^+\ch^{\rm phys}_1$,
\item $P_{1\rightarrow0}:=\int dx_1\,K_{1\rightarrow0}$ is an improper projector from $\ch^{\rm kin}_1$ to ${}^-\ch^{\rm phys}_0$,
\item $U_{0\rightarrow1}:=\int dx_0\,K^{f_+}_{0\rightarrow1}$ is an invertible map from ${}^-\ch^{\rm phys}_0$ to ${}^+\ch^{\rm phys}_1$, with
\item $U_{1\rightarrow0}:=\int dx_1\,K^{f_+}_{1\rightarrow0}$ being the inverse map from ${}^+\ch^{\rm phys}_1$ to ${}^-\ch^{\rm phys}_0$.
\end{itemize}
These maps and projectors can be diagrammatically summarized as follows:
\begin{diagram}
\ch^{\rm kin}_0&&&&\ch^{\rm kin}_1\\
&\rdLine~{P_{0\rightarrow1}}&&\ldLine~{P_{1\rightarrow0}}&\\
\dTo^{{}^-\mathbb{P}_0}&&&&\dTo_{{}^+\mathbb{P}_1}\\
&\ldTo&&\rdTo&\\
{}^-\ch^{\rm phys}_0&&\rTo_{U_{0\rightarrow1}}&&{}^+\ch^{\rm phys}_1.
\end{diagram}

Lastly, equation (\ref{proj01}) immediately implies unitarity of the constrained move $0\rightarrow1$, expressed in terms of the PIP+ at $n=1$ and the PIP-- at $n=0$:
\ba
\langle{}^+\psi^{\rm phys}_1\Big|{}^+\phi^{\rm phys}_1\rangle_{\rm phys+}&=&\int dx_1(\psi^{\rm kin}_1(x_1))^*{}^+\phi^{\rm phys}_1(x_1)\nn\\
&=&\int dx_1(\psi^{\rm kin}_1(x_1))^*\int dx_0\,K_{0\rightarrow1}(x_0,x_1)\,\phi^{\rm kin}_0(x_0)\nn\\
&=&\int dx_0\left(\int dx_1K_{1\rightarrow0}(x_0,x_1)\,\psi^{\rm kin}_1(x_1)\right)^*\phi^{\rm kin}_0(x_0)\nn\\
&=&\int dx_0({}^-\psi^{\rm phys}_0(x_0))^*\phi^{\rm kin}_0(x_0)\nn\\
&=&\langle{}^-\psi^{\rm phys}_0\Big|{}^-\phi^{\rm phys}_0\rangle_{\rm phys-}.\label{unit}
\ea
%\ba
%\langle{}^+\psi^{\rm phys}_1\Big|{}^+\phi^{\rm phys}_1\rangle_{\rm phys+}&=&\int dx_1({}^+\psi^{\rm kin}_1(x_1))^*{}^+\phi^{\rm phys}_1(x_1)\nn\\
%&=&\int dx_1({}^+\psi^{\rm kin}_1(x_1))^*\int d\xi^-_0(x_0)K_{0\rightarrow1}(x_0,x_1){}^-\phi^{\rm phys}_0(x_0)\nn\\
%&=&\int d\xi^-_0(x_0)\left(\int dx_1K_{1\rightarrow0}(x_0,x_1){}^+\psi^{\rm kin}_1(x_1)\right)^*{}^-\phi^{\rm phys}_0(x_0)\nn\\
%&=&\int d\xi^-_0(x_0)({}^-\psi^{\rm phys}_0(x_0))^*{}^-\phi^{\rm phys}_0(x_0)\nn\\
%&=&\langle{}^-\psi^{\rm phys}_0\Big|{}^-\phi^{\rm phys}_0\rangle_{\rm phys-}.
%\ea
%Notice that this result is independent on which particular kinematical states are chosen, provided they are mapped to the same physical states under $P^+_{1\rightarrow1}$ or $P^-_{0\rightarrow0}$. That is, for a given set of physical states this result is unambiguous. 
The same result can similarly be shown in the other direction or by means of the pre-- and post--fixed propagators. 

The transition amplitudes for a constrained move $0\rightarrow1$ thus read
\ba
\langle{}^+\psi^{\rm phys}_1\big|\,U_{0\rightarrow1}\,{}^-\phi^{\rm phys}_0\rangle_{\rm phys+}=\int dx_1\,dx_0\,(\psi^{\rm kin}_1)^*\,K^{f_+}_{0\rightarrow1}\,{}^-\phi^{\rm phys}_0.\label{trans}
\ea

We shall find the formal constructions of this section explicitly realized for a class of actions in section \ref{sec_lin} and in the toy model in section \ref{sec_ex}.

\subsection{Explicit construction for systems with constraints linear in the momenta}\label{sec_lin}

%Next, let us consider how we may express the PIP in a more practical fashion directly on $\ch^{\rm phys}$. 
%{\bf[lemma for especially strong conditions]}

In this section, we shall render some of the formal constructions of the previous section explicit---for systems subject to constraints linear in the momenta. Clearly, linearity in the momenta is a rather strong restriction. However, this restriction permits us to derive concrete results beyond the formal level, which generally is notoriously difficult with group averaging techniques. Furthermore, the following results will be important for
\begin{itemize}
\item {\it Systems with temporally varying discretization:} Specifically, the following results will be relevant for cylindrical consistency and evolving Hilbert spaces in section \ref{sec_evol} below. 
\item {\it Local time evolution moves:} These are discussed in the companion paper \cite{Hoehn:2014wwa}.
\item {\it Quadratic discrete actions:} Constraints resulting from quadratic discrete actions are necessarily linear in the momenta. This is studied in detail in \cite{Hoehn:2014aoa}. 
\end{itemize}

Classically, constraints linear in the momenta must be of the form (for notational clarity, we shall omit time step indices and the distinction between pre-- and post--objects for part of this section)
\ba
{C}_I=f_{Ii}({x})\,{p}_i-V_I({x}),\nn
\ea
where $f_{Ii}$, $I=1,\ldots,k$, $i=1,\ldots,\dim\cq$ are coefficient functions of a real matrix. With a suitable variable transformation, we can find new canonical pairs $(\lambda^I(x^i),\pi_I:=f_{Ii}({x})\,{p}_i)$, $I=1,\ldots,k$, and $(x^\alpha(x^i),\pi_\alpha(x^i,p_i))$, $\alpha=1,\ldots,\dim\cq-k$.\footnote{To this end, suitably complete the $(k\times\dim\cq)$ matrix $f_{Ii}$ to an invertible $(\dim\cq\times\dim\cq)$ matrix $f_{\Gamma i}$, where $\Gamma=1,\ldots,\dim\cq$ is split into two index sets $I=1,\ldots,k$ and $\alpha=1,\ldots,\dim\cq-k$. Then integrate the equations $\f{\p\lambda^I(x^i)}{\p x^j}=(f^{-1})_{Ij}(x^i)$ and $\f{\p x^\alpha(x^i)}{\p x^j}=(f^{-1})_{\alpha j}(x^{i})$. The conjugate momentum to $x^\alpha$ is $\pi_\alpha=f_{\alpha i}(x^i)p_i$.} Setting $V_I(\lambda^I,x^\alpha)=\f{\p S(\lambda^J,x^\alpha)}{\p \lambda^I}$, one can put the constraints into the form
\ba
{C}_I={p}_I-\f{\p S(\lambda^J,x^\alpha)}{\p \lambda^I}.\label{lemconform}
\ea
Without loss of generality, we therefore assume the constraints to be of the simpler form (\ref{lemconform}). Constraints of this form are abelian and admit global conditions $G_K(\lambda^I,x^\alpha)=0$ on the configuration variables alone which fix the flow of the constraints. We shall make use of such parameter fixing conditions in order to regularize otherwise divergent quantities.

\begin{lem}\label{lem1}
Consider a set of $k$ constraints $\hat{C}_I=\hat{p}_I-\f{\p S(x^J,x^\alpha)}{\p x^I}$ and $k$ global gauge fixing conditions $G_K(\lambda^I,x^\alpha)=0$ on the configuration variables. The following identity holds
\ba
(2\pi)^k\prod_{I=1}^k\delta(\hat{C}_I)\Big|\det\left([\hat{G}_M,\hat{C}_N]\right)\Big|\prod_{K=1}^k\delta(\hat{G}_K)\prod_{J=1}^k\delta(\hat{C}_J)\,\psi^{\rm kin}(x)=\psi^{\rm phys}(x).
\ea
\end{lem}

\begin{proof}
The proof is provided in appendix \ref{app}.
\end{proof}

For this special case, no factor ordering ambiguities occur. The identity may, however, also hold for a larger class of system for which $[\hat{G}_K,\hat{C}_I]$ is a function of the configuration variables only.

Hence, $\psi^{\rm kin}$ and $\tilde{\psi}^{\rm kin}:=\big|\det([\hat{G}_K,\hat{C}_I])\big|\prod_K\delta(\hat{G}_K)\,\psi^{\rm phys}$ are elements of the same orbit generated by the $\hat{C}_I$ and are mapped to the same $\psi^{\rm phys}\in\ch^{\rm phys}$ under $\prod_I\delta(\hat{C}_I)$.

Thanks to lemma \ref{lem1}, we can now explicitly write the PIP (\ref{pip}) purely in terms of physical states, i.e.\ in the form (\ref{regpip}).
\begin{lem}\label{lem2}
Let the conditions of Lemma \ref{lem1} be true. In the position representation, the PIP (\ref{pip}) then takes the form
\ba
\left\langle\psi^{\rm phys}\Big|\phi^{\rm phys}\right\rangle_{\rm phys}&=&\int_{\cq}d\xi(\lambda^I,x^\alpha)\,\left(\psi^{\rm phys}(\lambda^I,x^\alpha)\right)^*\phi^{\rm phys}(\lambda^I,x^\alpha)\nn\\
&=&(2\pi\hbar)^k\int_{\cq/\cg}\prod_\alpha dx_\alpha \,\left(\psi^{\rm phys}(\lambda^I,x^\alpha)\right)^*\phi^{\rm phys}(\lambda^I,x^\alpha),\label{lem2eq}
\ea
with regularized Faddeev-Popov-measure
\ba
d\xi(\lambda^I,x^\alpha):=(2\pi)^k\prod_{I,\alpha}d\lambda^{I}dx^\alpha\Big|\det\left([\hat{G}_K,\hat{C}_I]\right)\Big|\prod_{K=1}^k\delta(\hat{G}_K(\lambda^I,x^\alpha)).
\ea
The PIP is independent of the particular choice of gauge and the gauge conditions $G_K(\lambda^I,x^\alpha)=0$ employed to impose it.
\end{lem}
\begin{proof}
The proof is given in appendix \ref{app}.
\end{proof}

One may be surprised about the normalization factor $(2\pi\hbar)^k$ in (\ref{lem2eq}). Ultimately, it arises because the PIP in its original form (\ref{pip}) only involves one action of $\prod_K\delta(\hat{C}_K)$, which contains a normalization factor $(2\pi\hbar)^{-k}$ from (\ref{proj}). On the other hand, the PIP in its regularized form (\ref{lem2eq}) contains two actions of $\prod_K\delta(\hat{C}_K)$. Clearly, this normalization factor can be absorbed into the states, however, we keep it here for clarity.

In this case---and thus specifically for quadratic discrete actions---the PIP obtained via group-averaging coincides with the PIP constructed as a Faddeev-Popov-regularized KIP applied to physical states (the latter was also discussed in \cite{Henneaux:1992ig}). The Faddeev-Popov-determinant reads $\Delta_{FP}=\big|\det([\hat{G}_K,\hat{C}_I]/\hbar)\big|$.

In the same fashion, we can now explicitly write the invertible map from ${}^-\ch^{\rm phys}_0$ to ${}^+\ch^{\rm phys}_1$ (\ref{conprop2}, \ref{regprop}) in terms of a Faddeev-Popov regularized KIP:\footnote{We return to employing indices for the different time steps and a distinction between post-- and pre--objects. Recall that the {\it a posteriori} free parameters at $n=0$ are denoted by $\mu_0^I$, while the {\it a priori} free parameters at $n=1$ are denoted by $\lambda^I_1$.}

\begin{lem}\label{lem3}
Let the conditions of Lemma \ref{lem1} be true for a global move $0\rightarrow1$. In the position representation, the invertible map $U_{0\rightarrow1}:{}^-\ch^{\rm phys}_0\rightarrow{}^+\ch^{\rm phys}_1$ is given by
\ba
{}^+\psi^{\rm phys}_1(\lambda_1^J,x^\beta_1)%&=&\int_{\cq_0} \prod_{I,\alpha}d\mu_0^{I}\,dx_0^\alpha\, K_{0\rightarrow1}^{f_+}(\mu_0^I,x_0^\alpha;\lambda_1^J,x^\beta_1){}^-\psi^{\rm phys}_0(\mu_0^I,x_0^\alpha)\nn\\
&=&\int_{\cq_0} d\xi^-_0(\mu^I_0,x^\alpha_0)\, K_{0\rightarrow1}(\mu_0^I,x_0^\alpha;\lambda_1^J,x^\beta_1){}^-\psi^{\rm phys}_0(\mu_0^I,x_0^\alpha)\nn\\
&=&(2\pi\hbar)^{k_-}\int_{\cq_0/\cg^-_0}\prod_\alpha dx_\alpha \,K_{0\rightarrow1}(\mu_0^I,x_0^\alpha;\lambda_1^J,x^\beta_1){}^-\psi^{\rm phys}_0(\mu_0^I,x_0^\alpha),\nn
\ea
where the Faddeev-Popov regularized pre--measure reads
\ba
d\xi^-_0(\mu_0^I,x_0^\alpha):=(2\pi)^{k_-}\prod_{I,\alpha}d\mu_0^{I}\,dx_0^\alpha\Big|\det\left([{}^-\hat{G}^0_K,{}^-\hat{C}^0_I]\right)\Big|\prod_{K=1}^{k_-}\delta({}^-\hat{G}^0_K(\mu_0^I,x_0^\alpha)).\nn
\ea
$U_{0\rightarrow1}$ is independent of the \emph{a posteriori} free $\mu^I_0$ and the particular choice of the conditions ${}^-G_K(\mu_0^I,x^\alpha_0)=0$ fixing the \emph{a posteriori} free orbit $\cg^-_0$.
\end{lem}

\begin{proof}
The proof is presented in appendix \ref{app}.
\end{proof}
Clearly, the analogous result holds for the inverse map.

We note that both the pre--fixed propagator $K_{0\rightarrow1}^{f_+}$ and $\big|\!\det([{}^-\hat{G}^0_M,{}^-\hat{C}^0_N])\big|\prod_{K=1}^k\delta({}^-\hat{G}^0_K)\, K_{0\rightarrow1}$ are mapped to $K_{0\rightarrow1}$ under $\prod_{I=1}^k(\delta({}^-\hat{C}^0_I))^\dag$.

%\begin{Example}
%\emph{quadratic discrete actions} {\bf[still to be written!]}
%
%\emph{then compare to Miller's results. Miller's results seem to work in this case.}
%\end{Example}

\section{Temporally varying discretization, evolving Hilbert spaces and cylindrical consistency}\label{sec_evol}

Hitherto, we have only worried about systems with a constant configuration space dimension and thus kinematical Hilbert spaces $\ch_n^{\rm kin}=L^2(\cq_n,dx_n)$ such that $\cq_n\simeq\cq_{n'}$, $\forall\,n,n'$. As a next step, we generalize the previous construction to a quantum formalism for systems with temporally varying discretization and, consequently, temporally varying numbers of degrees of freedom. In particular, we now have to deal with the fact that $\dim\cq_n\neq\dim\cq_{n+1}$ is expressly allowed. Restricting again to $\cq_n\simeq\mathbb{R}^{N_n}$, the kinematical Hilbert spaces $\ch^{\rm kin}_n=L^2(\mathbb{R}^{N_n},dx_n)$ {`evolve'} from time step to time step in the sense that the classical configuration manifold over which the square integrable functions are defined changes in discrete time. We emphasize that this notion of `evolving' Hilbert spaces does {\it not} refer to a temporally varying Hilbert space dimension but to the temporally varying number of degrees of freedom of the underlying discretization which are necessary to describe a quantum state. (Of course, in a strictly mathematical sense, the kinematical Hilbert spaces do not evolve because $L^2(\mathbb{R}^{N_n},dx_n)$ and $L^2(\mathbb{R}^{N_m},dx_m)$ are isometrically isomorphic even if $N_n\neq N_m$ \cite{conwaybook}.) But what does this imply for the {\it physical} Hilbert spaces?

\subsection{Evolving phase spaces}

The corresponding classical formalism for evolving phase spaces has been developed in detail in \cite{Dittrich:2011ke,Dittrich:2013jaa,Hoehn:2014aoa}. The trick is to suitably extend the configuration spaces at both steps until they are of equal dimension. More specifically, for `old' variables $x^{\fo}_n$ that occur at step $n$, however, which have no counterpart at step $n+1$, extend $\cq_{n+1}$ by a suitable configuration manifold of appropriate dimension $\cq^{ext}_n$---coordinatized by `the missing' {\it a priori free} $x^{\fo}_{n+1}:=\lambda^\fo_{n+1}$---to an extended configuration manifold $\overline{\cq}_{n+1}=\cq_{n+1}\times\cq^{ext}_{n+1}$. Do the analogue for `new' variables $x^{\fn}_{n+1}$ which have no counterpart at step $n$. At this stage, $\dim\overline{\cq}_n=\dim\overline{\cq}_{n+1}$ and the formalism for constrained global moves for systems with constant configuration space dimension applies. For every newly introduced $\lambda^{\fo}_{n+1},\mu^{\fn}_n:=x^\fn_n$, one obtains a trivial post--constraint at $n+1$ and pre--constraint at $n$, respectively,
\ba
{}^+C^{n+1}_{\fo}=p^{n+1}_{\fo}=\f{\p S_n(x_n,x_{n+1})}{\p \lambda^{\fo}_{n+1}}=0,\q\q\q\q{}^-C^n_{\fn}=p^n_{\fn}=-\f{\p S_n(x_n,x_{n+1})}{\p \mu^{\fn}_n}=0,\label{trivpcons}
\ea
because $S_n$ depends on neither $\lambda^{\fo}_{n+1}$ nor $\mu^{\fn}_n$. This is also the reason why the new constraints (\ref{trivpcons}) will necessarily be first class and why the $\lambda^{\fo}_{n+1}$ and $\mu^{\fn}_n$ are proper gauge parameters which are both {\it a priori} and {\it a posteriori} free. Upon configuration space extension, the total number of post--constraints at step $n+1$ and pre--constraints at step $n$ (i.e., not only those in (\ref{trivpcons}) which are a consequence of the extension itself) coincide. Furthermore, the space of pre--observables at step $n$ is isomorphic to the space of post--observables at step $n+1$. Pre/post--observables at step $n$ are functions on the pre/post--constraint surface which are invariant under the action of the pre/post--constraints at $n$. The space of invariant functions therefore does {\it not} evolve in a given move $n\rightarrow n+1$, but is actually associated to it. However, different moves in the evolution are associated to different spaces of invariant functions and thus the number of propagating degrees of freedom is in general move dependent. For details on the classical situation, see \cite{Dittrich:2013jaa}.

\subsection{Evolving Hilbert spaces}

Our goal is to translate the same trick into the quantum theory, i.e.\ to suitably `extend'\footnote{We write `extend' in quotation marks because this extension refers to the configuration manifolds over which the square integrable functions are defined rather than a genuine Hilbert space extension.}  Hilbert spaces and then employ the quantum formalism for systems with constant configuration space dimension of the previous section on the resulting Hilbert spaces. This is not hard to do. Firstly, we adopt the configuration space extension from the classical formalism and (in the position representation) `extend' the corresponding kinematical Hilbert spaces as follows
\ba
\ch^{\rm kin}_n&=&\,\,\,\,L^2(\mathbb{R}^{N_n},dx_n)\q\,\,\q\underset{\text{extend}}{\longrightarrow}\q\overline{\ch}^{\rm kin}_n\,\,=L^2\left(\mathbb{R}^{N_n}\times\mathbb{R}^{N^{ext}_n},dx_n\prod_\fn d\mu^{\fn}_n\right)\nn\\
\ch^{\rm kin}_{n+1}&=&L^2(\mathbb{R}^{N_{n+1}},dx_{n+1})\,\q\underset{\text{extend}}{\longrightarrow}\q\overline{\ch}^{\rm kin}_{n+1}=L^2\left(\mathbb{R}^{N_{n+1}}\times\mathbb{R}^{N^{ext}_{n+1}},dx_{n+1}\prod_\fo d\lambda^{\fo}_{n+1}\right),\nn
\ea
where $\dim\overline{\cq}_n=N_n+N^{ext}_n=N_{n+1}+N^{ext}_{n+1}=\dim\overline{\cq}_{n+1}$. Elements of $\overline{\ch}^{\rm kin}$ are denoted by $\overline{\psi}^{\rm kin}$. The constraint quantization then proceeds on these `extended' kinematical Hilbert spaces. 

Next, we have to worry about two aspects : (A) apart from already existing non-trivial constraints, implement the new pre-- and post--constraints (\ref{trivpcons}) in the quantum theory and (B) ensure that the PIP does {\it not} depend on the auxiliary Hilbert space extension. 

~\\
(A) The first aspect is, of course, rather trivial: implementing (\ref{trivpcons}) in the standard position representation as derivative operators, immediately implies that physical post--states at step $n+1$ are independent of $\lambda^{{\fo}}_{n+1}$, while physical pre--states at step $n$ are independent of $\mu^{\fn}_n$---as one would expect. That is, physical states are constant along the newly added configuration space directions. The quantum pre-- and post--constraints corresponding to (\ref{trivpcons}) are trivially abelian, such that their improper projectors $\delta(\hat{p})$ commute with all other projectors (including those associated to non-trivial constraints) and no ordering ambiguities arise. For later use, we shall spell out the group averaging method applied to the post--constraints in (\ref{trivpcons}),
\ba
\prod_{\fo}\delta(\hat{p}^{n+1}_{\fo})\,\overline{\psi}^{\rm kin}_{n+1}(x_{n+1},\lambda^{\fo}_{n+1})&=&\left(\f{1}{2\pi\hbar}\right)^{N^{ext}_{n+1}}\int\prod_{\fo}\,ds^{\fo}e^{is^{\fo}\hat{p}^{n+1}_{\fo}/\hbar}\,\,\overline{\psi}^{\rm kin}_{n+1}(x_{n+1},\lambda^{\fo}_{n+1})\nn\\
&=&\left(\f{1}{2\pi\hbar}\right)^{N^{ext}_{n+1}}\int\prod_{\fo}ds^{\fo}\,\overline{\psi}^{\rm kin}_{n+1}(x_{n+1},\lambda^{\fo}_{n+1}+s^{\fo})\nn\\
&=&\left(\f{1}{2\pi\hbar}\right)^{N^{ext}_{n+1}}\psi^{\rm kin}_{n+1}(x_{n+1}).\label{trivga}
\ea
(We assume $\overline{\psi}^{\rm kin}$ to be integrable in the auxiliary variables $\lambda^{\fo}_{n+1}$.) 

This suggests that physical states are {\it cylindrical functions}. The notion of {\it cylindrical functions} originates in the theory of integration on infinite dimensional spaces but can, in particular, also be applied to finite dimensional spaces. Given a vector space $V$ and a finite dimensional subspace $W\subset V$, a {\it cylindrical function} with respect to $W$ is a function $\psi$ on $W$ which is constant along all directions orthogonal to $W$.\footnote{Hence, if the support of $\psi$ on $W$ is compact, its support on $V$ is a `cylinder'.} Clearly, if $\psi$ is cylindrical with respect to $W$, it will also be cylindrical with respect to $W'\supset W$. In our discrete evolution, the configuration space at an arbitrary time step $n$ can, of course, be arbitrarily extended compared to the original, unextended $\cq_{n}$ without affecting physical states. This implies that physical states, as functions on the extended configuration spaces $\overline{\cq}_{n}$, are {\it cylindrical functions} with respect to the unextended configuration space $\cq_{n}\simeq\mathbb{R}^{N_{n}}$ (and any finite dimensional space containing $\cq_{n}$). Cylindrical functions are a pivotal tool in loop quantum gravity \cite{Ashtekar:2004eh,Thiemann:2007zz,Dittrich:2012jq}, but can also be applied to scalar and gauge quantum field theories \cite{Ashtekar:1994bp}. 

~\\
(B) Clearly, the PIP must not depend on whether the extended or unextended kinematical Hilbert space is used to construct it. That is, the PIP ought to be independent of the auxiliary variables. This suggests a form of {\it cylindrical consistency} for physical states. Cylindrical consistency for cylindrical functions on a vector space $V$ is a consistency condition on the measure on $V$. Given two functions $\psi,\psi'$ which are cylindrical with respect to $W$ and $W'$, respectively, {\it cylindrical consistency} requires that (1) integration of $\psi$ should not depend on the choice of $W\subset V$ as long as $\psi$ is cylindrical with respect to $W$, and (2) the inner product of $\psi,\psi'$ as an integration should not depend on the choice of $W''\supset W,W'$, provided both $\psi,\psi'$ are cylindrical with respect to $W''$. A measure on $V$ which satisfies this condition is called {\it cylindrical}. It can often be uniquely determined and usually turns out to be probabilistic \cite{Ashtekar:1994bp,Ashtekar:2004eh,Thiemann:2007zz}.

Indeed, cylindrical consistency is explicitly realized for the PIP and one can easily determine the cylindrical measure. To this end, notice that lemma \ref{lem2} applies to the new post--constraints in (\ref{trivpcons}) and that $[\hat{G}^{n+1}_{\fo}(\lambda^{\fo}_{n+1}),\hat{p}^{n+1}_{\fo'}]=i\hbar\,\f{\p G^{n+1}_{\fo}}{\p \lambda^{\fo'}_{n+1}}$. It follows
\ba
\langle\overline{\psi}^{\rm kin}_{n+1}\big|{}^+{\phi}^{\rm phys}_{n+1}\rangle_{\overline{\ch}^{\rm kin}_{n+1}}\!\!\!\!\!\!&=&\int_{\overline{\cq}_{n+1}}dx_{n+1}\prod_\fo d\lambda^{\fo}_{n+1}\left(\overline{\psi}^{\rm kin}_{n+1}(x_{n+1},\lambda^{\fo}_{n+1})\right)^*{}^+{\phi}^{\rm phys}_{n+1}(x_{n+1})\nn\\
&\underset{\text{\tiny Lemma \ref{lem2}}}{=}&(2\pi\hbar)^{N^{ext}_{n+1}}\int_{\overline{\cq}_{n+1}}dx_{n+1}\prod_\fo d\lambda^{\fo}_{n+1}\left(\overline{\psi}^{\rm kin}_{n+1}(x_{n+1},\lambda^{\fo}_{n+1})\right)^*\prod_{\fo'}\delta(\hat{p}^{n+1}_{\fo'})\nn\\
&&\q\times\Big|\det\left([\hat{G}^{n+1}_{\fo},\hat{p}^{n+1}_{\fo'}]/\hbar\right)\Big|\prod_{\fo''}\delta(\hat{G}^{n+1}_{\fo''})\,{}^+{\phi}^{\rm phys}_{n+1}(x_{n+1})\nn\\
&=&\!\!\!\!\!\!(2\pi\hbar)^{N^{ext}_{n+1}}\int_{\overline{\cq}_{n+1}}dx_{n+1}\prod_\fo d\lambda^{\fo}_{n+1}\left(\prod_{\fo'}\delta(\hat{p}^{n+1}_{\fo'})\,\overline{\psi}^{\rm kin}_{n+1}(x_{n+1},\lambda^{\fo}_{n+1})\right)^*\nn\\
&&\q\times\Big|\det\left(\f{\p G^{n+1}_{\fo}}{\p \lambda^{\fo'}_{n+1}}\right)\Big|\prod_{\fo''}\delta(\hat{G}^{n+1}_{\fo''})\,{}^+\phi^{\rm phys}_{n+1}(x_{n+1})\nn\\
&\underset{(\ref{trivga})}{=}&\int_{\overline{\cq}_{n+1}}dx_{n+1}\,d\xi_{n+1}^+(\lambda^{\fo}_{n+1})\left({\psi}^{\rm kin}_{n+1}(x_{n+1})\right)^*\,{}^+\phi^{\rm phys}_{n+1}(x_{n+1})\nn\\
&=&\langle{\psi}^{\rm kin}_{n+1}\big|{}^+\phi^{\rm phys}_{n+1}\rangle_{{\ch}^{\rm kin}_{n+1}},\label{cylpip}
\ea
with cylindrical measure for the auxiliary variables
\ba
d\xi^+_{n+1}(\lambda^{\fo}_{n+1})=\prod_\fo d\lambda^{\fo}_{n+1}\Big|\det\left(\f{\p G^{n+1}_{\fo}}{\p \lambda^{\fo'}_{n+1}}\right)\Big|\,\delta(\hat{G}^{n+1}_{\fo}).\label{cyl}
\ea
This measure is {\it probabilistic} $\int_{\cq^{ext}_{n+1}}d\xi^+_{n+1}=1$.\footnote{In contrast to (\ref{lem2eq}), the measure (\ref{cyl}) does not contain any factors of $(2\pi)$. This is because in (\ref{cylpip}) we have employed the post--physical states ${}^+\phi^{\rm phys}_{n+1}$ with respect to the original unextended configuration space $\cq_{n+1}$. Had we instead directly constructed post--physical states ${}^+\overline{\phi}^{\rm phys}_{n+1}$ from $\overline{\cq}_{n+1}$, we would have ${}^+\phi^{\rm phys}_{n+1}=(2\pi\hbar)^{N^{ext}_{n+1}}{}^+\overline{\phi}^{\rm phys}_{n+1}$, in analogy to (\ref{trivga}). In this case the measure would again be of the form presented in lemma \ref{lem2}.} The analogous result holds for pre--physical states. 

~\\
That is to say, {\it physical states at a given step $n+1$ are cylindrical functions on $\overline{\cq}_{n+1}$} with respect to $\cq_{n+1}$ and the PIP is cylindrically consistent. So long as the measure of the auxiliary variables is probabilistic in the form (\ref{cyl}), the PIP (and any integration of physical states) will not depend on the choice of extended configuration manifold.

In this manner, we are able to consistently transform global evolution moves with a time varying discretization and varying number of variables on evolving (kinematical) Hilbert spaces, to a scenario of constrained global evolution moves with constant configuration space dimension. %Consistency is ensured by the fact that physical states are cylindrical functions on the extended configuration spaces. 
At this stage, all the results of the previous section as regards PIPs, propagators and unitarity apply. Thereby, we arrive at a consistent quantum formalism for systems with time varying numbers of degrees of freedom.  

We emphasize that, in analogy to the classical formalism mentioned above, the pre--physical Hilbert space ${}^-\ch^{\rm phys}_n$ at time step $n$ in a {\it fixed} quantum evolution move $n\rightarrow n+1$ is {\it dynamically} isomorphic to the post--physical Hilbert space ${}^+\ch^{\rm phys}_{n+1}$ at $n+1$, despite the `evolving' kinematical Hilbert spaces. Namely, the results of the previous section tell us that, upon suitable configuration space extension, the map
\ba
{}^+\psi^{\rm phys}_{n+1}=U_{n\rightarrow n+1}{}^-\psi^{\rm phys}_{n}=\int_{\overline{\cq}_n}dx_n\prod_\fn d\mu^{\fn}_n\,K^{f_+}_{n\rightarrow n+1}{}^-\psi^{\rm phys}_n\nn
\ea
defines a unitary bijection between ${}^+\ch^{\rm phys}_{n+1}$ and ${}^-\ch^{\rm phys}_n$. However, as we shall see in section \ref{sec_effcon}, the invertible map $U_{n\rightarrow n+1}$ depends strongly on the move $n\rightarrow n+1$---in analogy to the classical case \cite{Dittrich:2013jaa}. In this sense one may speak of physical Hilbert spaces that `evolve' from evolution move to evolution move.

%: classically, for a fixed $n$, the number of constraints associated to the move $n\rightarrow m:=n+1$ depends on the choice of the next time step $m$. For a further step $m'>m$ the numbers of constraints involved in the move $n\rightarrow m'$ is higher or equal to the number of constraints of the move $n\rightarrow m$ \cite{Dittrich:2013jaa,ph1}. Physically this happens because the number of degrees of freedom `propagating' from the initial $n$ to the final time $m$ can only remain constant or decrease with growing $m$ in a time varying discretization. In the quantum theory this means that $U_{n\rightarrow m}$ and $U_{n\rightarrow m'>m}$ will generally be subject to different constraints such that different moves will, in general, be associated to different physical Hilbert spaces (see section \ref{sec_effcon}). Thus, in this sense one may speak of physical Hilbert spaces that `evolve' from evolution move to evolution move. 

%The companion paper \cite{Hoehn:2014wwa} explores how {\it local} (coarse graining) quantum evolution moves introduce non-unitarity into the evolution. These moves do not map between different global time steps but locally evolve a given time step forward by changing its discretization. This then explains the change of $U_{n\rightarrow m}$ into $U_{n\rightarrow m'>m}$ in the course of evolution.

\begin{Example}
\emph{Let us examine a trivial example for cylindrical consistency which, nonetheless, illustrates the general features of Hilbert space `extensions'. Namely, consider a regular (non-constrained) global evolution move $0\rightarrow1$ as described in section \ref{sec_reg} such that $\cq_0\simeq\cq_1\simeq\mathbb{R}$. Nothing (except simplicity) stops us from artificially extending the configuration spaces at both steps to $\overline{\cq}_{0,1}\simeq\mathbb{R}\times\mathbb{R}$, where the additional dimension at $n=0$ is coordinatized by the {\it a posteriori} free $\mu_0$ and at $n=1$ by the {\it a priori} free $\lambda_1$. This leads to a pre--constraint $p^0_\mu=0$ and a post--constraint $p^1_\lambda=0$ whose implementation in the quantum theory implies nothing but the independence of ${}^-\psi^{\rm phys}_0$ on $\mu_0$ and of ${}^+\psi^{\rm phys}_1$ on $\lambda_1$. That is, the physical states are simply the original states of the regular (unextended) move $0\rightarrow1$.}

\emph{As gauge condition at $n=1$ we may simply choose $\lambda_1=c_1$ such that $\Delta_{FP}^1=1$. Using Lemma \ref{lem2}, the PIP+ at $n=1$ is given by the Faddeev-Popov regularized KIP
\ba
\langle{}^+\psi^{\rm phys}_1\big|{}^+\phi^{\rm phys}_1\rangle_{\rm phys+}&=&\int_{\mathbb{R}\times\mathbb{R}}dx_1\,d\lambda_1\,\delta(\lambda_1-c_1)\left({}^+\psi^{\rm phys}_1(x_1)\right)^*{}^+\phi^{\rm phys}_1(x_1)\nn\\
&=&\int_{\mathbb{R}}dx_1\left({}^+\psi^{\rm phys}_1(x_1)\right)^*{}^+\phi^{\rm phys}_1(x_1),\nn
\ea
which is nothing but the original inner product at $n=1$ of the regular move $0\rightarrow1$. The analogous result holds for the PIP-- at $n=0$. The physical states are thus cylindrical functions on $\overline{\cq}$ (which can be arbitrarily extended) with respect to $\cq_{0,1}\simeq\mathbb{R}$ and we have cylindrical consistency with probabilistic measure $d\xi^+_1=d\lambda_1\,\delta(\lambda_1-c_1)$.}

\emph{The propagator of the regular move $K_{0\rightarrow1}=M_{0\rightarrow1}\,e^{iS_1/\hbar}$ obviously must be Faddeev-Popov regularized as in Lemma \ref{lem3} when employed on the extended system. As pre--and post--fixed propagators we can choose 
\ba
K^{f_+}_{0\rightarrow1}=\delta(\mu_0-c_0)\,K_{0\rightarrow1},\q\q\q\q K^{f_-}_{0\rightarrow1}=\delta(\lambda_1-c_1)\,K_{0\rightarrow1},
\ea
which just cancel integration over the auxiliary variables. Evidently, these fixed propagators yield a unitary evolution in terms of the PIP and define an invertible map between ${}^-\ch^{\rm phys}_0$ and ${}^+\ch^{\rm phys}_1$---which are simply the original Hilbert spaces of the regular move $0\rightarrow1$.}
\end{Example}

In section \ref{sec_ex}, we shall study in depth a less trivial example featuring cylindrical consistency.

\section{Composition of constrained moves and divergences}\label{sec_fullcomp}

Consider a composition of moves $0\rightarrow1$ and $1\rightarrow2$. Classically, the post--constraints at $n=1$ are automatically satisfied by the move $0\rightarrow1$, however, constitute non-trivial constraints on the move $1\rightarrow2$. Conversely, the pre--constraints at $n=1$ establish non-trivial conditions on $0\rightarrow1$ \cite{Dittrich:2011ke,Dittrich:2013jaa}. When solving the equations of motion, i.e.\ matching the symplectic structures and imposing {\it both} pre-- and post--constraints at $n=1$, the following possibilities arise \cite{Dittrich:2013jaa}:
\begin{itemize}
\item[(a)] A subset of the pre--constraints (in possibly rewritten form) coincides with a subset of the post--constraints. We shall enumerate these constraints by $a=1,\ldots,A$, i.e.\ $C^1_a:={}^+C^1_a={}^-C^1_a$. In \cite{Dittrich:2013jaa} it was shown that these constraints generate genuine gauge symmetries. These will ultimately be responsible for divergences in the path integral.
\item[(b1)] A subset of the pre--constraints is independent of the post--constraints, however, is first class with respect to the full set of pre-- and post--constraints. We denote these pre--constraints by ${}^-C^1_{b_-}$, $b_-=1,\ldots,B_-$. These constraints do {\it not} generate symmetries and only arise if there exist degrees of freedom of the discretization at $n=1$ which are dynamically relevant (i.e.\ propagating) for the move $0\rightarrow1$ but not for $1\rightarrow2$ \cite{Dittrich:2013jaa}. This can only occur for a temporally varying discretization. The arguments put forward in \cite{Dittrich:2013xwa} suggest to interpret the ${}^-C^1_{b_-}$ as non-trivial coarse graining conditions on the move $0\rightarrow1$ which render degrees of freedom finer than a coarse graining scale set by the move $1\rightarrow2$ dynamically irrelevant. As we shall see shortly, these constraints will lead to non-unitary projections of physical Hilbert spaces and Dirac observables.
\item[(b2)] A subset of the post--constraints is independent of the pre--constraints, however, is first class with respect to the full set of pre-- and post--constraints. We denote these post--constraints by ${}^+C^1_{b_+}$, $b_+=1,\ldots,B_+$. Note that in general $B_-\neq B_+$. These constraints are the time-reverse of the ${}^-C^1_{b_-}$ and as such do {\it not} generate symmetries \cite{Dittrich:2013jaa}, but, instead, establish non-trivial coarse graining conditions on the move $1\rightarrow2$ in the sense of (b1). Equivalently, the ${}^+C^1_{b_+}$ ensure that the coarser (dynamical) data of the move $0\rightarrow1$ can be consistently represented and propagated with the finer (dynamical) data of the move $1\rightarrow2$. This will lead to non-unitary projections of Hilbert spaces.
\item[(c)] A subset of the pre--constraints (post--constraints) is independent of the post--constraints (pre--constraints), fixes the {\it a priori} free flows of the latter and is thereby second class.
\item[(d)] The full set of pre-- and post--constraints cannot be simultaneously implemented and the dynamics is inconsistent. 
\end{itemize}
In the sequel, we shall assume case (d) not to occur. Furthermore, it is standard practice to solve the second class constraints (i.e., here case (c)) at the classical level. Subsequently, one can attempt to quantize the Dirac bracket and find a quantum representation of a remaining set of independent variables, although there exist other methods \cite{Henneaux:1992ig}. We shall thus assume the second class constraints at $n=1$ to be classically solved. We emphasize, however, that in the discrete evolution second class constraints at step $n=1$ only arise when matching the symplectic structures at $n=1$ of the moves $0\rightarrow1$ and $1\rightarrow2$. That is, from considering the move $0\rightarrow1$ alone one will not know which of the post--constraints at $n=1$ will become second class upon composition with the move $1\rightarrow2$. In conclusion, at $n=1$ we split the remaining (first class) set of pre--constraints ${}^-C^1_I$, $I=1,\dots,k_-$ and post--constraints ${}^+C^1_J$, $J=1,\dots,k_+$ (in general, $k_-\neq k_+$) into the following three subsets $\{C^1_a\}_{1}^A$, $\{{}^-C^1_{b_-}\}_{1}^{B_-}$, $\{{}^+C^1_{b_+}\}_1^{B_+}$ and solely focus on these.

\subsection{The composition of two constrained global moves}\label{sec_comp}

Theorem 3.2 in \cite{Dittrich:2013jaa} implies that classically it is equivalent 
\begin{itemize}
\item[(i)] to firstly perform a canonical analysis for the individual moves $0\rightarrow1$ and $1\rightarrow2$ and subsequently perform a matching of the symplectic structures at $n=1$, or
\item[(ii)] to firstly solve the equations of motion at the level of the action at $n=1$ and then perform a canonical analysis for the effective move $0\rightarrow2$.
\end{itemize}
We shall now carry out the quantum analogues of the two procedures (i) and (ii) when composing the moves $0\rightarrow1$ and $1\rightarrow2$ and show that, again, they yield equivalent results (see figure \ref{fig_glue} for a schematic illustration).\\~\\
(i) The quantum analogue to (i) is to firstly match the pre-- and post--physical Hilbert spaces at $n=1$ and then convolute the resulting propagators of the moves $0\rightarrow1$ and $1\rightarrow2$. This procedure is more involved than the quantum analogue to (ii) below, however, for conceptual clarity of the framework we shall detail it here. The quick reader may skip to (ii) below.

\begin{figure}[hbt!]
\begin{center}
\psfrag{0}{$0$}
\psfrag{1}{$1$}
\psfrag{2}{$2$}
\psfrag{g}{ glue}
\psfrag{i}{ integrate}
\psfrag{h0}{\small ${}^-\ch^{\rm phys}_0$}
\psfrag{h02}{\small ${}^-\tilde{\ch}^{\rm phys}_0$}
\psfrag{h11}{\small ${}^+\ch^{\rm phys}_1$}
\psfrag{h12}{\small ${}^-\ch^{\rm phys}_1$}
\psfrag{h2}{\small ${}^+\ch^{\rm phys}_2$}
\psfrag{h22}{\small ${}^+\tilde{\ch}^{\rm phys}_2$}
\begin{subfigure}[b]{.22\textwidth}
\centering
\includegraphics[scale=.45]{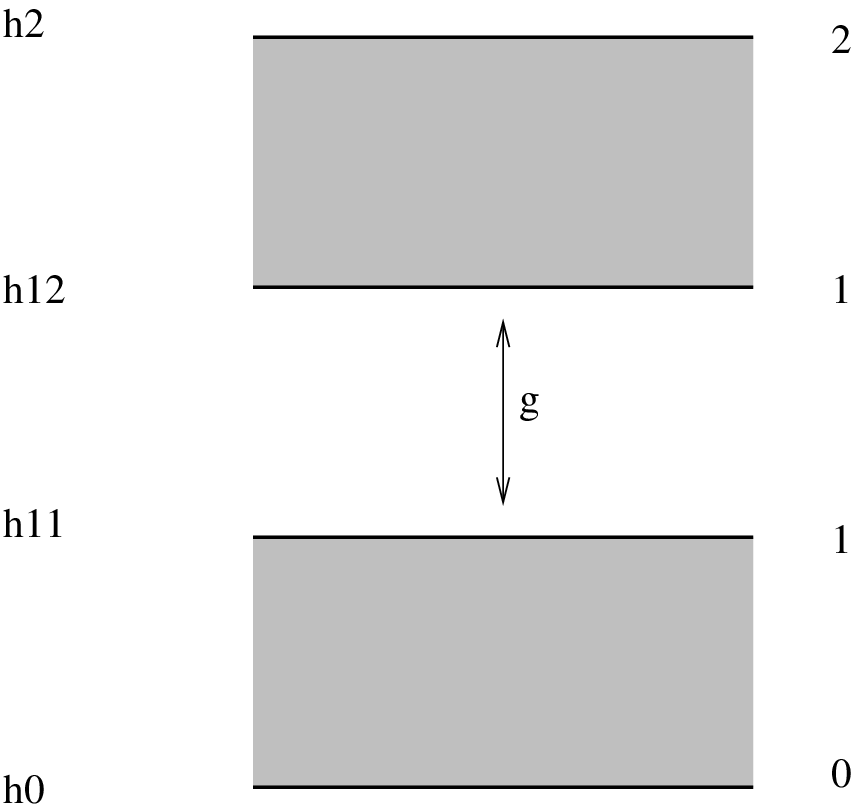}
\centering
\caption{\small }
\end{subfigure}
\hspace{2cm}
\begin{subfigure}[b]{.22\textwidth}
\centering
\includegraphics[scale=.45]{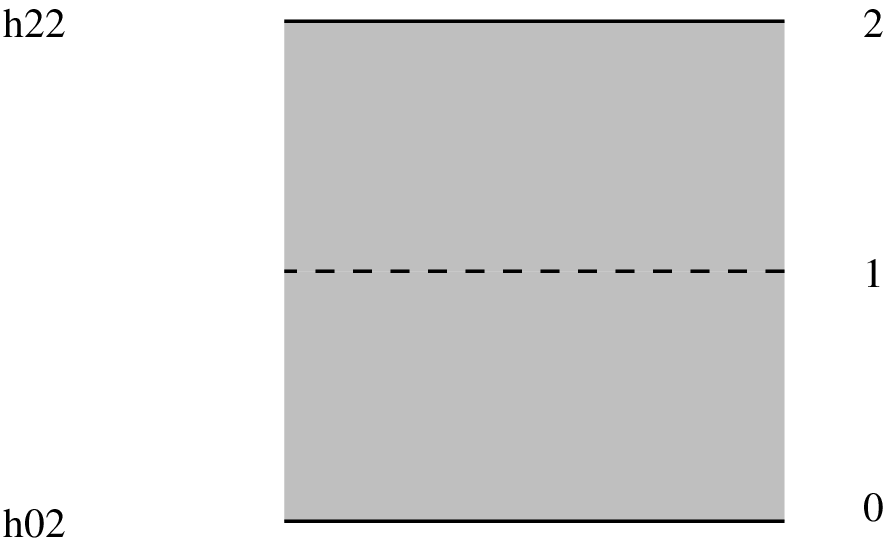}
\centering
\caption{\small }
\end{subfigure}
\hspace{2cm}
\begin{subfigure}[b]{.22\textwidth}
\centering
\includegraphics[scale=.45]{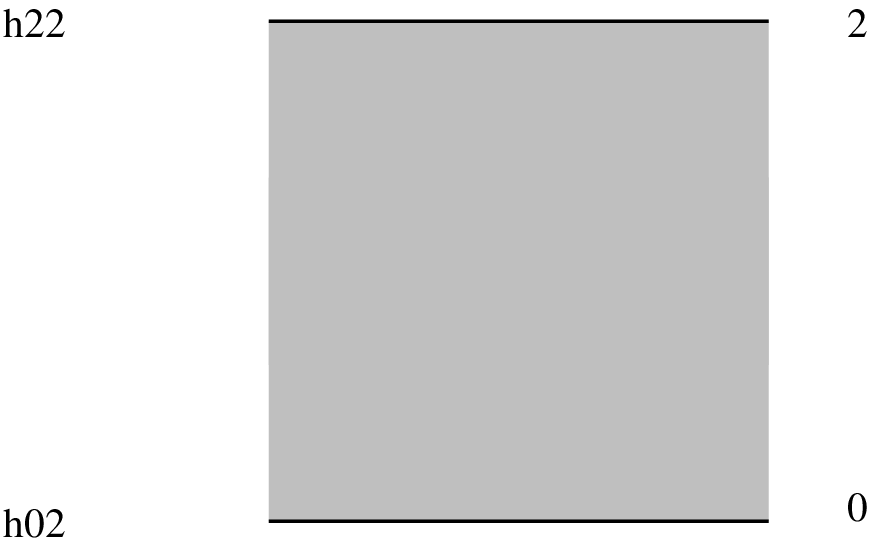}
\centering
\caption{\small }
\end{subfigure}
\caption{\small Composition of the moves $0\rightarrow1$ and $1\rightarrow2$ to an effective move $0\rightarrow2$. The composition can be done by (i) matching the pre-- and post--physical Hilbert spaces at $n=1$ and subsequent regularization, or, equivalently, by (ii) directly convoluting the two propagators without any state matching. If cases (b1) and (b2) occur, the pre-- and post--physical Hilbert spaces ${}^-\tilde{\ch}^{\rm phys}_0$ and ${}^+\tilde{\ch}^{\rm phys}_2$ of the move $0\rightarrow2$ are non-unitary projections of the original ${}^-\ch^{\rm phys}_0$ of the move $0\rightarrow1$ and ${}^+\ch^{\rm phys}_2$ of the move $1\rightarrow2$, respectively (see section \ref{sec_effcon})).}\label{fig_glue}
\end{center}
\end{figure}

Indeed, as a consequence of cases (b1) and (b2), in the quantum theory it generally occurs that ${}^+\ch^{\rm phys}_1\neq{}^-\ch^{\rm phys}_1$. We shall now (formally) construct a single physical Hilbert space $\ch^{\rm phys}_1$ for step $n=1$. For simplicity, let us assume that the remaining quantum pre-- and post--constraints at $n=1$ are also first class such that we can consistently impose both sets simultaneously in the quantum theory. In analogy to the regular moves of section \ref{sec_reg}, we require a matching of the {pre--} and {post--physical states} at $n=1$,
\ba
\tilde{\psi}^{\rm phys}_1:={}^-\psi^{\rm phys}_1\overset{!}{=}{}^+\psi^{\rm phys}_1.\label{psmatch}
\ea
Given that now
\ba
\tilde{\psi}^{\rm phys}_1(x_1)=\int dx_0\,\tilde{K}_{0\rightarrow1}(x_0,x_1){}^-\psi^{\rm kin}_0(x_0)=\int dx_1\left(\tilde{K}_{1\rightarrow2}(x_1,x_2)\right)^*{}^+\psi^{\rm kin}_2(x_2),\nn
\ea
we must replace (\ref{cond4}, \ref{cond5}) by the following conditions
\ba\label{globalcon}
\hat{C}^1_a\tilde{K}_{0\rightarrow1}(x_0,x_1)&=&\q\hat{C}^1_a\tilde{K}^*_{1\rightarrow2}(x_1,x_2)=0\nn\\
{}^-\hat{C}^1_{b_-}\tilde{K}_{0\rightarrow1}(x_0,x_1)&=&{}^-\hat{C}^1_{b_-}\tilde{K}^*_{1\rightarrow2}(x_1,x_2)=0\nn\\
{}^+\hat{C}^1_{b_+}\tilde{K}_{0\rightarrow1}(x_0,x_1)&=&{}^+\hat{C}^1_{b_+}\tilde{K}^*_{1\rightarrow2}(x_1,x_2)=0.
\ea
We denote the new propagators with a tilde to signify that now both pre-- and post--constraints hold. These new propagators are indeed different from the original propagators (\ref{conprop}) and describe different dynamics since formally (i.e.\ ignoring factor ordering ambiguities)\footnote{We emphasize again that in the expression (\ref{conprop}) (and the analogous expression for $K_{1\rightarrow2}$) we assume only those pre-- and post--constraints at $n=1$ to be quantized which are first class upon composition of the two moves.}
\ba
\tilde{K}_{0\rightarrow1}&=&{}^-\mathbb{P}^{B}_1\,\,\,\mathbb{P}^A_1\,\,\,{}^+\mathbb{P}^{B}_1\,\,\,({}^-\mathbb{P}_0)^*\,\,\,\kappa_{0\rightarrow1}\,\,\,\,\,\,\,\underset{(\ref{conprop})}{=}{}^-\mathbb{P}^{B}_1\,K_{0\rightarrow1},\nn\\
\tilde{K}_{1\rightarrow2}&=&{}^+\mathbb{P}_2\,({}^+\mathbb{P}^{B}_1)^*\,(\mathbb{P}^A_1)^*\, ({}^-\mathbb{P}^{B}_1)^*\,\kappa_{1\rightarrow2}\underset{(\ref{conprop})}{=}({}^+\mathbb{P}^{B}_1)^*\,K_{1\rightarrow2},\label{globalproj}
\ea
where we have used the following definitions
\ba
\mathbb{P}^A_1:=\prod_{a=1}^A\delta(\hat{C}^1_a),\q\q {}^-\mathbb{P}^{B}_1:=\prod_{b_-=1}^{B_-}\delta({}^-\hat{C}^1_{b_-}),\q\q{}^+\mathbb{P}^{B}_1:=\prod_{b_+=1}^{B_+}\delta({}^+\hat{C}^1_{b_+}).\label{def111}
\ea

\begin{R}
\emph{Had one chosen to use the construction of the pre-- and post--projectors as given in (\ref{altproj}), instead of (\ref{proj000}, \ref{proj}), one would proceed in a different manner. To consistently implement the construction as indicated in (\ref{altproj}) for all time steps and, in particular, when composing moves, one would have to replace the products of projectors at a fixed time step everywhere in the above discussion and in the sequel by a single projector containing the sum of all the involved constraints in its exponent. For instance, the product ${}^-\mathbb{P}^{B}_1\,\,\,\mathbb{P}^A_1\,\,\,{}^+\mathbb{P}^{B}_1$ in the first line of (\ref{globalproj}) would have to be replaced by the single projector}
\ba
\mathbb{\tilde P}'_1&:=&\f{1}{(2\pi\hbar)^{A+B_-+B_+}}\int_{\mathbb{R}^{A+B_-+B_+}}\prod_{a=1}^Ads^a_1\prod_{b_-=1}^{B_-}ds^{b_-}_1\prod_{b_+=1}^{B_+}ds^{b_+}_1\nn\\
&&\q\q\q\q\q\q\q\q\q\times \,\,\,\,e^{i/\hbar\left(\sum_{a=1}^As^a_1\hat{C}^1_a+\sum_{b_-=1}^{B_-}s^{b_-}_1{}^-\hat{C}^1_{b_-}+\sum_{b_+=1}^{B_+}s^{b_+}_1{}^+\hat{C}^1_{b_+}\right)}.\nn
\ea
\emph{That is, when new constraints arise in the composition of the moves, the specific construction of the projector changes. (Similarly, when using the other method, newly arising projectors in the products would generally require a reordering of the already applied projectors.) The following discussion can then be similarly carried out by replacing in this manner the appearing products of projectors by a single projector with a sum of the involved constraints in the exponential. The conclusions drawn from this method would be analogous and qualitatively the same to the ones we shall now discuss.}
\end{R}

The conditions (\ref{globalcon}) are strong non-local conditions on the local measures. More specifically, these are non-trivial conditions of the move $1\rightarrow2$ on the measure $\tilde{M}_{0\rightarrow1}$ and, vice versa, of the move $0\rightarrow1$ on the measure $\tilde{M}_{1\rightarrow2}$, such that generically $\tilde{M}_{0\rightarrow1}\neq M_{0\rightarrow1}$ and $\tilde{M}_{1\rightarrow2}\neq M_{1\rightarrow2}$. Nevertheless, these can generally be fulfilled by projection as in (\ref{globalproj}). If one similarly considers the composition of various future moves to an effective move $1\rightarrow n$ and subsequently composes the latter with the move $0\rightarrow1$ one will generally encounter non-trivial consistency conditions from the time steps $2,\ldots,n$ on $M_{0\rightarrow1}$ (see also the discussion further below). In the second procedure (ii) below such non-local conditions are avoided.

This entails that in the quantum theory the {\it a priori} free parameters also become {\it a posteriori} free and vice versa. As a consequence, the regularized {\it pre--} and {\it post--measures} of the PIP at $n=1$ and the (gauge) fixing conditions coincide (upon imposing {\it both} pre-- and post--constraints)
\ba
d\xi_1:=d\xi_1^-\equiv d\xi_1^+\nn
\ea
such that
\ba
\langle\tilde\psi^{\rm phys}_1\big|\tilde\phi^{\rm phys}_1\rangle_{\rm phys-}=\langle\tilde\psi^{\rm phys}_1\big|\tilde\phi^{\rm phys}_1\rangle_{\rm phys+}\nn
\ea
and no ambiguity in the definition of the PIP at $n=0$ arises. As a result, a single $\ch^{\rm phys}_1$ emerges as the (improper) projection of ${}^+\ch^{\rm phys}_1$ with ${}^-\mathbb{P}^{B}_1$ and of ${}^-\ch^{\rm phys}_1$ with ${}^+\mathbb{P}^{B}_1$.

Given that we have now imposed the pre-- and post--constraints on both the pre-- and post--states, the corresponding improper projectors at $n=1$ appear twice in the convolution 
\ba
\int dx_1\,\tilde{K}_{0\rightarrow1}\,\tilde{K}_{1\rightarrow2}=\int dx_1\,{}^+\mathbb{P}_2\,\kappa_{1\rightarrow2}\,({}^-\mathbb{P}^{B}_1)^2({}^+\mathbb{P}^{B}_1)^2\,(\mathbb{P}^A_1)^2\, \kappa_{0\rightarrow1}{}^-\mathbb{P}_0\label{div}
\ea
which thus diverges because of the integration over $\cg^+_1\times\cg^-_1/\cg^+_1\cap\cg^-_1$. That is, we have to remove again the ${}^-\mathbb{P}^B_1$ acting on $K_{0\rightarrow1}$ and the ${}^+\mathbb{P}^B_1$ acting on $K_{1\rightarrow2}$, both of which we had just introduced in (\ref{globalproj}). However, in addition the improper projector $\mathbb{P}^A_1$ on the symmetry generating constraints (case (a)) had been applied to both moves $0\rightarrow1$ and $1\rightarrow2$ from the start since it corresponds to those constraints which are both pre-- and post--constraints. That is, the projectors $\mathbb{P}^A_1$ account for genuine divergences resulting from integration over gauge symmetry orbits. For a proper regularization we must therefore remove one of them in (\ref{div}). In analogy to (\ref{pfprop}), we thus define the effective pre--fixed propagators for the composition $0\rightarrow1\rightarrow2$,
\ba
\tilde{K}^{f_+}_{0\rightarrow{1}}&=&{}^-\mathbb{P}^{B}_1\,\,\,\mathbb{P}^A_1\,\,\,{}^+\mathbb{P}^{B}_1\,\kappa_{0\rightarrow1}\q={}^-\mathbb{P}^{B}_1\,K^{f+}_{0\rightarrow{1}}\nn\\
\tilde{K}^{f_+}_{1\rightarrow2}&=&{}^+\mathbb{P}_2\,\kappa_{1\rightarrow2}\,\,\q\q\q\q\q=\q K^{f_+}_{1\rightarrow2}.\label{pfeffprop}
\ea
The effective post--fixed propagators are defined similarly. As in (\ref{comp}), this permits us to define the pre--fixed propagator of the composition of the moves as the (regularized) convolution
\ba
\tilde{K}^{f_+}_{0\rightarrow2}(x_2,x_0)&=&\int\,dx_1\,\tilde{K}^{f_+}_{1\rightarrow2}({x}_2,{x}_1)\,\tilde{K}^{f_+}_{0\rightarrow1}({x}_1,{x}_0)\nn\\
&=&\int\, dx_1\,{}^+\mathbb{P}_2\,\kappa_{1\rightarrow2}{}^-\mathbb{P}^{B}_1\,\,\,\mathbb{P}^A_1\,\,\,{}^+\mathbb{P}^{B}_1\,\kappa_{0\rightarrow1}.\label{regconv}
\ea
Comparing with (\ref{pfeffprop}), this can be written as the projection of a PIP at $n=1$ with ${}^+\mathbb{P}_2$,
\ba
\tilde{K}^{f_+}_{0\rightarrow2}(x_2,x_0)={}^+\mathbb{P}_2\,\left\langle\kappa_{2\rightarrow1}\Big|{}^-\mathbb{P}^{B}_1\,\,\,\mathbb{P}^A_1\,\,\,{}^+\mathbb{P}^{B}_1\,\kappa_{0\rightarrow1}\right\rangle_{\rm kin}.\label{convpip}
\ea
Thus, the convolution is (formally) finite. The pre--fixed reverse propagator satisfies an analogous convolution. 

{\it State matching} (\ref{psmatch}) yields for the effective move $0\rightarrow2$
\ba
{}^+\tilde\psi^{\rm phys}_2(x_2)&=&\int\,dx_1\,\tilde{K}^{f_+}_{1\rightarrow2}(x_2,x_1)\,\tilde{\psi}^{\rm phys}_1(x_1)\nn\\
&=&\int\,dx_1\,\tilde{K}^{f_+}_{1\rightarrow2}(x_2,x_1)\int\,dx_0\,\tilde{K}^{f_+}_{0\rightarrow1}(x_1,x_0){}^-\psi^{\rm phys}_0(x_0)\nn\\
&=&\int\,dx_0\,\tilde{K}^{f_+}_{0\rightarrow2}(x_2,x_0)\,{}^-\psi^{\rm phys}_0(x_0).\label{smatch}
\ea
We denote ${}^+\tilde\psi^{\rm phys}_2$ with a tilde because, as we shall see shortly in section \ref{sec_effcon}, the post--physical states at $n=2$ of the effective move $0\rightarrow2$ satisfy a new set of effective quantum post--constraints at $n=2$ which the post--physical states of the move $1\rightarrow2$ do not satisfy. Similarly, we shall see that $\tilde{K}_{0\rightarrow2}$ contains a further effective projector which projects the pre--physical states ${}^-\psi^{\rm phys}_0$ of the move $0\rightarrow1$ onto new effective pre--constraints at $n=0$. The ${}^-\psi^{\rm phys}_0$ are therefore only `partial' pre--physical states of the move $0\rightarrow2$. 

This concludes procedure (i) in the quantum theory.\\~\\
(ii) The analogue of procedure (ii) in the quantum theory is to directly convolute the propagators corresponding to $0\rightarrow1$ and $1\rightarrow2$ without any state matching. %This method is much simpler and now easy to sketch.

Indeed, in order to obtain a consistent dynamics for the effective move $0\rightarrow2$ we do not need to impose the state matching (\ref{psmatch}) because we can make use of the fact that the propagators act as (improper) projectors on solutions. Noting that ${}^+\psi^{\rm phys}_1=\int dx_0\,K_{0\rightarrow1}\,\psi^{\rm kin}_0=\mathbb{P}^A_1{}^+\mathbb{P}^{B}_1\int dx_0\,\kappa_{0\rightarrow1}{}^-\mathbb{P}_0\,\psi^{\rm kin}_0$ is now a kinematical state with respect to the constraints ${}^-\hat{C}^1_b$, we can project it with $\int dx_1K_{1\rightarrow2}={}^+\mathbb{P}_2\int dx_1\,\kappa_{1\rightarrow2}{}^-\mathbb{P}^{B}_1\,\mathbb{P}^A_1$ to yield the (formally) finite convolution
\ba
{}^+\tilde\psi^{\rm phys}_2(x_2)&=&\int dx_0{}^+\mathbb{P}_2\int dx_1\,\kappa_{1\rightarrow2}{}^-\mathbb{P}^{B}_1\,\mathbb{P}^A_1\,{}^+\mathbb{P}^{B}_1\,\kappa_{0\rightarrow1}{}^-\mathbb{P}_0\,\psi^{\rm kin}_0\nn\\
&=&\int dx_0\,\tilde K_{0\rightarrow2}^{f_+}{}^-\psi^{\rm phys}_0(x_0)\label{nomatch}
\ea
upon removal, for regularization purposes, of one of the doubly occurring projectors $\mathbb{P}^A_1$. As argued above, the $\mathbb{P}^A_1$ are thus ultimately responsible for genuine divergences in the path integral resulting from gauge symmetries. By genuine divergence we mean the fact that {\it the $\mathbb{P}^A_1$ always appear twice in a convolution of two moves} because they correspond to symmetry generators which are constraints that are both pre-- and post--constraints. In any convolution improper projectors of the type $\mathbb{P}^A_1$ will have to be regularized---in contrast to the projectors corresponding to cases (b1) and (b2) above which only arise in one of the two moves $0\rightarrow1$ or $1\rightarrow2$. 

The expression (\ref{nomatch}) is evidently equivalent to (\ref{regconv}) and (\ref{smatch}).\\~\\
The procedures (i) and (ii) are therefore also equivalent in the quantum theory. Given that (ii) is much simpler, we shall henceforth work with this method.

\subsection{Effective constraints and (non--)unitarity}\label{sec_effcon}

We shall now investigate the repercussions of the composition of two moves for the unitarity of the composed (`effective') dynamics. Let us recall the current state of affairs: $P_{0\rightarrow1}=\int dx_0{}^+\mathbb{P}_1\kappa_{0\rightarrow1}{}^-\mathbb{P}_0$ is a linear map from $\ch^{\rm kin}_0$ to ${}^+\ch^{\rm phys}_1$. The kernel of the linear map is $\ker(P_{0\rightarrow1})\simeq\hat{\cg}^-_0$, where $\hat{\cg}^-_0:=\{\psi^{\rm kin}_0,\phi^{\rm kin}_0\in\ch^{\rm kin}_0| {}^-\mathbb{P}_0\,\psi^{\rm kin}_0={}^-\mathbb{P}_0\,\phi^{\rm kin}_0\}$ is the pre--orbit in $\ch^{\rm kin}_0$. The first isomorphism theorem for vector spaces tells us that ${}^+\ch^{\rm phys}_1\simeq P_{0\rightarrow1}(\ch^{\rm kin}_0)\simeq\ch^{\rm kin}_0/\hat{\cg}^-_0\simeq{}^-\ch^{\rm phys}_0$. Indeed, as seen above, $U_{0\rightarrow1}=\int dx_0{}^+\mathbb{P}_1\,\kappa_{0\rightarrow1}$ (formally) yields a unitary isomorphism from ${}^-\ch^{\rm phys}_0$ to ${}^+\ch^{\rm phys}_1$. But what happens to the kernel and image of $P_{0\rightarrow1}$ in a composition?

In the classical formalism, constraints of the type (b1) and (b2) lead to a `propagation' of constraints in discrete time (see \cite{Dittrich:2013jaa,Hoehn:2014aoa} for details). Consider the pre--constraints of type (b1) at step $n=1$ of the move $1\rightarrow2$. In a composition these pre--constraints also have to be satisfied by the canonical data of the move $0\rightarrow1$: only those data at $n=0$ are permitted which also satisfy the pre--constraints at $n=1$ upon propagation with the time evolution map of the move $0\rightarrow1$. That is, the pre--constraints of case (b1) at $n=1$ have effectively `propagated' back to step $n=0$ and appear as new effective pre--constraints ${}^-\tilde{C}^0_\Lambda$. Similarly, in a composition the post--constraints of case (b2) at $n=1$ `propagate' forward to $n=2$ in the form of new effective post--constraints ${}^+\tilde{C}^2_\Lambda$. This is a consequence of the equations of motion at $n=1$ which act as secondary constraints.

We emphasize that cases (b1) and (b2) can only occur for a temporally varying discretization. The effective propagation of constraints in the classical formalism represents an effective change of discretization at steps $n=0,2$. Only those degrees of freedom of the discretization at steps $n=0$ and $n=2$ remain dynamically relevant which are compatible and Poisson commute with both the old and new effective pre-- and post--constraints at $n=0,2$, respectively \cite{Dittrich:2013jaa}. Following the arguments put forward in \cite{Dittrich:2013xwa}, the pre--constraints of case (b1) can be viewed as non-trivial conditions on the move $0\rightarrow1$, while the post--constraints of case (b2) can be interpreted as non-trivial coarse graining conditions on the move $1\rightarrow2$. Any information representing degrees of freedom below a certain coarse graining scale, although propagating in either of the moves $0\rightarrow1$ and $1\rightarrow2$, will not Poisson commute with the effective constraints and is thereby rendered dynamically irrelevant for the move $0\rightarrow2$.

%In the classical formalism, when solving the equations of motion at $n=1$, one may obtain new `effective' post--constraints ${}^+\tilde{C}^2_\Lambda$ at $n=2$ and/or new `effective' pre--constraints ${}^-\tilde{C}^0_\Lambda$ at $n=0$ \cite{Dittrich:2013jaa}. This happens until the total number of post--constraints at $n=2$ coincides with both the total number of pre--constraints at $n=0$ and the sum of the total number of independent first class constraints and half the number of second class constraints at $n=1$. That is, constraints can `propagate' forward to $n=2$ or backward to $n=0$ because the equations of motion at $n=1$ act as secondary constraints. However, everything remains consistent because the effective move $0\rightarrow2$ is simply a different move than the individual $0\rightarrow1$ and $1\rightarrow2$ and may involve a change of discretization \cite{Dittrich:2013jaa}. 

%In order for our quantization to be faithful and consistent, we also have to impose these additional `effective' constraints in the quantum theory because, in contrast to the classical case, the effective constraints will not automatically be satisfied by states and propagators, as we shall now discuss. %For example, for any new such post--constraint, we have to apply the corresponding improper projector $\delta({}^+\widehat{\tilde{C}}{}^2_\Lambda)$ to (\ref{regconv}). 

The situation in the quantum theory is analogous. Consider the composition of $0\rightarrow1$ and $1\rightarrow2$ in the quantum theory as given in (\ref{smatch}, \ref{nomatch}). The ${}^-\hat{C}^1_{b_-}$ are non-trivial constraints on the image of $P_{0\rightarrow1}$. Indeed, from (\ref{regconv}, \ref{nomatch}) we can pick out the effective linear map from $\ch^{\rm kin}_0$ to ${}^-\mathbb{P}^B_1({}^+\ch^{\rm phys}_1)$
\ba
\tilde{P}_{0\rightarrow1}:=\int dx_0{}^-\mathbb{P}^{B}_1\,\mathbb{P}^A_1\,{}^+\mathbb{P}^{B}_1\,\kappa_{0\rightarrow1}{}^-\mathbb{P}_0.\nn
\ea
Thanks to the improper projector ${}^-\mathbb{P}^B_1$, its image $
\tilde{P}_{0\rightarrow1}(\ch^{\rm kin}_0)\simeq{}^-\mathbb{P}^B_1({}^+\ch^{\rm phys}_1)$ is no longer contained in the post--physical Hilbert space $P_{0\rightarrow1}(\ch^{\rm kin}_0)\simeq{}^+\ch^{\rm phys}_1$. By the first isomorphism theorem, $\tilde{P}_{0\rightarrow1}(\ch^{\rm kin}_0)\simeq\ch^{\rm kin}_0/\ker(\tilde{P}_{0\rightarrow1})$. The latter can thus also no longer be contained in ${}^-\ch^{\rm phys}_0\simeq\ch^{\rm kin}_0/\ker(P_{0\rightarrow1})$, for otherwise we could use $U_{0\rightarrow1}$ to map it isomorphically into ${}^+\ch^{\rm phys}_1$. 

There must therefore exist another `effective' improper projector ${}^-\mathbb{\tilde P}_0$ at $n=0$ such that ${}^-\mathbb{\tilde P}_0({}^-\ch^{\rm phys}_0)\simeq\ch^{\rm kin}_0/\ker(\tilde{P}_{0\rightarrow1})\simeq{}^-\mathbb{P}^B_1({}^+\ch^{\rm phys}_1)$ and we can replace $\kappa_{0\rightarrow1}$ by a new kinematical propagator $\tilde{\kappa}_{0\rightarrow1}$ such that
\ba
\tilde{P}_{0\rightarrow1}=\int dx_0{}^-\mathbb{P}^{B}_1\,\mathbb{P}^A_1\,{}^+\mathbb{P}^{B}_1\,\tilde\kappa_{0\rightarrow1}{}^-\mathbb{P}_0{}^-\mathbb{\tilde P}_0.\label{neffkinprop}
\ea
This improper projector should be expressible in terms of a set of (independent) `effective' quantum constraints that is in number equal to the pre--constraints ${}^-\hat{C}^1_{b_-}$ constituting ${}^-\mathbb{P}^B_1$,
\ba
{}^-\mathbb{\tilde P}_0=\prod_{\Lambda=1}^{B_-}\delta({}^-\widehat{\tilde C}{}^0_\Lambda).\nn
\ea
%Classically, the number $B_-$ of pre--constraints ${}^-{C}^1_{b_-}$ at $n=1$ coincides with the number of new effective pre--constraints ${}^-\tilde{C}^0_\Lambda$ at $n=0$ \cite{Dittrich:2013jaa}. 
The $B_-$ `effective' quantum pre--constraints in the projector, ${}^-\widehat{\tilde C}{}^0_\Lambda$, can be viewed as the `backward propagated' ${}^-\hat{C}^1_{b_-}$ which suitably project the preimage of $\tilde{P}_{0\rightarrow1}$ such that it can be isomorphically mapped to ${}^-\mathbb{P}_1^B({}^+\ch^{\rm phys}_1)$. The ${}^-\widehat{\tilde C}{}^0_\Lambda$ should thus correspond to a quantization of the classical effective pre--constraints. It should be noted, however, that this quantization may not coincide with a na\"ive quantization of the classical `effective' constraints since non-trivial quantum corrections may arise from integrating out the intermediate time step $n=1$.

In conclusion, imposing $B_-$ new non-trivial quantum pre--constraints at $n=1$ on the move $0\rightarrow1$ produces $B_-$ new effective quantum pre--constraints at $n=0$. By the same reasoning, imposing the $B_+$ new quantum post--constraints ${}^+\hat{C}^1_{b_+}$ at $n=1$ on the move $1\rightarrow2$ produces $B_+$ new effective quantum post--constraints ${}^+\widehat{\tilde C}{}^2_{\Lambda'}$ at $n=2$ along with an effective improper projector ${}^+\mathbb{\tilde P}_2$. These can likewise be viewed as a (possibly non-trivial) quantization of the classical effective post--constraints ${}^+\tilde{C}^2_{\Lambda'}$ mentioned above.\footnote{One may be worried that the same reasoning applies to the pre--fixed propagator (\ref{pfprop}). That is, given that it is projected from the left with ${}^+\mathbb{P}_1$ should likewise imply that it already contains an `effective' projection on the right, despite ${}^-\mathbb{P}_0$ having been dropped in (\ref{pfprop}). However, firstly this `effective' projection is a kinematical one and depends on the choice of the kinematical propagator in a given orbit due to (\ref{propnomatch}) and can thus be viewed as a choice of `gauge'. By contrast, $\tilde{P}_{0\rightarrow1}$ is a dynamical projection of the physical propagator in $P_{0\rightarrow1}$ resulting from integrating out step $n=1$. Secondly, the pre--fixed propagator (\ref{pfprop}), in contrast to $\tilde{P}_{0\rightarrow1}$, is applied to ${}^-\ch^{\rm phys}_0$ and not $\ch^{\rm kin}_0$ such that the choice of the kinematical propagator within a given orbit is irrelevant.}

Using new kinematical propagators as in (\ref{neffkinprop}), the linear map (\ref{smatch}, \ref{nomatch}) can be recast 
\ba
{}^+\tilde{\psi}^{\rm phys}_2(x_2)&=&{}^+\mathbb{\tilde P}_2{}^+\mathbb{P}_2\int dx_0\,dx_1\,\tilde\kappa_{1\rightarrow2}{}^-\mathbb{P}^{B}_1\,\mathbb{P}^A_1\,{}^+\mathbb{P}^{B}_1\,\tilde\kappa_{0\rightarrow1}\,{}^-\mathbb{\tilde P}_0\,{}^-\psi^{\rm phys}_0\nn\\
&=&{}^+\mathbb{\tilde P}_2{}^+\mathbb{P}_2\int dx_0\,\kappa_{0\rightarrow2}(x_0,x_2){}^-\tilde{\psi}^{\rm phys}_0(x_0),\label{nomatch2}\\
&=&\int dx_0\,\tilde{K}^{\tilde f_+}_{0\rightarrow2}(x_0,x_2){}^-\tilde{\psi}^{\rm phys}_0(x_0),\nn
\ea
where the (non-unique) kinematical propagator of the move $0\rightarrow2$ reads
\ba
\kappa_{0\rightarrow2}=\int dx_1\,\tilde\kappa_{1\rightarrow2}{}^-\mathbb{P}^{B}_1\,\mathbb{P}^A_1\,{}^+\mathbb{P}^{B}_1\,\tilde\kappa_{0\rightarrow1}.\nn
\ea
The pre--physical states in (\ref{nomatch2}) of the move $0\rightarrow2$ are thus projections of the pre--physical states of the move $0\rightarrow1$, ${}^-\tilde{\psi}^{\rm phys}_0\in{}^-\tilde{\ch}^{\rm phys}_0:={}^-\mathbb{\tilde P}_0({}^-\ch^{\rm phys}_0)$, while the post--physical states in (\ref{nomatch2}) are projections of the post--physical states of the move $1\rightarrow2$, ${}^+\tilde{\psi}^{\rm phys}_2\in{}^+\tilde{\ch}^{\rm phys}_2:={}^+\mathbb{\tilde P}_2({}^+\ch^{\rm phys}_2)$. ${}^-\tilde{\ch}^{\rm phys}_0$ and ${}^+\tilde{\ch}^{\rm phys}_2$ are thus the effective pre-- and post--physical Hilbert spaces of the effective move $0\rightarrow2$ at $n=0,2$, respectively (see figure \ref{fig_glue} for a schematic illustration). At this stage, the reasoning of the previous sections applies again in which case $\tilde{U}_{0\rightarrow2}=\int dx_0{}^+\mathbb{\tilde P}_2{}^+\mathbb{P}_2\,\kappa_{0\rightarrow2}$ constitutes a {\it unitary} isomorphism from ${}^-\tilde{\ch}^{\rm phys}_0$ to ${}^+\tilde{\ch}^{\rm phys}_2$ for the move $0\rightarrow2$.

On the other hand, given that ${}^-\tilde{\ch}^{\rm phys}_0:={}^-\mathbb{\tilde P}_0({}^-\ch^{\rm phys}_0)$ is no longer contained in the original ${}^-\ch^{\rm phys}_0$ of the move $0\rightarrow 1$ and ${}^+\tilde{\ch}^{\rm phys}_2:={}^+\mathbb{\tilde P}_2({}^+\ch^{\rm phys}_2)$ is no longer contained in the original ${}^+\ch^{\rm phys}_2$ of the move $1\rightarrow2$, we see that integrating out step $n=1$ has lead to a {\it non-unitary} change of pre-- and post--physical Hilbert spaces at steps $n=0,2$, respectively (provided, of course, constraints of cases (b1) and (b2) occur). In analogy to the classical case sketched above, this situation can be interpreted as a dynamical change of discretization at steps $n=0,2$. The new effective quantum constraints at $n=0,2$ lead to a coarse graining of the discretization at steps $n=0,2$ in the sense of \cite{Dittrich:2013xwa}. The constraint `propagation' ensures that only the data up to a certain refinement or coarse graining scale is accepted. The information about finer degrees of freedom at $n=0,2$ which was relevant for the moves $0\rightarrow1$ and $1\rightarrow2$ is irreversibly projected out through ${}^+\mathbb{\tilde P}_2,{}^-\mathbb{\tilde P}_0$, whence the non-unitarity comes from. In section \ref{sec_dirac} we shall study this non-unitary projection in terms of Dirac observables which represent the physical degrees of freedom. 

The physical Hilbert spaces are therefore evolution move dependent such that in this sense one may speak of `evolving' physical Hilbert spaces. For instance, for a fixed $n$, the number of constraints associated to the move $n\rightarrow m$ depends on the choice of the next time step $m$. For a further step $m'>m$ the numbers of constraints involved in the move $n\rightarrow m'$ is higher or equal to the number of constraints of the move $n\rightarrow m$. Physically this happens because the number of degrees of freedom propagating from the initial $n$ to the final time $m$ can only remain constant or decrease with growing $m$ in a temporally varying discretization \cite{Dittrich:2013jaa,Hoehn:2014aoa}. The pre--physical Hilbert space ${}^-\ch^{\rm phys}_n$ at $n$ will thus be projected further and further by effective projectors associated to new effective pre--constraints at $n$ arising in the course of evolution. Accordingly, the unitary isomorphisms $U_{n\rightarrow m}$ and $U_{n\rightarrow m'}$ will generally be subject to different constraints and will thereby depend on the move. This move dependence of the physical Hilbert spaces appears explicitly in the toy model of section \ref{sec_ex}.
 
Furthermore, in \cite{Hoehn:2014wwa} it is shown in detail how coarse graining local evolution moves can change the physical inner product and thereby introduce non-unitarity into the dynamics. Finally, it should be noted that the composition of moves can also produce holonomic and boundary data constraints (which are neither pre-- nor post--constraints) at the boundary time steps \cite{Dittrich:2013jaa}. This is exhibited in \cite{Hoehn:2014aoa} in the context of quadratic discrete actions.

\begin{Example}\label{ex_ex}
\emph{We consider a rather extreme qualitative example for illustration. Let the move $0\rightarrow1$ be unconstrained while the move $1\rightarrow2$ is fully constrained. It follows that all $N$ pre-- and all $N$ post--constraints at $n=1,2$, respectively, are linear in the momenta \cite{Dittrich:2013jaa}. In the quantum theory this implies that the physical post--state ${}^+\psi^{\rm phys}_2$ and the physical pre--state ${}^-\psi^{\rm phys}_1$ are {\it unique} (up to normalization). Thus, given that the propagator of the move $1\rightarrow2$ must satisfy these constraints, it must factorize into these unique states,
\ba
K_{1\rightarrow2}(x_1,x_2)=  {}^+\psi^{\rm phys}_2(x_2)({}^-\psi^{\rm phys}_1(x_1))^*.\nn
\ea
On account of this product structure there exists no non-trivial correlation between steps $1$ and $2$ and thereby no dynamics.}

\emph{On the other hand, the kernel of the map $\int dx_0\,K_{0\rightarrow1}$ is trivial because the move $0\rightarrow1$ is unconstrained. Therefore, ${}^+\ch^{\rm phys}_1\equiv\ch^{\rm kin}_1$ and ${}^-\ch^{\rm phys}_0\equiv\ch^{\rm kin}_0$. }

\emph{The composition with $1\rightarrow2$ yields a factorization and thus trivial dynamics for $0\rightarrow2$,
\ba
K_{0\rightarrow2}(x_0,x_2)={}^+\psi^{\rm phys}_2(x_2)\int dx_1({}^-\psi^{\rm phys}_1(x_1))^*K_{0\rightarrow1}(x_0,x_1)={}^+\psi^{\rm phys}_2(x_2)({}^-\phi^{\rm phys}_0(x_0))^*.\nn
\ea
The kernel of $\int dx_0\,K_{0\rightarrow2}$ is all of $\ch^{\rm kin}_0$. Indeed, given a fixed choice of $K_{0\rightarrow1}$, ${}^-\phi^{\rm phys}_0(x_0)$ is now uniquely determined (up to normalization and an irrelevant integration constant). Consequently, ${}^-\phi^{\rm phys}_0$ must also satisfy $N$ effective pre--constraints ${}^-\widehat{\tilde C}{}^0_\Lambda$ that single this state out. The effective pre--physical Hilbert space ${}^-\tilde{\ch}^{\rm phys}_0$ is now one-dimensional and distinct from the original ${}^-\ch^{\rm phys}_0\equiv\ch^{\rm kin}_0$. This indicates qualitatively how the physical Hilbert space can depend on the evolution move.}
\end{Example}

\subsection{A discrete `general boundary formulation'}\label{sec_gbf}

In a space-time context, different global evolution moves correspond to different (triangulated) regions of space-time that can be glued together (see figure \ref{fig_gbf} for a schematic illustration of the case where the boundaries of the space-time regions have two connected components). Since such regions of space-time and the associated boundary hypersurfaces can be quite arbitrary it is not surprising that different evolution moves are associated to different physical Hilbert spaces, constraints and physical degrees of freedom. 

We emphasize at this stage that, although we make explicit use of `initial' and `final' time steps (or boundaries) in our construction of the dynamics, the formalism is general because these time steps can be chosen or assigned quite arbitrarily. We have furthermore never invoked the signature of the underlying space-time geometry and, in particular, we have not resorted to the distinction between timelike, spacelike or null directions in the case of a Lorentzian signature. Accordingly, the `time' direction employed in the formalism is a mere consequence of the choice of `initial' and `final' discrete time steps and the latter need not even correspond to spacelike hypersurfaces in a Lorentzian geometry. Indeed, the formalism equally applies to more general hypersurfaces and, likewise, to geometries with Euclidean signature. Moreover, a given individual time step can correspond to connected or disconnected components of a boundary or even to an empty boundary (as necessary for the discrete incarnation of the `no-boundary' proposal, see sections \ref{sec_dirac} and \ref{sec_ex}). Similarly, the `initial' and `final' hypersurfaces can intersect and together form a connected boundary of a space-time region. This will be more amply discussed in the companion article \cite{Hoehn:2014wwa} in combination with local moves (for the analogous classical situation see \cite{Dittrich:2011ke,Dittrich:2013jaa}). What matters is that one can associate an action to a space-time region and that such regions can be glued together along a common component of their boundaries. 

In fact, this formalism can be viewed as a discrete incarnation of the `general boundary formulation' of quantum theory \cite{Oeckl:2003vu,Oeckl:2005bv,Oeckl:2010ra,Oeckl:2011qd}. Also there the Hilbert spaces and amplitude maps depend crucially on the space-time region and its boundary under consideration.

\begin{figure}[hbt!]
\begin{center}
\psfrag{0}{$0$}
\psfrag{1}{$1$}
\psfrag{2}{$2$}
\psfrag{r1}{\large $R_1$}
\psfrag{r2}{\large $R_2$}
\psfrag{r}{\large $R$}
\psfrag{h0}{\tiny ${}^-\ch^{\rm phys}_0$}
\psfrag{h02}{\tiny ${}^-\tilde{\ch}^{\rm phys}_0$}
\psfrag{h11}{\tiny ${}^+\ch^{\rm phys}_1$}
\psfrag{h12}{\tiny ${}^-\ch^{\rm phys}_1$}
\psfrag{h2}{\tiny ${}^+\ch^{\rm phys}_2$}
\psfrag{h22}{\tiny ${}^+\tilde{\ch}^{\rm phys}_2$}
\psfrag{s0}{\small $\Sigma_0$}
\psfrag{s1}{\small $\Sigma_1$}
\psfrag{s2}{\small $\Sigma_2$}
\hspace*{-1.2cm}\begin{subfigure}[b]{.22\textwidth}
\centering
\includegraphics[scale=.33]{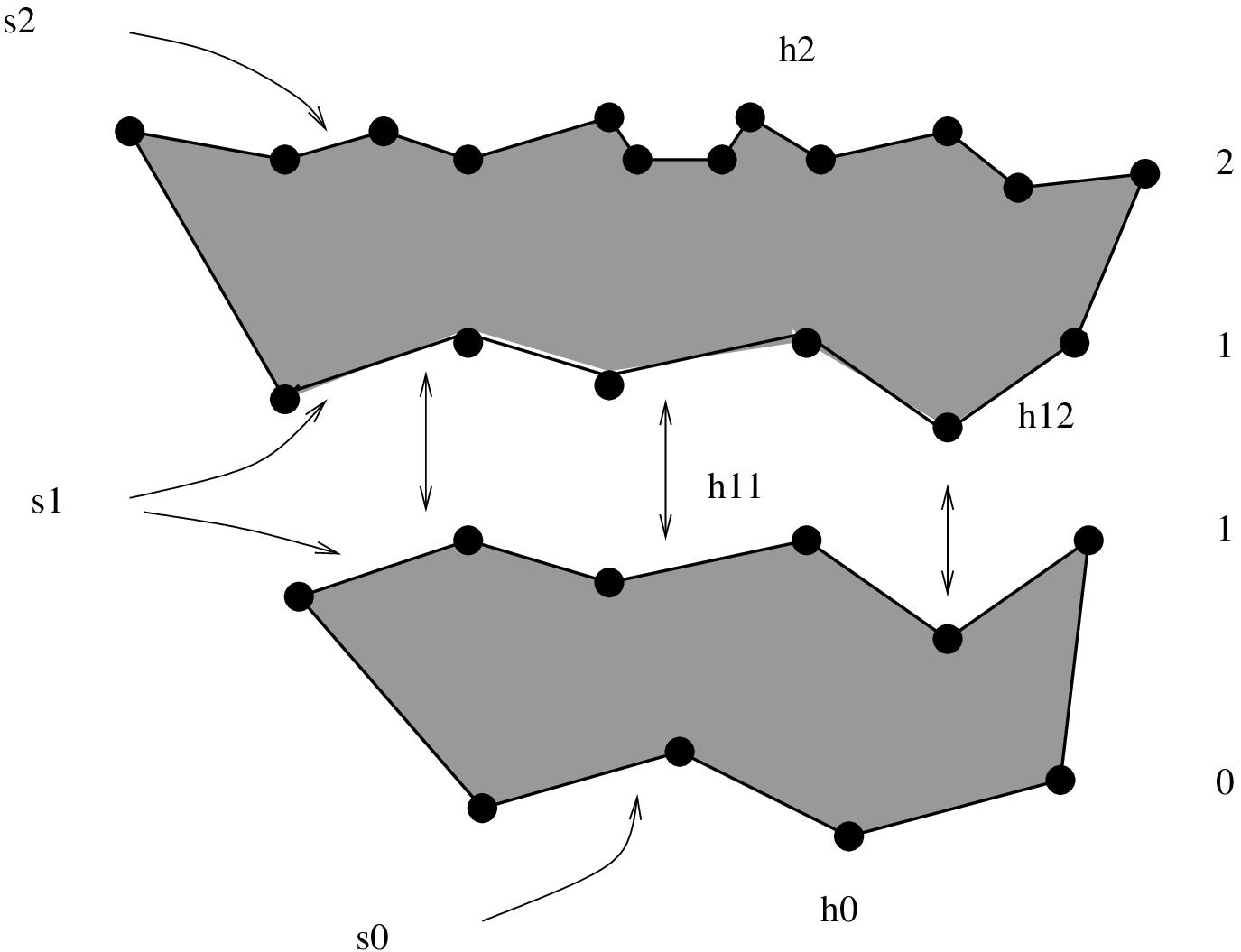}
\centering
\caption{\small }
\end{subfigure}
\hspace{1.8cm}
\begin{subfigure}[b]{.22\textwidth}
\centering
\includegraphics[scale=.33]{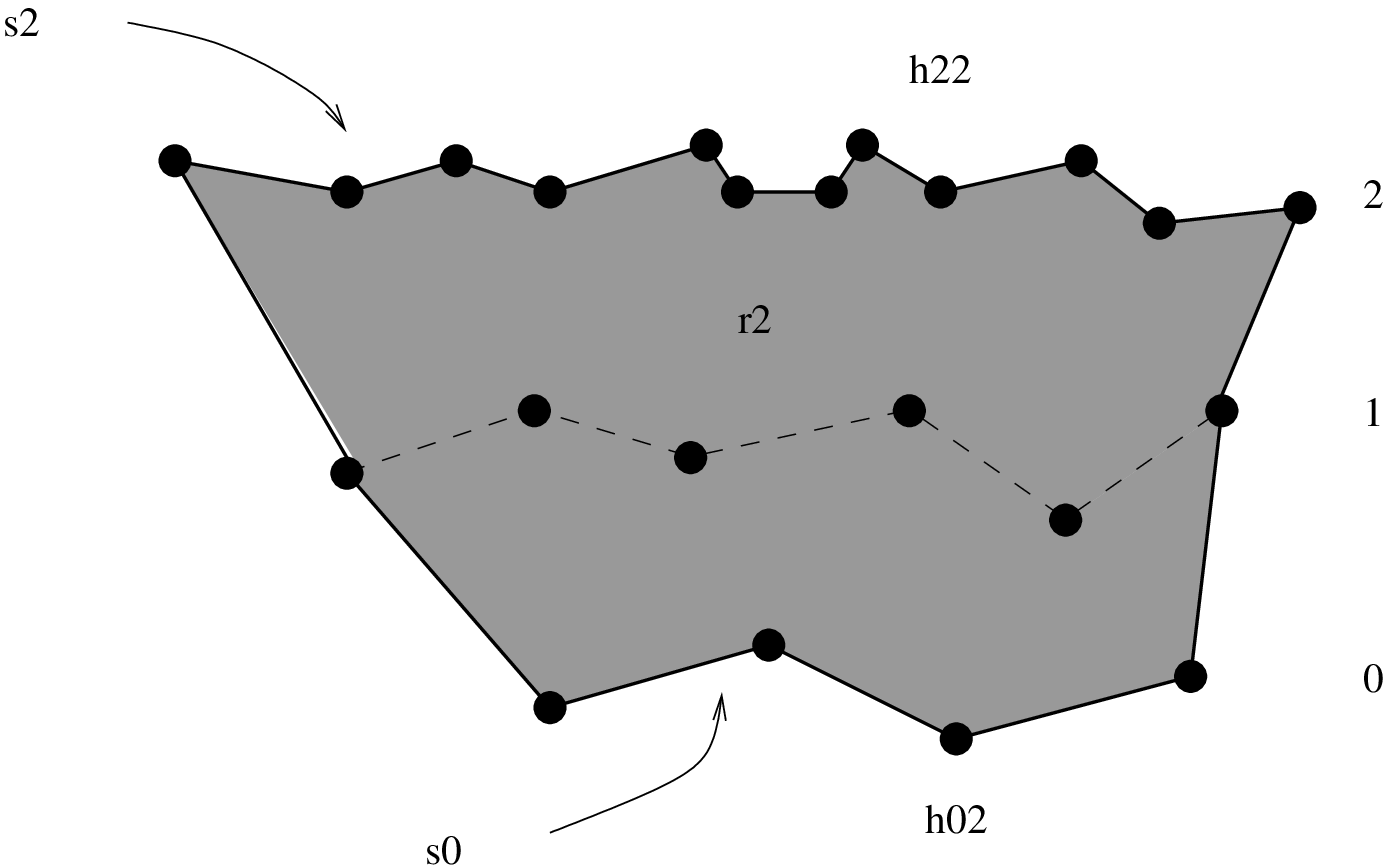}
\centering
\caption{\small }
\end{subfigure}
\hspace{1.8cm}
\begin{subfigure}[b]{.22\textwidth}
\centering
\includegraphics[scale=.33]{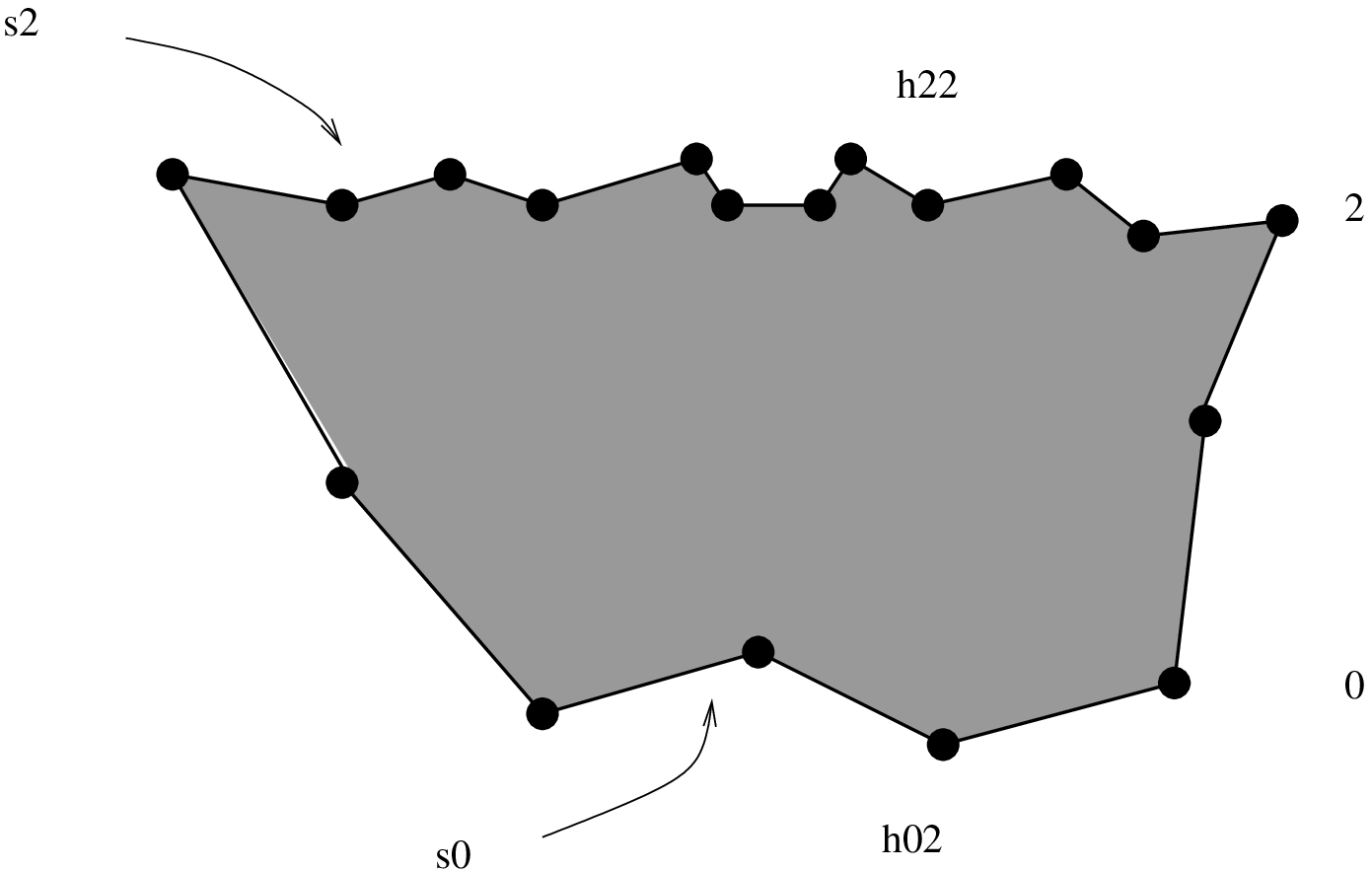}
\centering
\caption{\small }
\end{subfigure}
\caption{\small In a space-time context, different global evolution moves correspond to different regions in space-time. Here we associate the (variables in the) boundary surface $\Sigma_n$ to time step $n$. The move $0\rightarrow1$ is associated to the region $R_1$ with boundary hypersurface $\Sigma_0\cup\Sigma_1$ and boundary Hilbert space $\ch^{\rm phys}_{0\rightarrow1}={}^-\ch^{\rm phys}_0\otimes({}^+\ch^{\rm phys}_1)^*$. The move $1\rightarrow2$ is associated to the region $R_2$ with boundary hypersurface $\Sigma_1\cup\Sigma_2$ and boundary Hilbert space $\ch^{\rm phys}_{1\rightarrow2}={}^-\ch^{\rm phys}_1\otimes({}^+\ch^{\rm phys}_2)^*$. Gluing these regions along $\Sigma_1$ to produce the effective move $0\rightarrow2$ yields the new region $R$ with boundary $\Sigma_0\cup\Sigma_2$. The new region is associated to a new boundary Hilbert space $\ch^{\rm phys}_{0\rightarrow2}={}^-\tilde\ch^{\rm phys}_0\otimes({}^+\tilde\ch^{\rm phys}_2)^*$.}\label{fig_gbf}
\end{center}
\end{figure}

Let us briefly outline the relation between the two frameworks. Since the pre-- and post--physical Hilbert spaces are associated to a given move, we can define a new physical Hilbert space associated to the move $0\rightarrow1$ which is the tensor product of its pre-- and (dual) post--physical Hilbert spaces $\ch^{\rm phys}_{0\rightarrow 1}:={}^-\ch^{\rm phys}_0 \otimes\left({}^+\ch^{\rm phys}_{1}\right)^*$. In the spirit of the `general boundary formulation' the physical Hilbert space of the move $0\rightarrow1$ is then a `boundary' Hilbert space which, at least in a space-time context, is associated to the boundary of a space-time region. Physical boundary states of the move can be written as $\psi^{\rm phys}_{0\rightarrow1}:={}^-\psi^{\rm phys}_0\otimes\iota_1({}^+\phi^{\rm phys}_1)\in\ch^{\rm phys}_{0\rightarrow1}$, where $\iota_1:{}^+\ch^{\rm phys}_1\rightarrow({}^+\ch^{\rm phys}_1)^*$ is a linear isometric involution (see \cite{Oeckl:2003vu,Oeckl:2005bv,Oeckl:2010ra,Oeckl:2011qd} for details).

In the `general boundary formulation' \cite{Oeckl:2003vu,Oeckl:2005bv,Oeckl:2010ra,Oeckl:2011qd} the so-called amplitude map essentially encodes the information about the dynamics. In the present formalism this amplitude map $\rho_{0\rightarrow 1}:\ch^{\rm phys}_{0\rightarrow 1}\rightarrow\mathbb{C}$ can be defined via the unitary isomorphism $U_{0\rightarrow 1}$ and is given by (\ref{trans})
\ba
\rho_{0\rightarrow 1}(\psi^{\rm phys}_{0\rightarrow 1}):=\langle{}^+\psi^{\rm phys}_{1}\big|U_{0\rightarrow 1}({}^-\phi^{\rm phys}_0)\rangle_{\rm phys+}.\nn
\ea
This is the transition amplitude for the initial state ${}^-\phi^{\rm phys}_0$ and final state ${}^+\psi^{\rm phys}_1$. $\rho_{0\rightarrow1}$ is therefore determined by the propagator (and thus the action and measure) of the move $0\rightarrow1$ and thereby contains all the information about the dynamics. 

In order to obtain the amplitude map for the composition of the moves $0\rightarrow1$ and $1\rightarrow2$ one can translate axiom (T5) of \cite{Oeckl:2010ra,Oeckl:2011qd} into the discrete. Using the state matching of procedure (i) in section \ref{sec_comp} one can choose an orthonormal basis $\{\tilde\xi^{\rm phys}_{1i}\}_{i\in I}$ in $\ch^{\rm phys}_1$ (the matched pre-- and post--physical Hilbert space at $n=1$). Then axiom (T5) of \cite{Oeckl:2010ra,Oeckl:2011qd} becomes in our formalism
\ba
\tilde\rho_{0\rightarrow2}(\tilde\psi^{\rm phys}_{0\rightarrow2})&=&\sum_{i\in I}\langle \tilde U_{2\rightarrow1}\,{}^+\tilde\psi^{\rm phys}_2\big|\tilde\xi_{1i}^{\rm phys}\rangle_{\rm phys}\langle\tilde\xi^{\rm phys}_{1i}\big|\tilde U_{0\rightarrow1}\,{}^-\tilde\phi^{\rm phys}_0\rangle_{\rm phys}\nn\\
&=&
\langle \tilde U_{2\rightarrow1}\,{}^+\tilde\psi^{\rm phys}_2\big|\tilde U_{0\rightarrow1}\,{}^-\tilde\phi^{\rm phys}_0\rangle_{\rm phys}\nn\\
&=&\langle {}^+\tilde\psi^{\rm phys}_2\big|\tilde U_{1\rightarrow2}\circ\tilde U_{0\rightarrow1}\,{}^-\tilde\phi^{\rm phys}_0\rangle_{\rm phys+}\nn\\
&\underset{(\ref{nomatch2})}{=}&\int dx_2\,(\psi^{\rm kin}_2)^*\int dx_0\,\tilde{K}^{\tilde f_+}_{0\rightarrow2}\,{}^-\tilde{\phi}^{\rm phys}_0.\nn
\ea
Here $\tilde U_{0\rightarrow1}=\int dx_0\,{}^-\mathbb{P}^B_1\mathbb{P}^A_1{}^+\mathbb{P}^B_1\tilde\kappa_{0\rightarrow1}$ as in (\ref{neffkinprop}) (and similarly for $\tilde U_{1\rightarrow2}$). In the third line we have made use of unitarity and invertibility in the sense of (\ref{unit}, \ref{coninvert}). This gives the correct transition amplitude for the new boundary states $\tilde{\psi}^{\rm phys}_{0\rightarrow2}={}^-\tilde\phi^{\rm phys}_0\otimes\iota_2({}^+\tilde\psi^{\rm phys}_2)$ in the new boundary Hilbert space $\ch^{\rm phys}_{0\rightarrow2}:={}^-\tilde\ch^{\rm phys}_0\otimes({}^+\tilde\ch^{\rm phys}_2)^*$.

Using the tools of the `general boundary formulation', the present formalism can thus be equivalently formulated in terms of boundary Hilbert spaces and amplitude maps. In particular, the probability interpretation of the `general boundary formulation' applies to the present formalism. For evolving physical Hilbert spaces the definition of transition probabilities between initial an final states depends crucially on the given evolution move.

\subsection{Path integral for constrained global moves}\label{sec_pi}

The path integral (PI) for a composition of constrained global moves $0\rightarrow1\rightarrow2\rightarrow\cdots\rightarrow n$ can now be given in two forms. Firstly, the divergent and non-fixed PI simply reads
\ba
K_{0\rightarrow n}(x_n,x_0)&=&\int_{\cq^{n-1}}\prod^{n-1}_{j=0}K_{j\rightarrow j+1}(x_{j+1},x_j)\prod^{n-1}_{l=1}dx_l\nn\\
&{=}&\int_{\cq^{n-1}}e^{i/\hbar\sum^n_{m=1}S_m(x_m,x_{m-1})}\prod^{n-1}_{j=0}M_{j\rightarrow j+1}(x_{j+1},x_j)\prod^{n-1}_{l=1}dx_l.
\ea
On the other hand, the pre--fixed PI is given by 
\ba
\tilde{K}^{f_+}_{0\rightarrow n}(x_n,x_0)&=&\int_{\cq^{n-1}}\prod^{n-1}_{j=0}\tilde{K}^{f_+}_{j\rightarrow j+1}(x_{j+1},x_j)\prod^{n-1}_{l=1}dx_l\nn\\
&=&\int_{\cq^{n-1}}{}^+\mathbb{P}_n\,\kappa_{n-1\rightarrow n}\,\prod^{n-2}_{j=0}\left(\mathbb{P}^A_{j+1}\,{}^-\mathbb{P}^B_{j+1}\,{}^+\mathbb{P}^B_{j+1}\,\kappa_{j\rightarrow j+1}(x_{j+1},x_j)\right)\prod^{n-1}_{l=1}dx_l\nn\\
&{=}&\int_{\cq^{n-1}}e^{i/\hbar\sum^n_{m=1}S_m(x_m,x_{m-1})}\prod^{n-1}_{j=0}M_{j\rightarrow j+1}(x_{j+1},x_j)\prod^{n-1}_{l=1}d\xi_l(x_l).\label{regpi}
\ea
The expression in (\ref{convpip}) manifests that this construction of a PI for constrained global evolution moves is essentially a (projected) sequence of PIPs.

A few comments concerning the last expression for the regularized PI are in place:
\begin{itemize}
\item The definition is clearly formal. For instance, even apart from the formal regularized measures $d\xi_l$ and possible ordering ambiguities of the $\delta(\hat{C}^j)$ within the projectors, the conditions on the measure $M_{j\rightarrow j+1}$ may not in general uniquely determine it.
\item From the second line it can be seen that every (improper) projector on the (`bare') constraint at each step (except at $n=0$) is implemented precisely once to yield a finite state sum. We recall from section \ref{sec_effcon}, however, that in general new `effective' quantum constraints will arise when integrating out intermediate time steps. These effective constraints may require additional regularizations in the above PI and the corresponding projector must only be implemented once. For instance, if an `effective' post--constraint coincides with a `bare' pre--constraint at some step $n=j$, then its projector must be removed in order to obtain a finite result.
\item Integrating (\ref{regpi}) over step $0$ (and inserting ${}^-\mathbb{P}_0$) yields as projector $\ch^{\rm kin}_0\rightarrow{}^+\tilde{\ch}^{\rm phys}_n$ onto all post--constraints at $n$. We comment further on this in the context of simplicial gravity in section \ref{sec_qg}.
\item The construction of the PI in terms of kinematical propagators and (improper) projectors is convenient because it easily allows one to (at least formally) keep track of the divergences.
%\item In quantum gravity the PI and thus time evolution is expected to act as a projector onto solutions to the quantum diffeomorphism and Hamiltonian constraints \cite{Halliwell:1990qr,Rovelli:1998dx,Noui:2004iy,Thiemann:2013lka}. On the other hand, if constructed via (\ref{regpi}), the PI for discrete systems will project onto solutions to all post--constraints. However, for discrete quantum gravity models the diffeomorphism symmetry of the continuum is generically broken \cite{Bahr:2009ku,biancaval} and the set of pre-- and post--constraints does not in general include the Hamiltonian and diffeomorphism constraints \cite{Dittrich:2009fb,Dittrich:2011ke,Dittrich:2013jaa}. In the discrete, the latter do not in general arise as exact pre-- or post--constraints but rather as approximate or `pseudo'-constraints \cite{Dittrich:2009fb,Dittrich:2011ke,gambini}. In this case, the PI can only be expected to act as an approximate projector onto solutions to the Hamiltonian and diffeomorphism constraints for large scales \cite{biancadiffeo,Dittrich:2013xwa}. 
%\item As argued in \cite{Dittrich:2013xwa}, 
\end{itemize}

%It is reasonable to have the same regularized measure $d\xi$ for the PIP as for the physical state evolution, given that $K_{0\rightarrow1}$ must satisfy both ${}^+\hat{C}^1_I$ and ${}^-\hat{C}^0_J$. In a sense, one can view $K_{0\rightarrow1}\propto{}^+\psi^{\rm phys}_1({}^-\phi^{\rm phys}_0)^*$. Note, however, that this does not mean that the PIP measure and the path integral (PI) measure are identical. $d\xi_1$ is the PIP measure while the fixed PI measure is $dM=d\xi_0M_{0\rightarrow1}$.

\section{Temporally varying discretization and Dirac observables}\label{sec_dirac}

The evolution move dependence of the pre-- and post--physical Hilbert spaces at a given time step $n$ has severe consequences for the physical degrees of freedom at step $n$ as embodied by the Dirac observables. These too become move dependent. 

We briefly summarize the situation in the classical formalism as this helps to understand the situation in the quantum theory (for a detailed discussion see \cite{Dittrich:2013jaa,Hoehn:2014aoa}). While in the continuum the Dirac observables as propagating degrees of freedom can be determined from a constraint analysis at one instant of time, in the discrete two discrete time steps are necessary to have a notion of propagation. The global Hamiltonian time evolution map of the move $0\rightarrow1$ is well-defined and invertible on the space of pre-- and post--orbits $\mathfrak{H}_0:\cc^-_0/\cg^-_0\rightarrow \cc^+_1/\cg^+_1$. {\it Pre--observables} at $n=0$ are functions $O^-_0$ on the pre--constraint surface $\cc^-_0$ which weakly Poisson commute with all pre--constraints $\{O^-_0,{}^-C^0\}\simeq0$ on $\cc^-_0$. The pre--constraints alone form a first class constraint set. Similarly, {\it post--observables} at $n=1$ are functions $O^+_1$ on the post--constraint surface $\cc^+_1$ which weakly Poisson commute with all post--constraints $\{O^+_1,{}^+C^1\}\simeq0$ on $\cc^+_1$. The post--constraints alone likewise are first class. $\mathfrak{H}_0$ maps the pre--observables bijectively into the post--observables. 

But for a temporally varying discretization the pre--constraint surface $\cc^-_0$ and thus the Poisson algebra of pre--observables at $n=0$ depend on the evolution move. For an effective evolution move $0\rightarrow2$ new effective pre--constraints at $n=0$ may arise which eliminate pre--observables of the move $0\rightarrow1$. The same, of course, holds true for the post--constraints and post--observables at any given step. That is, the pre-- and post--observables in the discrete are always associated to a fixed evolution move rather than a single time step. Since a time evolution move may change the discretization at a given time step it may change the number of propagating degrees of freedom.

The situation in the quantum theory is completely analogous. Consider an evolution move $0\rightarrow1$. In order for an operator to be well-defined on a physical Hilbert space, it must commute with all quantum constraints. We shall call
%\begin{Definition}{\bf(\emph{Dirac pre-- and post--observables})}
\begin{itemize}
\item an operator $\hat{O}^-_0$ which commutes with all quantum pre--constraints at $n=0$, $[\hat{O}^-_0,{}^-\hat{C}^0_I]=0$, $\forall\, I$, a \emph{quantum pre--observable} on the pre--physical Hilbert space ${}^-\ch^{\rm phys}_0$. 
\item An operator $\hat{O}^+_1$\  which commutes with all quantum post--constraints at $n=1$, $[\hat{O}^+_1,{}^+\hat{C}^1_I]=0$, $\forall\, I$, a \emph{quantum post--observable} on the post--physical Hilbert space ${}^+\ch^{\rm phys}_1$. 
\end{itemize}
%\end{Definition}

The quantum pre--observables at $n=0$ and the quantum post--observables at $n=1$ are in one-to-one correspondence. Indeed, we can employ the unitary isomorphisms $U_{0\rightarrow1}:{}^-\ch^{\rm phys}_0\rightarrow{}^+\ch^{\rm phys}_1$ and $U_{1\rightarrow0}:{}^+\ch^{\rm phys}_1\rightarrow{}^-\ch^{\rm phys}_0$ (see section \ref{sec_unitary}) to map quantum pre--observables to quantum post--observables and vice versa:
\ba
\hat{O}^+_1=U_{0\rightarrow1}\,\hat{O}^-_0\,U_{1\rightarrow0},\q\q\q
\hat{O}^-_0=U_{1\rightarrow0}\,\hat{O}^+_1\,U_{0\rightarrow1}.\nn
\ea
$U_{0\rightarrow1}\,\hat{O}^-_0\,U_{1\rightarrow0}$ constitutes a well-defined map ${}^+\ch^{\rm phys}_1\rightarrow{}^+\ch^{\rm phys}_1$ since
\ba
[U_{0\rightarrow1}\,\hat{O}^-_0\,U_{1\rightarrow0},{}^+\hat{C}^1_I]\,{}^+\psi^{\rm phys}_1=U_{0\rightarrow1}\,\hat{O}^-_0\,U_{1\rightarrow0}\underset{=0}{\underbrace{{}^+\hat{C}^1_I\,{}^+\psi^{\rm phys}_1}}-\underset{=0}{\underbrace{{}^+\hat{C}^1_I\,U_{0\rightarrow1}}}\,\hat{O}^-_0\,U_{1\rightarrow0}\,{}^+\psi^{\rm phys}_1=0.\nn
\ea
Similarly, $U_{1\rightarrow0}\,\hat{O}^+_1\,U_{0\rightarrow1}$ establishes a well-defined map ${}^-\ch^{\rm phys}_0\rightarrow{}^-\ch^{\rm phys}_0$ because
\ba
[U_{1\rightarrow0}\,\hat{O}^+_1\,U_{0\rightarrow1},{}^-\hat{C}^0_I]\,{}^-\psi^{\rm phys}_0=0.\nn
\ea
These maps can also be used to construct a `Heisenberg picture' for systems with temporally varying discretization which evolves observables rather than physical states. 

However, just like in the classical case, ${}^-\ch^{\rm phys}_0$ and ${}^+\ch^{\rm phys}_1$ depend on the evolution move. For instance, as discussed in section \ref{sec_effcon}, new effective quantum pre--constraints at $n=0$ may appear in the move $0\rightarrow2$ such that the pre--physical Hilbert space at $n=0$ changes to ${}^-\mathbb{\tilde P}_0({}^-\ch^{\rm phys}_0)$ and becomes evolution move dependent. But correspondingly the quantum pre--observable algebra which commutes with all pre--constraints at $n=0$ must change. Only those quantum pre--observables survive which also commute with the new effective pre--constraints. The remaining set of quantum pre--observables can then be mapped under $U_{0\rightarrow2}:{}^-\mathbb{\tilde P}_0({}^-\ch^{\rm phys}_0)\rightarrow{}^+\mathbb{\tilde P}_2({}^+\ch^{\rm phys}_2)$ to a surviving set of quantum post--observables at step $2$. Of course, the same reasoning holds also for post--constraints and quantum post--observables on post--Hilbert spaces at an arbitrary step. Just as in the classical case, the pre-- and post--observables at some step $n$ are therefore associated to a given evolution move---as indicated above---rather than step $n$ only. Different evolution moves will be associated with different pre-- and post--observable algebras such that one has a genuinely varying number of physical degrees of freedom in the course of evolution.

Classically, for the composition of the moves $0\rightarrow1$ and $1\rightarrow2$ one can define the notion of a reduced phase space which corresponds to $\cc^-_1\cap\cc^+_1/\cg_1$, where $\cg_1$ is the first class orbit at $n=1$ \cite{Dittrich:2013jaa}. The reduced phase space corresponds to the set of observables which propagates from $0$ via $1$ to $2$. In particular, any observable $O_1=O^+_1=O^-_1$ which is both a post--observable of the move $0\rightarrow1$ and a pre--observable of the move $1\rightarrow2$ corresponds to such data that propagates {\it through} step $1$. But due to the second class constraints there are other possibilities \cite{Dittrich:2013jaa,Hoehn:2014aoa}. 

In the quantum theory, the situation as regards the composition of moves is clear: any quantum post--observable $\hat{O}^+_1$ that can be mapped to a quantum pre--observable $\hat{O}^-_0$ via the move $0\rightarrow1$ must commute with all post--constraints $[\hat{O}^+_1,{}^+\hat{C}^1_I]=0$. Likewise, any pre--observable $\hat{O}^-_1$ that can be mapped to a post--observable $\hat{O}^+_2$ via the move $1\rightarrow2$ must commute with all pre--constraints $[\hat{O}^-_1,{}^-\hat{C}^1_J]=0$. Accordingly, any observable $\hat{O}_1$ that can be mapped using $U_{0\rightarrow1}$ to a pre--observable $\hat{O}^-_0$ {\it and} using $U_{1\rightarrow2}$ to a post--observable $\hat{O}^+_2$,
\ba
\hat{O}_1=U_{0\rightarrow1}\,\hat{O}^-_0\,U_{1\rightarrow0},\q\q\q \hat{O}^+_2=U_{1\rightarrow2}\,\hat{O}_1\,U_{2\rightarrow1},\nn
\ea
must be both a quantum pre-- and post--observable and commute with all constraints at step $1$
\ba
[\hat{O}_1,{}^+\hat{C}^1_I]=0=[\hat{O}_1,{}^-\hat{C}^1_J].\nn
\ea
In analogy to the classical case, this corresponds to degrees of freedom that propagate from time $0$ {\it through} step $1$ to step $2$ since the composition yields
\ba
 \hat{O}^+_2=U_{1\rightarrow2}\,{}^-\mathbb{P}^B_1\,U_{0\rightarrow1}\,\hat{O}^-_0\,U_{1\rightarrow0}\,{}^+\mathbb{P}^B_1\,U_{2\rightarrow1}.\nn
\ea
The insertion of the projectors ${}^-\mathbb{P}^B_1,{}^+\mathbb{P}^B_1$ is necessary for the reasons discussed in section \ref{sec_comp}. The coinciding pre-- and post--observables $\hat{O}_1$ are thus well-defined operators on the single physical Hilbert space $\ch^{\rm phys}_1$ obtained after matching the pre-- and post--physical states at $n=1$ (see procedure (i) in section \ref{sec_comp}).

Viewing the discrete time evolution as coarse graining, refining or entangling operations \cite{Dittrich:2013xwa}, the present discussion shows how this can affect the algebra of quantum pre-- and post--observables at a given time step and in the course of evolution. In particular, a coarse graining time evolution move leads to additional constraints that irreversibly project out physical degrees of freedom. These issues are further discussed in \cite{Hoehn:2014wwa} in the context of local evolution moves.

\begin{figure}[hbt!]
\begin{center}
\psfrag{0}{$0$}
\psfrag{n}{$n$}
\psfrag{nx}{$n+x$}
\psfrag{no}{\large `Nothing'}
\hspace*{-3.5cm}\begin{subfigure}[b]{.22\textwidth}
\centering
\includegraphics[scale=.3]{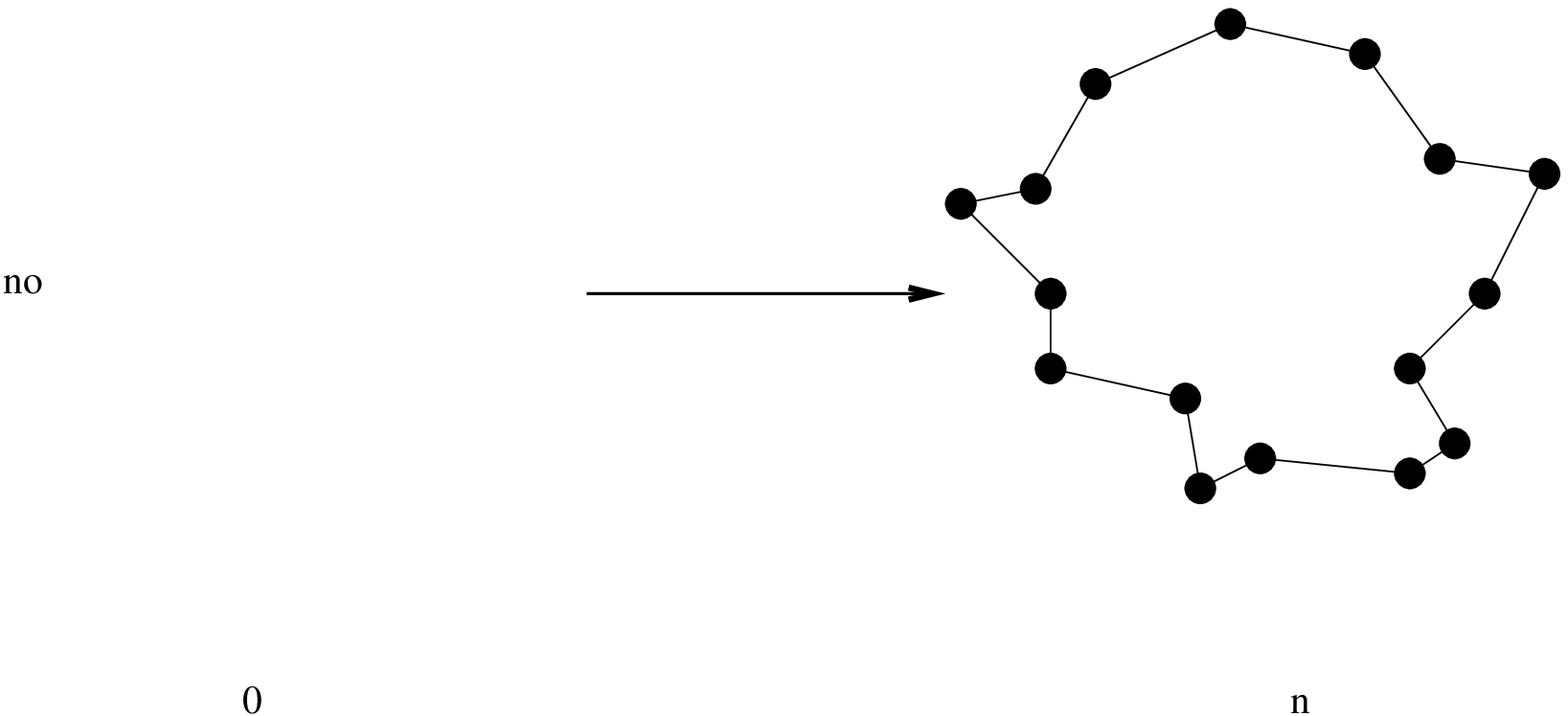}
\centering
\caption{\small }
\end{subfigure}
\hspace*{4.5cm}
\begin{subfigure}[b]{.22\textwidth}
\centering
\includegraphics[scale=.3]{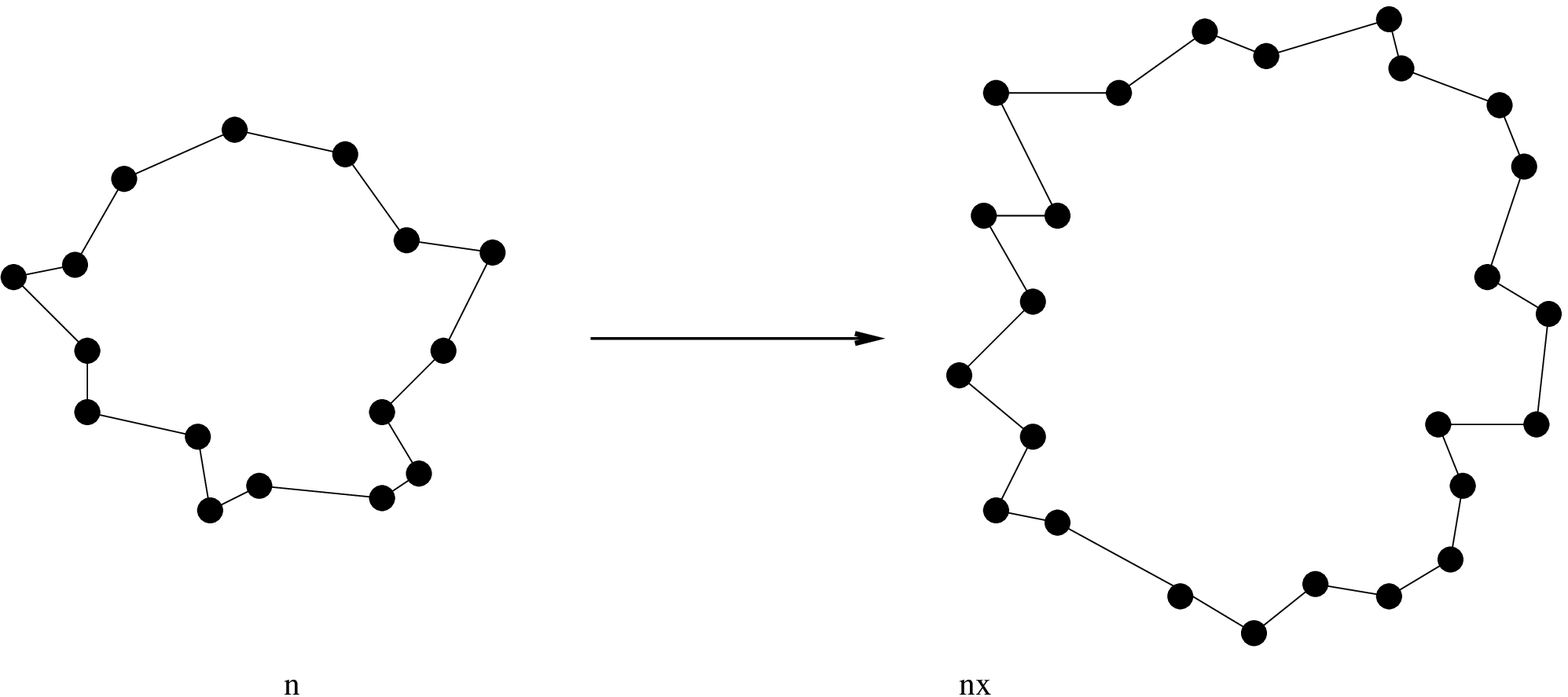}
\centering
\caption{\small }
\end{subfigure}
\caption{\small Schematic illustration of a discrete version of the `no boundary' proposal. (a) A move $0\rightarrow n$ from an empty triangulation to some spherical hypersurface is devoid of pre--observables at $0$ and post--observables at $n$. (b) A further move $n\rightarrow n+x$, however, may feature pre--observables at $n$.}\label{fig_nb}
\end{center}
\end{figure}

Finally, let us briefly comment on the discrete incarnation \cite{Dittrich:2013jaa} of the `no boundary' proposal \cite{Hartle:1983ai} for a quantum gravity vacuum state. This corresponds to starting with an empty space-time triangulation at $n=0$ and evolving to a triangulated spherical `spatial' hypersurfaces at later time steps $1,2,\ldots$ (see figure \ref{fig_nb} for an illustration). Any move $0\rightarrow n$ is totally constrained \cite{Dittrich:2011ke,Dittrich:2013jaa}. This implies that there can be no non-trivial quantum pre--observables $\hat{O}^-_0$ at step $0$ and no non-trivial quantum post--observables $\hat{O}^+_n$ at the future step $n$. This corresponds to the absence of propagating degrees of freedom for any evolution move $0\rightarrow n$ \cite{Dittrich:2013jaa}. Such an evolution can be viewed as refining the discretization of the evolving spatial hypersurface by only adding vacuum degrees of freedom \cite{Dittrich:2013xwa}. The corresponding unique Hartle-Hawking state is then a unique vacuum state associated to the move $0\rightarrow n$. However, the absence of non-trivial post--observables $\hat{O}^+_n$ at a step $n$ does not imply that there are no non-trivial pre--observables $\hat{O}^-_n$ at the same step which may propagate from $n$ under an evolution move $n\rightarrow n+x$. The physical states associated to such a move $n\rightarrow n+x$ can then also no longer be unique. We shall now illustrate this explicitly in a toy model.

\section{Toy model: `creation from nothing'}\label{sec_ex}

For an explicit illustration of the formalism we shall consider a toy model for a discrete version \cite{Dittrich:2013jaa} of the `no--boundary' proposal \cite{Hartle:1983ai}. Namely, we shall consider a free scalar field on the vertices of a 2D lattice which evolves from `nothing' to a two-- and then to a four--dimensional phase space at the subsequent time steps. The evolution moves of this toy model are depicted in figure \ref{toy1}. In this simple example the measure will be uniquely determined (up to unitary phase changes). This toy model also serves as a concrete example to the discussion in sections \ref{sec_effcon} and \ref{sec_dirac}, showing how the post--physical Hilbert space and quantum pre-- and post--observables at a given step depends on the evolution move under consideration. The physical post--states for the moves $0\rightarrow1$ and $0\rightarrow2$ will be unique because these moves are totally constrained. These unique physical states can be viewed as a vacuum of non-propagating degrees of freedom (see also \cite{Dittrich:2013xwa}).%This feature is general and will occur in any evolution from `nothing' to a finite number of variables as such an evolution is fully constrained \cite{Dittrich:2013jaa}. %We keep it simple in order to be able to explicitly solve the model.

\begin{figure}[hbt!]
\begin{center}
\psfrag{0}{$n=0$}
\psfrag{1}{$n=1$}
\psfrag{2}{$n=2$}
\psfrag{n}{\large `Nothing'}
\psfrag{f1}{$\varphi_1$}
\psfrag{f11}{$\varphi_2^1$}
\psfrag{f12}{$\varphi_2^2$}
\hspace*{-0.5cm}\begin{subfigure}[b]{.22\textwidth}
\centering
\includegraphics[scale=.45]{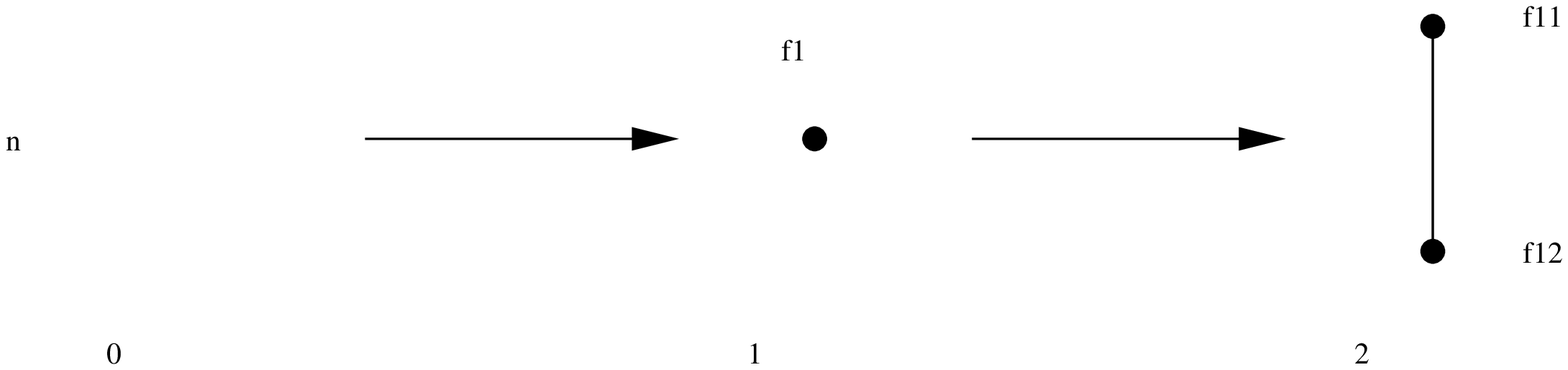}
\centering
\caption{\small }
\end{subfigure}
\hspace*{8.5cm}
\begin{subfigure}[b]{.22\textwidth}
\centering
\includegraphics[scale=.45]{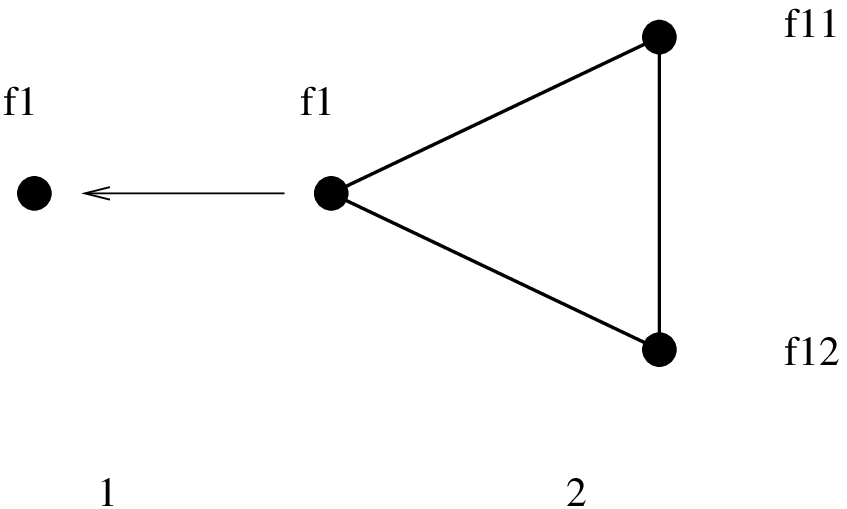}
\centering
\caption{\small }
\end{subfigure}
\caption{\small The toy model of a scalar field on the vertices of a 2D lattice for a `creation from nothing'. (a) The move $0\rightarrow1$ starts from an empty set and introduces one vertex with one field variable $\varphi_1$ at $n=1$. The move $1\rightarrow2$ maps to a field configuration on two vertices at $n=2$. (b) The move $1\rightarrow2$ is a singular move which corresponds to gluing a triangle onto a single vertex. }\label{toy1}
\end{center}
\end{figure}

\subsection{The move $0\rightarrow1$}

We associate the `vertex action'
\ba
S_1=(\vphi_1)^2, \nn
\ea
to the evolution move $0\rightarrow1$ from `nothing' to a two-dimensional phase space at $n=1$. One can extend the empty $\cq_0=\emptyset$ to $\overline{\cq}_0\simeq\mathbb{R}$, coordinatized by $\varphi_0$. Classically, the trivial dependence of $S_1$ on $\varphi_0$ leads to a pre-- and post--constraint
\ba 
{}^-C^0_1=p^0_1\q\q\q{}^+C^1_1=p^1_1-2\varphi_1.\nn
\ea
$\varphi_0$ and $\varphi_1$ are thus an {\it a posteriori} and {\it a priori} free parameter, respectively.

The pre--physical state at $n=0$ is constant,
$%\ba\label{ex-states1}
{}^-\psi^{\rm phys}_0=const.
$ %\ea
Using group averaging (\ref{proj}) and the Baker-Campbell-Hausdorff formula, one finds the post--physical state at $n=1$
\ba\label{ex-ppstate1}
{}^+\psi^{\rm phys}_1(\vphi_1)=\delta({}^+\hat{C}^1_1)\,\psi^{\rm kin}_1\!\!\!&=&\!\!\!\!\f{1}{2\pi\hbar}\int ds\,e^{is\left(\hat{p}^1_1-2\vphi^1\right)/\hbar}\,\psi^{\rm kin}_1(\vphi_1)
%&=&\f{1}{2\pi}\int ds\,e^{-isa(s+2x_1^1)}e^{is\hat{p}^1_1}\,\psi^{\rm kin}_1(x_1^1)\nn\\
%&=&\f{1}{2\pi}\int ds\,e^{-isa(s+2x_1^1)}\,\psi^{\rm kin}_1(x_1^1+s)\nn\\
=\f{1}{2\pi\hbar}\,e^{i(\vphi_1)^2/\hbar}\int dz\,e^{-iz^2/\hbar}\,\psi^{\rm kin}_1(z)\nn\\
&=&const\times e^{i(\vphi_1)^2/\hbar}.
\ea
(We assume $\psi^{\rm kin}_1$ to be finite under a complex Gauss transform.) %In the second line we employed the Baker-Campbell-Hausdorff formula. 
Note that both states are {\it unique} (up to normalization) such that ${}^-\ch^{\rm phys}_0$ and ${}^+\ch^{\rm phys}_1$ of this move are one-dimensional. Accordingly, there are no non-trivial quantum pre-- and post--observables associated to this move.% on account of the fact that the move $0\rightarrow1$ is totally constrained. %(\ref{ex-states1}) can also be obtained via group averaging (\ref{proj}), provided that the kinematical states are integrable.

Thanks to lemma \ref{lem2} one can write the PIP-- at $n=0$ with Faddeev-Popov regularized pre--measure $d\xi^-_0=(2\pi\hbar)\,d\vphi_0\,\delta(\vphi_0-c_0)$ as
\ba
\langle{}^-\psi^{\rm phys}_0|{}^-\psi^{\rm phys}_0\rangle_{\rm phys-}=(2\pi\hbar)\int d\vphi_0\,\delta(\vphi_0-c_0)({}^-\psi^{\rm phys}_0)^*{}^-\psi^{\rm phys}_0\overset{!}{=}1,\nn
\ea
%which evidently is independent of the choice of gauge $\vphi_0=c_0$ in the . (In this trivial case $\Delta_{FP}^0=1$.) Notice that this measure is probabilistic thereby implementing cylindrical consistency. For normalization 
such that we choose ${}^-\psi^{\rm phys}_0=\f{1}{\sqrt{2\pi\hbar}}\,e^{i\theta/\hbar}$, with phase $\theta=const$. Similarly, (\ref{ex-ppstate1}) allows us to write the PIP+ at $n=1$ with Faddeev-Popov regularized post--measure $d\xi_1^+=(2\pi\hbar)\,\delta(\vphi_1-c_1)\,d\vphi_1$
\ba
\langle{}^+\psi^{\rm phys}_1|{}^+\psi^{\rm phys}_1\rangle_{\rm phys+}\!\!\!\!&=&\!\!\!\!\int d\vphi_1\,(\psi^{\rm kin}_1(\vphi_1))^*e^{i(\vphi_1)^2/\hbar}\f{1}{2\pi\hbar}\int dz e^{-iz^2/\hbar}\,\psi^{\rm kin}_1(z)\nn\\
&=&\!\!\!\!\!\f{1}{2\pi\hbar}\Big|\int d\vphi_1 e^{-i(\vphi_1)^2/\hbar}\,\psi^{\rm kin}_1(\vphi_1)\Big|^2\!\!\!=(2\pi\hbar)\int d\vphi_1\delta(\vphi_1-c_1)({}^+\psi^{\rm phys}_1)^*{}^+\psi^{\rm phys}_1.\nn
\ea
%. (Again, $\Delta_{FP}^1=1$.) The PIP+ does not depend on the choice of $c_1$ since $({}^+\psi^{\rm phys}_1)^*{}^+\psi^{\rm phys}_1$ is independent of $\vphi_1^1$. For normalization w
We therefore choose ${}^+\psi^{\rm phys}_1=\f{1}{\sqrt{2\pi\hbar}}\,e^{i(\vphi_1)^2/\hbar}$.

Given that the propagator has to satisfy both the pre-- and post--constraint, it reads
\ba
K_{0\rightarrow1}={}^+\psi^{\rm phys}_1\left({}^-\psi^{\rm phys}_0\right)^*=\f{1}{{2\pi\hbar}}\,e^{i(S_1-\theta)/\hbar}=\f{1}{{2\pi\hbar}}\,e^{i\left((\vphi_1)^2-\theta\right)/\hbar}.\nn
\ea
(Compare this with example \ref{ex_ex}.) The pre--fixed propagator $K_{0\rightarrow1}^{f_+}={\delta(\vphi_0-c_0)}\,e^{i\left((\vphi_1)^2-\theta\right)/\hbar}$ %The propagator measure is therefore constant: $M_{0\rightarrow1}=1/\sqrt{2\pi}e^{-i\varphi}$.
follows from lemma \ref{lem3}. Evidently, ${}^+\psi^{\rm phys}_1(\vphi_1)=\int d\vphi_0\,K^{f_+}_{0\rightarrow1}{}^-\psi^{\rm phys}_0$. Unitarity of this move $0\rightarrow1$ is a tautology, given that both unique states are normalized. 
%Nevertheless, using the pre--fixed propagator, one consistently finds
%\ba
%\langle{}^+\psi^{\rm phys}_1|{}^+\psi^{\rm phys}_1\rangle_{\rm phys+}&=&2\pi\int d\vphi_1^1\delta(\vphi_1^1-c_1)({}^+\psi^{\rm phys}_1)^*{}^+\psi^{\rm phys}_1\nn\\
%&=&2\pi\int d\vphi_1^1d\vphi_0^1\delta(\vphi_1^1-c_1)\f{1}{\sqrt{2\pi}}\, e^{-ia(\vphi_1^1)^2}\f{\delta(\vphi_0^1-c_0)}{\sqrt{2\pi}}e^{i\left(a(\vphi_1^1)^2-\varphi\right)}\,e^{i\varphi}\nn\\
%&=&\int d\vphi_0^1\delta(\vphi_0^1-c_0)=\langle{}^-\psi^{\rm phys}_0|{}^-\psi^{\rm phys}_0\rangle_{\rm phys-}
%\ea
%{\bf[is unitarity for a totally constrained move even necessary? Trivial anyway, given that both unique states are normalized, so unitarity a tautology!]}

\subsection{The move $1\rightarrow2$}

The action of the triangle of the move $1\rightarrow2$ which evolves the system from a two-- to a four--dimensional phase space at $n=2$ is (up to an overall factor $\f{1}{2}$ which we ignore)
\ba
S_2=(\vphi_1)^2+(\vphi_2^1)^2+(\vphi_2^2)^2-\vphi_1\vphi_2^1-\vphi_1\vphi_2^2-\vphi_2^1\vphi_2^2.\nn
\ea
Extending $\cq_1\simeq\mathbb{R}$ to $\overline{\cq}_1\simeq\mathbb{R}^2$ where the auxiliary dimension is coordinatized by $\vphi_1^2$, we have $\dim\overline{\cq}_1=\dim\cq_2=2$. $S_2$ does not depend on $\vphi_1^2$ such that a pre-- and post--constraint arise
\ba
{}^-C^1_2=p^1_2,\q\q\q\q{}^+C^2=p^2_1-p^2_2-3(\vphi_2^1-\vphi_2^2).\nn
\ea
It is convenient to perform a variable transformation at $n=2$
\ba
u_2=\vphi_2^1-\vphi_2^2,\q\q v_2=\vphi_2^1+\vphi_2^2,\q\q p^2_u=\f{1}{2}(p^2_1-p^2_2),\q\q p^2_v=\f{1}{2}(p^2_1+p^2_2).\nn
\ea
In these variables, 
\ba
S_2(\vphi_1,u_2,v_2)&=&(\vphi_1)^2+\f{3}{4}(u_2)^2+\f{1}{4}(v_2)^2-\vphi_1v_2,\label{ex-s2}\\
{}^+C^2_u&=&p^2_u-\f{3}{2}\,u_2.\label{ex-cu}
\ea
We thus have the {\it a posteriori} free (auxiliary) $\vphi_1^2$ and the {\it a priori} free $u_2$.

In the quantum theory, the pre--physical state is given by an arbitrary (square integrable) function of $\vphi_1$, ${}^-\psi^{\rm phys}_1={}^-{\psi}^{\rm phys}_1(\vphi_1)$. Group averaging yields the post--physical state:
\ba
{}^+\psi^{\rm phys}_2(u_2,v_2)=\delta({}^+\hat{C}^2_u)\,\psi^{\rm kin}_2&=&%\f{1}{2\pi}\int ds\, e^{is(\hat{p}^2_u-du_2)}\,\psi^{\rm kin}_2(u_2,v_2)\nn\\
%&=&\f{1}{2\pi}\int ds\, e^{-isd(s/2+u_2)}e^{is\hat{p}^2_u}\,\psi^{\rm kin}_2(u_2,v_2)\nn\\
%&=&\f{1}{2\pi}\int ds\, e^{-isd(s/2+u_2)}\,\psi^{\rm kin}_2(u_2+s,v_2)\nn\\
\f{1}{2\pi\hbar}\,e^{\f{3i}{4\hbar}(u_2)^2}\int dz\, e^{-\f{3i}{4\hbar}z^2}\,\psi^{\rm kin}_2(z,v_2)=e^{\f{3i}{4\hbar}(u_2)^2}\,\overline{\psi}^{\rm phys}_2(v_2),\nn
\ea
where $\overline{\psi}^{\rm phys}_2(v_2)$ is an arbitrary (square integrable) function of $v_2$. % In the second line, we have made use of the Baker-Campbell-Hausdorff formula. 
The pre-- and post--physical states of the move $1\rightarrow2$ are therefore non-unique and ${}^-\ch^{\rm phys}_1\simeq L^2(\mathbb{R},d\vphi_1)$, while ${}^+\ch^{\rm phys}_2\simeq L^2(\mathbb{R},dv_2)$. Indeed, using lemma \ref{lem2}, the PIP-- at $n=1$ is given by
\ba
\langle{}^-\psi^{\rm phys}_1\Big|{}^-\phi^{\rm phys}_1\rangle_{\rm phys-}%&=&\int d\vphi_1^1d\vphi_1^2(\psi^{\rm kin}_1(\vphi_1^1,\vphi_1^2))^*\,{}^-{\phi}^{\rm phys}_1(\vphi_1^1)\nn\\
%&\underset{\text{Lemma \ref{lem2}}}{=}&\int d\vphi_1^1d\vphi_1^2\delta(\vphi_1^2-c_1^2)({}^-{\psi}^{\rm phys}_1(\vphi_1^1))^*\,{}^-{\phi}^{\rm phys}_1(\vphi_1^1)\nn\\
%&=&
=\int d\xi^-_1\,({}^-{\psi}^{\rm phys}_1(\vphi_1))^*\,{}^-{\phi}^{\rm phys}_1(\vphi_1)\label{pip-}
\ea
with cylindrical pre--measure $d\xi^-_1=(2\pi\hbar)\,d\vphi_1\,d\vphi_1^2\,\delta(\vphi_1^2-c_1^2)$. Likewise, the PIP+ at $n=2$ is 
\ba
\langle{}^+\psi^{\rm phys}_2\Big|{}^+\phi^{\rm phys}_2\rangle_{\rm phys+}&=&\f{1}{2\pi\hbar}\int du_2\,dv_2\,(\psi^{\rm kin}_2(u_2,v_2))^*\,e^{\f{3i}{4\hbar}(u_2)^2}\int dz\, e^{-\f{3i}{4\hbar}z^2}\,\phi^{\rm kin}_2(z,v_2)\nn\\
%&=&{2\pi}\int dv_2 (\tilde{\psi}^{\rm phys}_2(v_2))^*\,\tilde{\phi}^{\rm phys}_2(v_2)\nn\\
%&=&{2\pi}\int dv_2({}^+\psi^{\rm phys}_2(u_2,v_2))^*\,{}^+\phi^{\rm phys}_2(u_2,v_2)\nn\\
&=&\int d\xi^+_2(u_2,v_2)({}^+\psi^{\rm phys}_2(u_2,v_2))^*\,{}^+\phi^{\rm phys}_2(u_2,v_2),\nn
\ea
with Faddeev-Popov regularized post--measure $d\xi^+_2=({2\pi\hbar})\,\delta(u_2-c_2)\,du_2\,dv_2$. %Notice that this measure is (up to $2\pi$) probabilistic in the {\it a priori} free variable $u_2$.%---a feature of cylindrical consistency.

The independent quantum pre--observables of $1\rightarrow2$ are $\hat{\vphi}_1,\hat{p}^1$, while the  quantum post--observables are $\hat{v}_2,\hat{p}^2_v$.

For the propagator $K_{1\rightarrow2}$ we require
\ba
{}^+\hat{C}^2_u\,K_{1\rightarrow2}=0={}^-\hat{C}^1_2\,(K_{1\rightarrow2})^*,\nn
\ea
such that
\ba
K_{1\rightarrow2}(\vphi_1,u_2,v_2)=\f{1}{2\pi\hbar}\,e^{\f{3i}{4\hbar}(u_2)^2}\,f_{1\rightarrow2}(v_2,\vphi_1),\nn
\ea
where $f_{1\rightarrow2}$ is a function which we shall now determine. Thanks to lemma \ref{lem3} the (Faddeev-Popov) pre-- and post--fixed propagators read%\footnote{For the pre--fixed propagator we absorb the factor of $(2\pi\hbar)$ into ${}^-\psi^{\rm phys}_1$ due to (\ref{pip-}).}
\ba\label{ex-pfprop}
K^{f_+}_{1\rightarrow2}=(2\pi\hbar)\,\delta(\vphi_1^2-{\vphi'}_1^2)\,K_{1\rightarrow2},\q\q\q\q K^{f_-}_{1\rightarrow2}=(2\pi\hbar)\,\delta(u_2-u'_2)\,K_{1\rightarrow2}.\nn
\ea
respectively. These must satisfy (\ref{coninvert}) in the form
\ba
\delta(u_2-u'_2)\,\delta(v_2-v'_2)&\overset{!}{=}&\int d\vphi_1\,d\vphi_1^2\,K^{f_+}_{1\rightarrow2}(\vphi_1,u_2,v_2)\,\left(K^{f_-}_{1\rightarrow2}(\vphi_1,u'_2,v'_2)\right)^*\nn\\
\delta(\vphi_1-{\vphi'}_1)\,\delta(\vphi_1^2-{\vphi'}_1^2)&\overset{!}{=}&\int du_2\,dv_2\,K^{f_-}_{1\rightarrow2}({\vphi'}_1,u_2,v_2)\left(K^{f_+}_{1\rightarrow2}(\vphi_1,u_2,v_2)\right)^*.\nn%\nn\\
%&=&2\pi\int du_2dv_2\delta(u_2-u'_2)K_{1\rightarrow2}({x'}_1^1,u_2,v_2)\left(\delta(\vphi_1^2-{x'}_1^2)K_{1\rightarrow2}(\vphi_1^1,u_2,v_2)\right)^*\nn\\
%&=&2\pi\delta(\vphi_1^2-{x'}_1^2)\int dv_2\,\kappa_{1\rightarrow2}(v_2,{x'}^1_1)\kappa_{1\rightarrow2}^*(v_2,\vphi_1^1),\label{ex-cond2}
\ea
%where $I,\tilde{I}$ are arbitrary functions---parametrizing the non-uniqueness of the PI measure contained in $\kappa_{1\rightarrow2}$---such that $I(v_2,v_2)=1=\tilde{I}(\vphi_1^1,\vphi_1^1)$. From (\ref{ex-cond1}, \ref{ex-cond2}) we thus obtain the following conditions on $\kappa_{1\rightarrow2}$:
%\ba
%\f{1}{(2\pi)^2}e^{i\alpha \vphi_1^1(v_2-v'_2)}I(v_2,v'_2)&=&\kappa_{1\rightarrow2}(v_2,\vphi_1^1)\kappa_{1\rightarrow2}^*(v'_2,\vphi_1^1),\nn\\
%\f{1}{(2\pi)^2}e^{i\beta v_2(\vphi_1^1-{x'}_1^1)}\tilde{I}(\vphi_1^1,{x'}_1^1)&=&\kappa_{1\rightarrow2}(v_2,{x'}^1_1)\kappa_{1\rightarrow2}^*(v_2,\vphi_1^1)
% \ea
% The last two equations imply $|\kappa_{1\rightarrow2}(v_2,\vphi_1^1)|^2=1/(2\pi)^2$ and thereby $\kappa_{1\rightarrow2}(v_2,\vphi_1^1)=\f{1}{2\pi}e^{ih(v_2,\vphi_1^1)}$ for some real function $h$. In fact, the last two equations entail
% \ba
% h(v_2,\vphi_1^1)=\alpha \vphi_1^1v_2+s(v_2)+t(\vphi_1^1)
% \ea
% for $\alpha\in\mathbb{R}$ and arbitrary real functions $s(v_2), t(\vphi_1^1)$. The propagator thus takes the form
As one can easily check, up to unitary phase changes, these conditions uniquely imply
  \ba
 K_{1\rightarrow2}(\vphi_1,u_2,v_2)=\left(\f{1}{{2\pi\hbar}}\right)^{3/2}\,e^{iS_2(\vphi_1,u_2,v_2)/\hbar}=\left(\f{1}{{2\pi\hbar}}\right)^{3/2}\,e^{i\left((\vphi_1)^2+\f{3}{4}(u_2)^2+\f{1}{4}(v_2)^2-\vphi_1v_2\right)/\hbar}.\label{ex-prop12}\nn
 \ea 
% Using (\ref{prop}) and (\ref{ex-s2}), the (non-fixed) measure for the move $1\rightarrow2$ reads
% \ba
% M_{1\rightarrow2}(\vphi_1^1,v_2,u_2)=\f{1}{2\pi}e^{i\left((\alpha-c) \vphi_1^1v_2+s(v_2)-d/2(v_2)^2+t(\vphi_1^1)-b(\vphi_1^1)^2\right)} .
% \ea
% This measure is clearly non-unique because of the occurrence of $\alpha, s(v_2), t(\vphi_1^1)$. Consequently, there are infinitely many consistent possibilities for implementing the move $1\rightarrow2$. However, as we shall see below, the composition with the move $0\rightarrow1$ uniquely fixes all these ambiguities.
 
 Using ${}^+\psi^{\rm phys}_2=\int d\vphi_1\,d\vphi_1^2\,K^{f_+}_{1\rightarrow2}{}^-\psi^{\rm phys}_1$ and the above equations, it is straightforward to check that $1\rightarrow2$ is unitary
  \ba
  \langle{}^+\psi^{\rm phys}_2\Big|{}^+\phi^{\rm phys}_2\rangle_{\rm phys+}= \langle{}^-\psi^{\rm phys}_1\Big|{}^-\phi^{\rm phys}_1\rangle_{\rm phys-}.\nn
  \ea

%  = 
%Thanks to (\ref{ex-pfprop}, \ref{ex-prop12}), unitarity of the move $1\rightarrow2$ follows straightforwardly,
% \ba
% \langle{}^+\psi^{\rm phys}_2\Big|{}^+\phi^{\rm phys}_2\rangle_{\rm phys+}&=&{2\pi}\int dv_2 ({}^+\psi^{\rm phys}_2(u_2,v_2))^*\,{}^+\phi^{\rm phys}_2(u_2,v_2)\nn\\
% &\underset{(\ref{ex-prop12})}{=}&\int dv_2(\tilde{\psi}^{\rm phys}_2(v_2))^*\int d\vphi_1^1 e^{i\left(\alpha \vphi_1^1v_2+s(v_2)+t(\vphi_1^1)\right)}\,{}^-\phi^{\rm phys}_1(\vphi_1^1)\nn\\
% &=&\int d\vphi_1^1\left(\int dv_2\,e^{-i\left(\alpha \vphi_1^1v_2+s(v_2)+t(\vphi_1^1)\right)}\tilde{\psi}^{\rm phys}_2(v_2)\right)^* \,{}^-\phi^{\rm phys}_1(\vphi_1^1)\nn\\
% &=&\int d\vphi_1^1 ({}^-\psi^{\rm phys}_1(\vphi_1^1))^*\,{}^-\phi^{\rm phys}_1(\vphi_1^1)=\langle{}^-\psi^{\rm phys}_1\Big|{}^-\phi^{\rm phys}_1\rangle_{\rm phys-}.
% \ea
%We have made use of the fact that the factors of $2\pi$ in the propagator and the regularized measure cancel in the reverse transformation between $\,{}^-\phi^{\rm phys}_1(\vphi_1^1)$ and $\tilde{\phi}^{\rm phys}_2(v_2)$:
% \ba
% {}^-\phi^{\rm phys}_1(\vphi_1^1)&=&\int du_2dv_2\left(K^{f_-}_{1\rightarrow2}(\vphi_1^1,u_2,v_2)\right)^*{}^+\phi^{\rm phys}_2(u_2,v_2)\nn\\
% &=&\int dv_2\,e^{-i\left(\alpha \vphi_1^1v_2+s(v_2)+t(\vphi_1^1)\right)}\,\tilde{\phi}^{\rm phys}_2(v_2)\nn
% \ea

 \subsection{Composition to the effective move $0\rightarrow2$}
 
Finally, let us compose the moves $0\rightarrow1$ and $1\rightarrow2$ to an effective move $0\rightarrow2$. Gaussian integration yields for the pre--fixed propagator\footnote{We have applied a phase shift to eliminate an $i$ in the measure.}
%Clearly, $[{}^+\hat{C}^1_1,{}^-\hat{C}^1_2]=0$, such that we may impose the quantum pre-- and post--constraint at $n=1$ simultaneously. Notice that we trivially have ${}^-\hat{C}^1_2\tilde{K}_{0\rightarrow1}=0$, but we must also require\footnote{As in section \ref{sec_comp}, we denote the propagator with a tilde in order to signify that now both pre-- and post--constraints are satisfied.}
%\ba
%0\overset{!}{=}{}^+\hat{C}^1_1(\tilde{K}_{1\rightarrow2})^*.\label{ex-cond3}
%\ea
%From this it follows that $\alpha=0$ and $t(\vphi_1^1)=-S_1(\vphi_1^1)$ such that
%\ba
% \tilde{K}_{1\rightarrow2}(\vphi_1^1,u_2,v_2)=\f{1}{2\pi}\,e^{i\left(d/2(u_2)^2 +s(v_2)-a(\vphi_1^1)^2\right)}. 
% \ea
%Noting that $\vphi_1^1$ is now also a free parameter for $1\rightarrow2$, we find\footnote{The factor of $2\pi$ requires a rescaling by $\sqrt{2\pi}$ of the physical pre--states at $n=1$.} $d\xi_1:=d\xi_1^-=d\xi_1^+=d\vphi_1^1d\vphi_1^2\,2\pi\,\delta(\vphi_1^2-{x'}_1^2)\,\delta(\vphi_1^1-{x'}_1^1)$ and thus
\ba\label{ex-effprop}
{K}^{f_+}_{0\rightarrow2}\!\!\!\!\!&=&\!\!\!\!\!\int d\vphi_1\,d\vphi_1^2\,{K}^{f_+}_{1\rightarrow2}\,{K}^{f_+}_{0\rightarrow1}%\nn\\
%&=&\f{1}{(2\pi)^{3/2}}\,e^{i\left(d/2(u_2)^2 +s(v_2)-\varphi\right)} \int d\vphi_1^1d\vphi_1^2\,2\pi\delta(\vphi_1^1-{x'}^1_1)\delta(\vphi_1^2-{x'}^2_1)\delta(\vphi_0^1-{x'}_0^1)\nn\\
%&=&\f{1}{\sqrt{2\pi}}\,e^{i\left(d/2(u_2)^2 +s(v_2)-\varphi\right)} \delta(\vphi_0^1-{x'}_0^1).
=\f{1}{2}\,e^{\f{i}{8\hbar}(6(u_2)^2+(v_2)^2-8\theta)}\,\delta(\vphi_0-c_0)=\f{1}{2}\,e^{i(\tilde{S}_{02}(u_2,v_2)-\theta)/\hbar}\,\delta(\vphi_0-c_0),\nn
\ea 
where 
\ba
\tilde{S}_{02}=\f{3}{4}(u_2)^2+\f{1}{8}(v_2)^2\nn
\ea
is the classical effective action or Hamilton's principal function (i.e.\ $S_1+S_2$ with $\vphi_1$ integrated out) for the effective move $0\rightarrow2$. Hence, the (effective) post--physical state at $n=2$ for the move $0\rightarrow2$ becomes unique
\ba
{}^+\tilde{\psi}^{\rm phys}_2(u_2,v_2)=\int d\vphi_0\,K^{f_+}_{0\rightarrow2}\,{}^-\psi^{\rm phys}_0=\f{1}{\sqrt{8\pi\hbar}}\,e^{i\tilde{S}_{02}(u_2,v_2)/\hbar}\nn
\ea
and satisfies the new effective post--constraints ${}^+\tilde{C}^2_v=p^2_v-\f{1}{4}v_2$ of the action $\tilde{S}_{02}$
\ba
{}^+\widehat{\tilde C}{}^2_v\,{}^+\tilde{\psi}^{\rm phys}_2=(\hat{p}^2_v-\f{1}{4}v_2)\,\f{1}{\sqrt{8\pi\hbar}}\,e^{i\tilde{S}_{02}(u_2,v_2)/\hbar}=0.\nn
\ea
The move $0\rightarrow2$ is therefore fully constrained---in contrast to $1\rightarrow2$. The post--physical Hilbert space ${}^+\tilde{\ch}^{\rm phys}_2$ of the move $0\rightarrow2$ is thus one-dimensional---in contrast to the post--physical Hilbert space ${}^+\ch^{\rm phys}_2\simeq L^2(\mathbb{R},dv_2)$ of the move $1\rightarrow2$ which is infinite dimensional. Similarly, for the move $0\rightarrow2$ there are no non-trivial quantum pre-- and post--observables---in contrast to $1\rightarrow2$. This provides an explicit example for the discussion in sections \ref{sec_effcon} and \ref{sec_dirac} and illustrates how the post--physical Hilbert space and the Dirac pre-- and post--observables at a given time step depend on the evolution move. ${}^-\ch^{\rm phys}_0$ and ${}^+\tilde\ch^{\rm phys}_2$ can be regarded as representing a unique physical vacuum state without propagating degrees of freedom.

\section{Remarks on the special situation in simplicial gravity models}\label{sec_qg}

So far we have considered general variational discrete systems. %The formalism applies, in particular, to a scalar field on a triangulation. 
However, even apart from the fact that Euclidean configuration spaces $\cq\simeq\mathbb{R}^N$ are not appropriate for quantum gravity models (Euclidean here does {\it not} refer to the space-time signature), there are a number of special properties of gravity one has to take into account when adapting the present formalism to a simplicial gravity model. 

The situation in gravity is special for many reasons. Classically, the dynamics and diffeomorphism symmetry of the continuum theory is generated by the Hamiltonian and diffeomorphism constraints. The Dirac hypersurface deformation algebra of these constraints implies a path independence of the evolution between an initial and final spatial hypersurface \cite{kieferbook}. This constraint structure also entails that in quantum gravity there is no (coordinate) time evolution and physical states are {\it a priori} `timeless' \cite{Kuchar:1991qf,Isham:1992ms,Anderson:2012vk}. Instead, the path integral is expected to act as a projector onto solutions to the quantum Hamiltonian and diffeomorphism constraints \cite{Halliwell:1990qr,Rovelli:1998dx,Noui:2004iy,Thiemann:2013lka}. This, in particular, means that physical states, solving the Hamiltonian and diffeomorphism constraints, do not evolve under the action of a time evolution operator (given by an exponential of the constraints). This apparent `timelessness' notwithstanding, the physical states contain the entire information about the dynamics and a notion of evolution with respect to internal clock degrees of freedom can often be extracted using the relational paradigm of dynamics \cite{Rovelli:2004tv,Rovelli:1989jn,Bojowald:2010xp,Bojowald:2010qw,Hohn:2011us,Tambornino:2011vg}.

For 3D vacuum Regge Calculus \cite{Regge:1961px} (without cosmological constant), albeit being a simplicial gravity model, the situation is analogous. The reason is that the (flat) space-time discretization in terms of a (flat) Regge triangulation is a so-called perfect discretization which preserves the symmetries and dynamics of the continuum \cite{Bahr:2009qc,Bahr:2009mc,Dittrich:2011vz}. This is a consequence of the fact that the 3D theory is special in that it does not contain any local propagating degrees of freedom. Since it is a perfect discretization of the continuum theory, 3D Regge Calculus is also a totally constrained system and features the Hamiltonian and diffeomorphism constraints as pre-- and post--constraints \cite{Dittrich:2011ke}. In fact, the pre--constraints always coincide with the post--constraints in 3D Regge Calculus (i.e.\ they are always of case (a) in section \ref{sec_fullcomp}) such that non-trivial coarse graining pre-- and post--constraints of cases (b1) and (b2) do not arise. This, again, is a consequence of the topological nature of the theory. The system is hyperbolic in the sense that different solutions arising from a given initial data set are equivalent by symmetry transformations \cite{Dittrich:2011ke,Dittrich:2013jaa}. In the discrete such symmetry transformations correspond to vertex translations within the triangulation that can also move a vertex on top of another or split a vertex into two \cite{Rocek:1982fr,Bahr:2009ku,Dittrich:2009fb,Dittrich:2011ke}. In analogy to the continuum, this, in particular, implies a path independence as argued in \cite{Dittrich:2013xwa}: given an initial and a final spatial triangulated hypersurface $\Sigma_i$ and $\Sigma_f$, it does not matter by means of which discrete evolution moves, and thus through which triangulated hypersurfaces, one evolves from $\Sigma_i$ to $\Sigma_f$ (see figure \ref{fig_pi} for an illustration). Thanks to these symmetries and the hyperbolicity it does not matter which spatial triangulations one chooses in the evolution; without loss of generality, in the 3D theory, one may restrict the dynamics to (spatial) triangulation preserving global evolution moves.\footnote{In a local evolution, one would, however, require the full set of Pachner moves \cite{Dittrich:2011ke,Hoehn:2014wwa}.}

In its quantized form as the Ponzano--Regge spin foam model \cite{ponzreg,Freidel:2004vi}, 3D Quantum Regge Calculus features a diffeomorphism symmetry \cite{Freidel:2002dw}, its path integral is a projector onto solutions of the quantum constraints \cite{Noui:2004iy} and the model is triangulation independent \cite{Perez:2012wv,Zapata:2002eu,Dittrich:2011vz}. Hence, also in the quantum theory one encounters a path independence and hyperbolicity of the evolution. Since non-trivial coarse graining constraints of cases (b1) and (b2) of section \ref{sec_fullcomp} do not arise, there will also not occur any non-unitary projection of physical Hilbert spaces and pre-- and post--observables in the 3D quantum theory. Hence, time evolution must always unitarily map between isomorphic pre-- and post--physical Hilbert spaces which, nevertheless, may correspond to different discretizations. In this case, the move dependence of physical Hilbert spaces disappears.

 As argued in \cite{Dittrich:2013xwa}, the discretization changing Hamiltonian time evolution \cite{Dittrich:2011ke,Bonzom:2011hm} can be reconciled with the fact that physical states do not evolve for a totally constrained system by identifying the physical states at different time steps and on different spatial triangulations with one another. The discretization changing time evolution can rather be viewed as a refining, coarse graining or entangling operation which represents one and the same physical state on different discretizations and thereby different Hilbert spaces. From this perspective it is also not useful to distinguish between `forward' and `backward' evolution in quantum gravity---as we have done so far in this work by only considering factors of $e^{iS_1/\hbar}$ (rather than also $e^{-iS_1/\hbar}$) in a propagator $K_{0\rightarrow1}$---because nothing physically changes. Instead, `forward' and `backward' evolution should be simultaneously considered in the quantum theory which leads to a superposition of both\footnote{Superposition of `forward' and `backward' relational evolution in an internal clock degree of freedom is also a generic feature of gravitational systems, see \cite{Bojowald:2010xp,Bojowald:2010qw,Hohn:2011us} for a detailed discussion. However, this superposition of internal time directions originates in the quadratic momentum structure of the Hamiltonian constraint rather than a summation over both positive and negative values of lapse and shift.} and can be viewed as an integration over both positive and negative values of lapse and shift. This, moreover, helps to ensure diffeomorphism symmetry and bears on the tension between `causality' (understood in the sense of time direction) and gauge invariance in quantum gravity elaborated on in \cite{Teitelboim:1983fh} (see also \cite{Dittrich:2013xwa,Hoehn:2014wwa} for a discussion). This is also the reason for the appearance of the Regge action in the cosine (rather than an exponential) in the semiclassical limit of spin foam models \cite{Conrady:2008mk,Barrett:2009gg,Perez:2012wv}. 

\begin{SCfigure}%[hbt!]
%\begin{center}
\psfrag{sf}{$\Sigma_f$}
\psfrag{si}{$\Sigma_i$}
%\centering
\includegraphics[scale=.35]{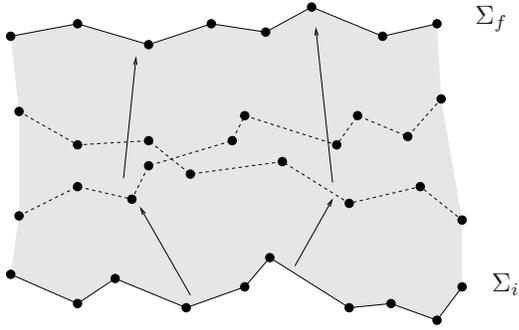}
%\centering
\hspace*{1.8cm} \caption{\small A discrete path independence, as also discussed in \cite{Dittrich:2013xwa}, requires that the discrete evolution from an initial spatial hypersurface $\Sigma_i$ to a final spatial hypersurface $\Sigma_f$ is independent of the choice of evolution moves and intermediate hypersurfaces through which one evolves. 3D Regge Calculus features this path independence, while 4D Regge Calculus does not. }\label{fig_pi}
%\end{center}
\end{SCfigure}

For 4D simplicial gravity models, on the other hand, the diffeomorphism symmetry of the continuum is generically broken for curved solutions \cite{Bahr:2009ku,Bahr:2011xs,Bonzom:2013ofa} and the Hamiltonian and diffeomorphism constraints \cite{Dittrich:2009fb,Dittrich:2011ke,Dittrich:2013jaa} do not in general arise as exact pre-- or post--constraints but rather as approximate or `pseudo'-constraints \cite{Dittrich:2009fb,Dittrich:2011ke,Gambini:2002wn,Gambini:2005vn}. The canonical dynamics is generated by global or local evolution moves \cite{Dittrich:2009fb,Dittrich:2011ke} and, given the generic absence of Hamiltonian and diffeomorphism constraints, not equivalent to a constraint generated dynamics---in contrast to the 3D case. As a consequence of the broken diffeomorphism symmetries, 4D Regge Calculus is a classically non-hyperbolic system such that different solutions arising from a given initial data set may be inequivalent in the sense that they can no longer be mapped into each other by symmetry transformations \cite{Dittrich:2011ke,Dittrich:2013jaa}. In particular, the 4D theory does not feature the path independence of the continuum or of 3D Regge Calculus; it does matter by means of which evolution moves and through which spatial triangulated hypersurfaces one evolves from an initial $\Sigma_i$ to a final $\Sigma_f$ (see figure \ref{fig_pi}). In contrast to the 3D theory, the non-hyperbolicity implies that restricting the dynamics to spatial triangulation preserving global evolution moves also entails a restriction of the solution space of 4D Regge Calculus \cite{Dittrich:2011ke,Dittrich:2013jaa}.\footnote{For a non-trivial dynamics the full set of canonical Pachner evolution moves in 4D is required \cite{Dittrich:2011ke,Dittrich:2013jaa,Hoehn:2014wwa,dh4}.} 

In the quantum theory the path integral for discrete systems, if constructed via (\ref{regpi}), will project onto solutions to all post--constraints. However, since for 4D Regge Calculus the set of pre-- and post--constraints does not in general include the Hamiltonian and diffeomorphism constraints \cite{Dittrich:2009fb,Dittrich:2011ke,Dittrich:2013jaa} and the latter arise rather as `pseudo'-constraints, the path integral can only be expected to act as an approximate projector onto solutions to the Hamiltonian and diffeomorphism constraints for large scales \cite{Dittrich:2008pw,Dittrich:2013xwa}. Consequently, since time evolution is no longer a symmetry, (pre-- and post--)physical states will generally evolve non-trivially under the time evolution moves---including refining ones. One can no longer expect a path independence in the quantum theory \cite{Dittrich:2013xwa}; evolving via distinct sets of evolution moves from a given $\Sigma_i$ to a $\Sigma_f$ will in general also generate distinct pre-- and post--physical states at $i$ and $f$. This is because there will exist many inequivalent bulk triangulations (in the sense that there exist no symmetry transformations translating among them) interpolating the fixed pair $\Sigma_i$ and $\Sigma_f$. This will also hold for a refining time evolution. Similarly, the quantum theory can be expected to be non-hyperbolic in the following sense: given an initial pre--physical state, inequivalent post--physical states may be generated from it by different sequences of evolution moves. In this case, one can no longer identify different physical states at different time steps as being one and the same physical state, but represented on different discretizations, as in the 3D theory.

Nevertheless, despite the broken symmetries, the discretization changing time evolution moves in 4D simplicial gravity should also not be viewed as a proper time evolution (in an external discrete time) but likewise as coarse graining, refining and entangling moves and the `forward' and `backward' direction must be superposed as in the 3D case \cite{Dittrich:2013xwa}. The difference to the 3D case is that, due to the absence of a path independence, %the discretization changing dynamics will generically change the physical states such that 
pre-- and post--physical states even related by refining evolution moves can in general no longer be identified with each other. Furthermore, since the 4D theory contains propagating degrees of freedom, non-trivial coarse graining pre-- or post--constraints of cases (b1) and (b2) in section \ref{sec_fullcomp} can occur, such that observables, physical states and Hilbert spaces will change non-unitarily under coarse graining moves.

In order to construct a consistent 4D quantum gravity theory from simplicial discretizations, one should rather aim at constructing a cylindrically consistent dynamics which also admits the definition of a continuum limit \cite{Dittrich:2012jq,Dittrich:2013xwa}. The notion of dynamical cylindrical consistency employs the (quantum) time evolution maps of refining time evolution moves as embedding maps of Hilbert spaces of coarser physical states into Hilbert spaces of finer physical states. (Note that within the present formalism such dynamical embedding maps are specified by the action.) The states on the coarser and finer discretization are identified as the same physical state but represented on different discretizations such that, in analogy to the continuum, physical states do not evolve in the `external' discrete time steps (for refining moves). %This requires the absence of non-trivial coarse graining constraints of cases (b1) and (b2) in the reverse direction from finer to coarser states. 
This yields a cylindrical equivalence class of states that can also be embedded in the continuum Hilbert space \cite{Dittrich:2012jq,Dittrich:2013xwa}. The cylindrical consistency condition for such dynamical embedding maps is equivalent to a path independence (from coarser to finer discretizations) \cite{Dittrich:2013xwa}. That is, any choice of a refining time evolution leading to the same discretization yields the same physical state such that it is meaningful to identify states. This, in turn, implies the implementation of a consistent (anomaly free) discrete Dirac hypersurface deformation algebra \cite{Bonzom:2013tna}. Dynamical cylindrical consistency thus implies a strong relation among diffeomorphism symmetry, discretization independence and path independence \cite{Bahr:2009ku,Bahr:2009qc,Bahr:2011uj,Dittrich:2012qb,Dittrich:2011vz,Dittrich:2013xwa}. Given the discussion above, this relation can be extended to include hyperbolicity in the above sense.

This notion of dynamical cylindrical consistency requires the discretization to be `perfect' in the sense of supporting the continuum dynamics \cite{Dittrich:2013xwa,Bahr:2009qc,Bahr:2009mc,Dittrich:2011vz}. It should be noted that cylindrical consistency of states always requires the direction from coarser to finer discretizations (or graphs) because this corresponds to the direction of the inductive limit of Hilbert spaces (on discretizations) which gives rise to the continuum Hilbert space. That is, a state on a coarser discretization can be identified, using a refining time evolution, with a state on a finer discretization. However, the converse is not true (with the exception of 3D Regge Calculus which is topological). Not every state on a finer discretization can be identified, using a coarse graining time evolution, with a state on a coarser discretization because the former will, in general, carry more dynamical information. For theories with local propagating degrees of freedom, going from a finer to a coarser discretization must in general lead to the non-trivial coarse graining pre--constraints of case (b1) at the `finer' time step. This, in particular, includes perfect discretizations with propagating degrees of freedom that encode the continuum dynamics. That is, while dynamical cylindrical consistency of perfect discretizations with propagating degrees of freedom implies that refining time evolution can be used to identify states on different discretizations, it does \underline{not} preclude the non-unitary projections of physical Hilbert spaces and Dirac pre-- and post--observables under non-trivial coarse graining moves.

%As mentioned before, these conditions are satisfied for perfect discretizations such as 3D Regge Calculus. Indeed, since for 3D Regge Calculus the pre--constraints are always automatically satisfied \cite{Dittrich:2011ke} (i.e.\ they are always of case (a) in section \ref{sec_fullcomp}), non-trivial coarse graning pre-- and post--constraints of cases (b1) and (b2) do not arise. Similarly, pre-- and post--constraints of cases (b1) and (b2) can, in general, not arise in perfect discretizations. Therefore, for perfect discretizations non-unitary projections of physical Hilbert spaces and Dirac observables cannot occur. Accordingly, in the absence of such non-trivial coarse graining constraints, the time evolution moves will always unitarily map between isomorphic pre-- and post--physical Hilbert spaces which, nevertheless, may correspond to different discretizations. In this case the move dependence of the pre-- and post--physical Hilbert spaces disappears such that it is meaningful to identify states at different time steps. In fact, in this case one would directly work with the continuum Hilbert space. Furthermore, since pre-- and post--observables cannot be projected out irreversibly for perfect discretizations, the (dynamical) cylindrically consistent pre-- and post--observables which survive discretization changing moves must ultimately coincide with the continuum Dirac observables which commute with the continuum Hamiltonian and diffeomorphism constraints.

A generic 4D simplicial space-time discretization, on the other hand, is not a perfect discretization and thus does not give rise to a cylindrically consistent theory, as discussed above. But, using coarse graining techniques, one can construct effective theories that feature an improved behaviour as regards symmetries and dynamics \cite{Bahr:2009ku,Bahr:2009qc,Bahr:2011uj,Dittrich:2012qb,Dittrich:2012jq}. At fixed points of the coarse graining flow these theories can be expected to possess a cylindrically consistent dynamics (see \cite{Dittrich:2013voa} for work in this direction). For interacting theories such improved discretizations generally involve very non-local couplings which render them analytically difficult. An alternative proposal to construct a cylindrically consistent continuum theory of quantum gravity has been put forward in \cite{Dittrich:2012jq,Dittrich:2013xwa}: instead of focusing on constructing a perfect discretization from an underlying simplicial discretization, one may work with the amplitude maps of the `general boundary formulation'---discussed in section \ref{sec_gbf}---and require that these constitute (dynamical) cylindrically consistent observables (which can thus be defined on the continuum Hilbert space).

 \section{Summary and conclusions}\label{sec_conc}

One of the most pressing issues in quantum gravity is to better understand and interpret the discretization (or graph) changing Hamiltonian dynamics appearing in various approaches \cite{Thiemann:1996ay,Thiemann:1996aw,Alesci:2010gb,Dittrich:2011ke,Dittrich:2013jaa,Dittrich:2013xwa} and, related to this, to understand the relation between covariant state sum models and canonical approaches \cite{Noui:2004iy,Dittrich:2009fb,Dittrich:2011ke,Alesci:2011ia,Thiemann:2013lka,Alexandrov:2011ab,Bonzom:2011tf,Alesci:2010gb}.
 
As a step in this direction, this manuscript offers a systematic quantum formalism for variational discrete systems with flat Euclidean configuration spaces $\cq\simeq\mathbb{R}^N$ which is applicable to both Euclidean and Lorentzian space-time signature. It employs the action to generate the canonical dynamics in terms of propagators and thereby directly links the covariant and canonical picture. The formalism encompasses both constrained and unconstrained global evolution moves and incorporates both constant and evolving Hilbert spaces. It thus applies to both discretization preserving and changing dynamics. Pre-- and post--physical states are constructed through projection of kinematical states with group averaging projectors. In order to construct the state sum from a composition of global evolution moves, we introduce the notion of kinematical propagators. In analogy to kinematical and physical states, physical propagators are constructed via the projection of kinematical propagators with group averaging projectors. Such a projection procedure can be viewed as a construction principle for the path integral of constrained quantum systems. This method automatically keeps track of divergences arising in the path integral; divergences can be easily regularized by dropping (or gauge fixing) doubly occurring projectors in the convolution of propagators. This also suggests a new perspective on the study of tracing and regularizing divergences in spin foam quantum gravity models \cite{Riello:2013bzw,Bonzom:2013ofa,Bonzom:2010ar}.

The various types of constraints in the quantum theory and their roles can be discussed. In summary:
 \begin{itemize}
 \item[(a)] Constraints $\hat{C}^n_a$ that are both pre-- and post--constraints are first class symmetry generators and responsible for genuine divergences in the composition of propagators to a state sum.
 \item[(b1)] Pre--constraints ${}^-\hat{C}^n_{b_-}$ which are independent of the post--constraints but first class are non-trivial coarse graining conditions on the post--physical states of the move $n-1\rightarrow n$. These constraints project out a subset of Dirac observables of the move $n-1\rightarrow n$, corresponding to `finer information', in the composition with $n\rightarrow n+1$. Such pre--constraints do {\it not} cause divergences in the path integral, but yield non-unitary projections of Hilbert spaces and `fine grained' Dirac observables.
 \item[(b2)] Post--constraints ${}^+\hat{C}^n_{b_+}$ which are independent of the pre--constraints but first class ensure that a pre--physical state at $n-1$ carrying `coarser information' can be evolved under $n-1\rightarrow n$ into a post--physical state at $n$ on a refined discretization. Such post--constraints are non-trivial coarse graining conditions for the pre--physical states at $n$ of the move $n\rightarrow n+1$. Non-trivial Dirac observables of the move $n\rightarrow n+1$, corresponding to `finer' information, are projected out in the composition with $n-1\rightarrow n$. These post--constraints do {\it not} lead to divergences in the path integral, but to non-unitary projections of physical Hilbert spaces and `fine grained' Dirac observables.
 \item[(c)] Second class pre-- and post--constraints \cite{Dittrich:2013jaa} are solved classically.
 \end{itemize}
 We emphasize that constraints of cases (b1) and (b2) only occur for systems with a temporally varying discretization and propagating degrees of freedom. %Perfect discretizations which reproduce the continuum dynamics and symmetries exactly \cite{Bahr:2009mc,Bahr:2009qc,Bahr:2011uj} do not admit constraints of cases (b1) and (b2).
 
As pointed out in section \ref{sec_effcon}, the composition of evolution moves generating a discretization changing dynamics can lead to a `propagation' of quantum constraints. For instance, in a composition of two moves $0\rightarrow1$ and $1\rightarrow 2$, pre--constraints at $1$ of case (b1) above induce new effective pre--constraints at step $0$ for the effective move $0\rightarrow2$. These effective pre--constraints ensure that the pre--physical states at $0$ that are mapped to post--states at $1$ only carry information up to a certain allowed refinement scale. In this case, they can be consistently propagated further to $2$. In other words, the pre--states at $0$ and the post--states at $1$ must correspond to a refinement of a coarser state. Only those Dirac observables at $0$ survive the composition of the moves which also commute with the effective pre--constraints at $0$. Any additional Dirac observables representing physical degrees of freedom carrying information below this refinement scale do not commute with the new pre--constraints and are projected out via the corresponding projectors. By the same reasoning, the post--constraints at $1$ of case (b2) above give rise to new effective post--constraints at step $2$. These too ensure that degrees of freedom of the move $1\rightarrow2$ below a certain refinement scale set by the move $0\rightarrow1$ are projected out.

For a temporally varying (imperfect) discretization, pre-- and post--constraints---and therefore the pre-- and post--physical Hilbert spaces as well as the quantum pre-- and post--observables---are thus evolution move dependent. This was also explicitly demonstrated in the toy model of section \ref{sec_ex}. %Similarly, the physical post-- and pre--Hilbert spaces and the quantum post-- and pre--observables at a fixed step $n$ depend on the evolution moves $i\rightarrow n$ and $n\rightarrow f$, respectively, for some choice of initial and final steps $i,f$. 
In particular, the pre-- and post--constraints of cases (b1) and (b2) above enforce a non-unitary projection of physical Hilbert spaces upon composing evolution moves. In this sense, the physical Hilbert spaces evolve in the course of evolution. Coarse graining thus leads to an irreversible loss of information in the dynamics because the non-unitary projections of Hilbert spaces cannot be undone. For further illustration and an explicit implementation of the present formalism for quadratic discrete actions, see also \cite{Hoehn:2014aoa}. 

%As argued in section \ref{sec_evol}, physical states are cylindrical functions on the configuration spaces corresponding to the surviving set of degrees of freedom above a certain coarse graining scale. In the same vein the surviving set of Dirac observables can be considered a cylindrically consistent set of observables.

As pointed out in section \ref{sec_gbf}, the present formalism can be viewed as a discrete version of the `general boundary formulation' of quantum theory \cite{Oeckl:2003vu,Oeckl:2005bv,Oeckl:2010ra,Oeckl:2011qd}. In this light, it is not surprising that different global evolution moves are generally associated to different constraints, physical Hilbert spaces and physical degrees of freedom. In a space-time context, different global evolution moves correspond to different triangulated space-time regions and these can be quite arbitrary. 

The situation in simplicial quantum gravity models is special for many reasons, as discussed in section \ref{sec_qg}. The time evolution moves should be viewed as generating a coarse graining, refining or entangling of the degrees of freedom of the discretization, rather than a proper time evolution of the physical states \cite{Dittrich:2013xwa}. Furthermore, the quantum time evolution moves in simplicial gravity contain a superposition of `forward' and `backward' orientation of the `time direction'. For 4D simplicial gravity models the situation is additionally complicated by the fact that the diffeomorphism symmetry of the continuum is generically broken \cite{Bahr:2009ku,Dittrich:2008pw,Bahr:2011xs,Bonzom:2013ofa} such that the path independence and hyperbolicity of the continuum dynamics is absent. In order to construct a cylindrically consistent 4D discrete dynamics which features such a path independence, diffeomorphism symmetry and admits the definition of a continuum limit \cite{Dittrich:2012jq,Dittrich:2013xwa}, a coarse graining procedure toward improved or perfect discretizations is necessary \cite{Bahr:2009ku,Bahr:2009qc,Bahr:2011uj,Dittrich:2012qb}.

Before the present quantum formalism can be directly applied to non-perturbative quantum gravity models, it firstly needs to be generalized to incorporate systems with arbitrary configuration manifolds---as its classical counterpart \cite{Dittrich:2013jaa,Dittrich:2011ke}. This will lead to global and topological non-trivialities in the quantization \cite{isham2}, but there is no obstruction in principle. However, it can be expected that the qualitative features of this formalism remain largely unchanged. Moreover, so far we have only considered the pure state case. In order to generalize the formalism to also include mixed states one may proceed along the lines of the positive formalism for the `general boundary formulation' \cite{Oeckl:2012ni} and adapt it to the discrete.

The global evolution moves of the present manuscript can always be decomposed into sequences of local evolution moves. In a space-time context, the latter do not evolve an entire hypersurface at once, but locally update the discretization of the hypersurface. For instance, in simplicial gravity models these local moves encompass the Pachner evolution moves \cite{Dittrich:2011ke,Thiemann:1996ay,Thiemann:1996aw,Bonzom:2011hm,Alesci:2010gb} which constitute the most general and basic time evolution moves. The quantization of such local evolution moves by means of the present formalism is the topic of the companion paper \cite{Hoehn:2014wwa}.

 \appendix
 
 \section{Proofs of the lemmas of section \ref{sec_lin}}\label{app}
 
We begin with the proof of Lemma \ref{lem1}
\begin{proof}
We employ the variable splitting introduced in the beginning of section \ref{sec_lin}. Using the standard position representation, it is straightforward to convince oneself that 
\ba
\hat{C}_I\,\psi(\lambda^I,x^\alpha)=e^{iS(\lambda^I,x^\alpha)/\hbar}\,\hat{p}_I\,e^{-iS(\lambda^I,x^\alpha)/\hbar}\,\psi(\lambda^I,x^\alpha).\nn
\ea
Hence,
\ba
\left(\hat{C}_I\right)^n=\left(e^{iS(\lambda^I,x^\alpha)/\hbar}\,\hat{p}_I\,e^{-iS(\lambda^I,x^\alpha)/\hbar}\right)^n=e^{iS(\lambda^I,x^\alpha)/\hbar}\,\left(\hat{p}_I\right)^n\,e^{-iS(\lambda^I,x^\alpha)/\hbar}\nn
\ea
and the (improper) projectors take the form (the spectra of the $\hat{C}_I$ are absolutely continuous)
\ba
\delta(\hat{C}_I)=\f{1}{2\pi\hbar}\int dt\,e^{is\hat{C}_I/\hbar}=\f{1}{2\pi\hbar}\int dt\,e^{iS/\hbar}\,e^{it\hat{p}_I/\hbar}\,e^{-iS/\hbar}.
\ea
Since the constraints are abelian, $[\hat{C}_I,\hat{C}_J]=0$, the different (improper) projectors $\delta(\hat{C}_I)$ commute and no factor ordering ambiguity in the definition of the physical states arises,
\ba
\psi^{\rm phys}(\lambda^I,x^\alpha)&=&\f{1}{(2\pi\hbar)^k}\int \prod_I\left(dt^I\,e^{iS/\hbar}e^{it^I\hat{p}_I/\hbar}e^{-iS/\hbar}\right)\psi^{\rm kin}(\lambda^I,x^\alpha)\nn\\
&=&\f{1}{(2\pi\hbar)^k}\int e^{iS/\hbar}\prod_I\left(dt^I\,e^{it^I\hat{p}_I/\hbar}\right)\phi^{\rm kin}(\lambda^I,x^\alpha)\nn\\
&=&\f{1}{(2\pi\hbar)^k}\,e^{iS(\lambda^I,x^\alpha)/\hbar}\int \prod_Idt^I\,\phi^{\rm kin}(\lambda^I+t^I,x^\alpha).\label{lemphys}
\ea
We have defined $\phi^{\rm kin} (\lambda^I,x^\alpha):=e^{-iS(\lambda^I,x^\alpha)/\hbar}\,\psi^{\rm kin}(\lambda^I,x^\alpha)$. Notice that $\psi^{\rm phys}(\lambda^I,x^\alpha)$ can only depend on the $\lambda^I$ through the factor $e^{iS(\lambda^I,x^\alpha)/\hbar}$. (Dependence on other time steps cancels out.)

Next, we note that $[\hat{G}_M(\lambda^I,x^\alpha),\hat{C}_I]=i\hbar\,\f{\p G_M}{\p \lambda^I}$, such that
\ba
\prod_{K=1}^k\delta(\hat{G}_K(\lambda^I,x^\alpha))=\f{\prod_{K=1}^k\delta(\lambda^K-c^K)}{\Big|\det\left([\hat{G}_M,\hat{C}_N]/\hbar\right)\Big|}\label{delta}
\ea
where $c^K$ is the value $\lambda^K$ must take after solving $G_K=0$, $K=1,\ldots,k$ (note that only a single solution $c^K$ per $\lambda^K$ exists, since $G_K$ are global gauge conditions). No factor ordering ambiguities arise because $[\hat{G}_M(\lambda^I,x^\alpha),\hat{C}_I]=f(\lambda^I,x^\alpha)$.

Finally, the conjunction of (\ref{lemphys}) and (\ref{delta}) yields the desired result
\ba
&&\hspace*{-1cm}(2\pi)^k\prod_I\delta(\hat{C}_I)\Big|\det\left([\hat{G}_M,\hat{C}_N]\right)\Big|\prod_K\delta(\hat{G}_K(\lambda^I,x^\alpha))\prod_J\delta(\hat{C}_J)\,\psi^{\rm kin}(\lambda^I,x^\alpha)\nn\\
&=&\f{1}{(2\pi\hbar)^k}\,e^{iS/\hbar}\int\prod_I\left(dt'^Idt^Ie^{it'^I\hat{p}_I}\right)\,e^{-iS/\hbar}\prod_K\delta(\lambda^K-c^K)\,e^{iS/\hbar}\,\phi^{\rm kin}(\lambda^I+t^I,x^\alpha)\nn\\
&=&\f{1}{(2\pi\hbar)^k}\,e^{iS/\hbar}\int\prod_Idt'^Idt^I\delta(\lambda^I+t'^I-c^I)\,\phi^{\rm kin}(\lambda^I+t^I+t'^I,x^\alpha)\nn\\
&=&\f{1}{(2\pi\hbar)^k}\,e^{iS(\lambda^I,x^\alpha)/\hbar}\int\prod_Idt^I\phi^{\rm kin}(c^I+t^I,x^\alpha)\underset{(\ref{lemphys})}{=}\psi^{\rm phys}(\lambda^I,x^\alpha).\nn
\ea\end{proof}
 
 Next, we prove lemma \ref{lem2}
 \begin{proof}
In the position representation, (\ref{pip}) reads
\ba
\left\langle\psi^{\rm phys}\Big|\phi^{\rm phys}\right\rangle_{\rm phys}\!\!\!\!\!\!\!\!\!\!\!\!\!\!\!&=&\!\!\!\!\!\!\!\int_{\cq}\prod_{I,\alpha}d\lambda^Idx^\alpha\left(\psi^{\rm kin}(\lambda^I,x^\alpha)\right)^*\prod_{I=1}^k\delta(\hat{C}_I)\,\phi^{\rm kin}(\lambda^I,x^\alpha)\nn\\
%&=&\!\!\!\!\!\!\!\int_{\cq}\prod_{I,\alpha}d\lambda^Idx^\alpha\left(\psi^{\rm kin}(\lambda^I,x^\alpha)\right)^*\phi^{\rm phys}(\lambda^I,x^\alpha)\nn\\
&\underset{\text{\tiny{Lemma \ref{lem1}}}}{=}&\!\!\!\!\!\!\!(2\pi)^k\int_{\cq}\prod_{I,\alpha}d\lambda^Idx^\alpha\left(\psi^{\rm kin}\right)^*\prod_{I=1}^k\delta(\hat{C}_I)\Big|\det\left([\hat{G}_M,\hat{C}_N]\right)\Big|\prod_{K=1}^k\delta(\hat{G}_K)\,\phi^{\rm phys}\nn\\
&=&\!\!\!\!\!\!\!(2\pi)^k\int_{\cq}\prod_{I,\alpha}d\lambda^Idx^\alpha\left(\prod_{I=1}^k\delta(\hat{C}_I)\psi^{\rm kin}\right)^*\Big|\det\left([\hat{G}_M,\hat{C}_N]\right)\Big|\prod_{K=1}^k\delta(\hat{G}_K)\,\phi^{\rm phys}\nn\\
&=&\!\!\!\!\!\!\!\!\!\!(2\pi)^k\!\!\int_{\cq}\prod_{I,\alpha}d\lambda^Idx^\alpha\Big|\!\det\left([\hat{G}_M,\hat{C}_N]\right)\!\!\Big|\prod_{K=1}^k\delta(\hat{G}_K)\!\!\left(\psi^{\rm phys}(\lambda^I,x^\alpha)\right)^*\!\!\phi^{\rm phys}(\lambda^I,x^\alpha).\nn
\ea
In the fourth line, we have made use of the fact that the $\hat{C}_I$ are self-adjoint with respect to the KIP on $\ch^{\rm kin}$ such that we may pull the (improper) projectors from one side to the other. The last step is possible because $[\hat{G}_K,\hat{C}_I]$ only depends on the configuration variables.

On account of the determinant of the Jacobian, the PIP does {\it not} depend on the particular choice of the gauge conditions $G_K(\lambda^I,x^\alpha)=0$. Furthermore, from (\ref{lemphys}) it follows that $\left(\psi^{\rm phys}(\lambda^I,x^\alpha)\right)^*\phi^{\rm phys}(\lambda^I,x^\alpha)$ is independent of the gauge parameter $\lambda^I$. Hence, the PIP is gauge-invariant. Finally, (\ref{lem2eq}) follows from integration of (\ref{delta}) over $\prod_Id\lambda^I$.
\end{proof}

We close with the proof of lemma \ref{lem3} 
 \begin{proof}
We begin with the last line in (\ref{conprop2}) and proceed analogously to the proof of lemma \ref{lem2},
\ba
{}^+\psi^{\rm phys}_1\!\!\!\!\!\!\!\!\!\!\!\!%&=&\int dx_0K_{0\rightarrow1}(x_1,x_0)\psi^{\rm kin}_0(x_0)=\int dx_0\left(\prod_{I=1}^k\delta({}^+\hat{C}^1_I)\prod^k_{J=1}\delta^*({}^-\hat{C}^0_J)\,\kappa_{0\rightarrow1}(x_0,x_1)\right)\psi^{\rm kin}_0(x_0)\nn\\
%&=&\int dx_0\left(\prod_{I=1}^k\delta({}^+\hat{C}^1_I)\,\kappa_{0\rightarrow1}(x_0,x_1)\right)\prod^k_{J=1}\delta({}^-\hat{C}^0_J)\,\psi^{\rm kin}_0(x_0)\nn\\
%&=&\int dx_0\left(\prod_{I=1}^k\delta({}^+\hat{C}^1_I)\,\kappa_{0\rightarrow1}(x_0,x_1)\right){}^-\psi^{\rm phys}_0(x_0)\nn\\
&=&\!\!\!\!\!\int_{\cq_0} \prod_{I,\alpha}d\mu_0^{I}\,dx_0^\alpha\, K_{0\rightarrow1}^{f_+}{}^-\psi^{\rm phys}_0\nn\\
&\underset{\text{\tiny Lemma \ref{lem1}}}{=}&\!\!\!\!\!(2\pi)^{k_-}\int_{\cq_0} \prod_{I,\alpha}d\mu_0^{I}\,dx_0^\alpha\, K_{0\rightarrow1}^{f_+}\prod_{I=1}^{k_-}\delta({}^-\hat{C}^0_I)\Big|\det\left([{}^-\hat{G}^0_M,{}^-\hat{C}^0_N]\right)\Big|\prod_{K=1}^{k_-}\delta({}^-\hat{G}^0_K)\,{}^-\psi^{\rm phys}_0\nn\\
&=&\!\!\!\!\!(2\pi)^{k_-}\int_{\cq_0} \prod_{I,\alpha}d\mu_0^{I}\,dx_0^\alpha \left(\prod_{I=1}^{k_-}\delta^*({}^-\hat{C}^0_I)K_{0\rightarrow1}^{f_+}\right)\Big|\!\det\left([{}^-\hat{G}^0_M,{}^-\hat{C}^0_N]\right)\!\Big|\prod_{K=1}^{k_-}\delta({}^-\hat{G}^0_K)\,{}^-\psi^{\rm phys}_0\nn\\
&\underset{(\ref{conprop})}{=}&\!\!\!\!\!(2\pi)^{k_-}\int_{\cq_0} \prod_{I,\alpha}d\mu_0^{I}\,dx_0^\alpha\,\Big|\det\left([{}^-\hat{G}^0_M,{}^-\hat{C}^0_N]\right)\Big|\prod_{K=1}^{k_-}\delta({}^-\hat{G}^0_K)\, K_{0\rightarrow1}\,{}^-\psi^{\rm phys}_0.\nn
\ea
Independence of the value of $\mu_0^I$ follows from recalling (\ref{cond5}) and otherwise for the same reason as $\lambda^I$-independence at the end of the proof of Lemma \ref{lem2}.
\end{proof}

\section*{Acknowledgements}
The author is grateful to Bianca Dittrich for numerous insightful discussions and comments on an earlier version of this manuscript. The author, furthermore, thanks Wojciech Kaminski and Rafael Sorkin for discussion and an anonymous referee for useful comments. Research at Perimeter Institute is supported by the Government of Canada through Industry Canada and by the Province of Ontario through the Ministry of Research and Innovation.

\bibliography{bibliography}{}
\bibliographystyle{utphys}

\end{document}